\newif\ifFullVersion
\FullVersiontrue 

\documentclass[11pt,draftclsnofoot, onecolumn]{IEEEtran}

\usepackage{times}
\usepackage{amsmath,dsfont}
\usepackage{amssymb,amsthm}
\usepackage{epsfig,verbatim}
\usepackage{subfigure}
\usepackage{setspace}
\usepackage{color}
\usepackage{cite}
\usepackage{epstopdf}
\usepackage{graphics}
\usepackage{accents}
\usepackage{acronym}
\usepackage{enumitem}
\usepackage[bookmarks,colorlinks]{hyperref}
\usepackage{mathtools}
\usepackage[utf8]{inputenc} 
\usepackage[T1]{fontenc}
\usepackage{url}
\usepackage{ifthen}

\usepackage{graphics}
\usepackage{tikz}
\usetikzlibrary{shapes.geometric, arrows}
\usepackage{soul}

\newtheorem{definition}{Definition}
\newtheorem{theorem}{Theorem}
\newtheorem{corollary}{Corollary}
\newtheorem{proposition}{Proposition}
\newtheorem{lemma}{Lemma}
\newtheorem{remark}{Remark}

\newtheorem{assumption}{Assumption}
\newcommand{\E}{\mathds{E}}

\newcommand{\myVec}[1]{{\boldsymbol{#1}}}

\newcommand{\mySet}[1]{\mathcal{#1}}
\newcommand{\Cmat}{\mathsf{C}}
\newcommand{\Chmat}{\mathsf{\hat{C}}}
\newcommand{\Cjmat}{\tilde{\Cmat}}

\newcommand{\myFtn}[1]{
	\ifFullVersion
	\footnote{#1}
	\fi	
}

\newcommand{\Xin}{{\myVec{X}}}

\newcommand{\Wi}{{\myVec{W}}}

\newcommand{\opt}{^{\rm opt}}
\newcommand{\RnD}{R_n(D)}
\newcommand{\zkn}[2]{\tilde{Z}_{#1,#2}}
\newcommand{\zkeps}[1]{Z_{#1}}
\newcommand{\zk}[1]{Z_{#1}}
\newcommand{\zetan}[1]{\zeta_{#1}}

\newcommand{\Fkeps}[1]{F_{#1}}
\newcommand{\Fkn}[2]{\tilde{F}_{#1,#2}}

\newcommand{\mN}{\mathcal{N}}
\newcommand{\mZ}{\mathcal{Z}}
\newcommand{\mK}{\mathcal{K}}

\newcommand{\mR}{\mathcal{R}}

\newcommand{\mS}{\mathcal{S}}

\newcommand{\mT}{\mathcal{T}}

\newcommand{\talpha}{\tilde{\alpha}}
\newcommand{\tbeta}{\tilde{\beta}}
\newcommand{\liminfk}{\mathop{\lim\inf}\limits_{k\rightarrow\infty}}

\newcommand{\Det}{{ \rm Det}}
\newcommand{\myEps}{\epsilon}
\newcommand{\eps}{\epsilon}
\newcommand{\myEpsn}{\epsilon_n}
\newcommand{\Wc}{W_{\rm c}}
\newcommand{\Weps}{W_{\myEps}}

\newcommand{\Seps}{S_{\myEps}}

\newcommand{\Sn}{S_{n}}

\newcommand{\VecSn}{\myVec{S}_{n}}
\newcommand{\VecSeps}{\myVec{S}_{\myEps}}

\newcommand{\Sheps}{\hat{S}_{\epsilon}}
\newcommand{\Shepsk}{\hat{\myVec{S}}_{\epsilon}^{(k)}}
\newcommand{\Shn}{\hat{S}_{n}}
\newcommand{\Shvec}{\hat{\myVec{S}}}
\newcommand{\shvec}{\hat{s}}
\newcommand{\Shnl}{\hat{S}_{n_l}}
\newcommand{\Snl}{S_{n_l}}

\newcommand{\Sc}{S_{\rm c}}
\newcommand{\Snk}{\VecSn^{(k)}}
\newcommand{\Wnk}{\myVec{W}_n^{(k)}}
\newcommand{\Shnk}{\hat{\myVec{S}}_{n}^{(k)}}
\newcommand{\Sepsk}{\VecSeps^{(k)}}
\newcommand{\Wepsk}{\myVec{W}_{\epsilon}^{(k)}}

\newcommand{\Charac}{\Phi}

\newcommand{\Wn}{W_{n}}

\newcommand{\Cwc}{\sigma^2_{\Wc}}

\newcommand{\rseps}{r_{\Seps}}

\newcommand{\rsn}{r_{\Sn}}
\newcommand{\Csc}{\sigma^2_{\Sc}}
\newcommand{\Sigseps}{\sigma^2_{\Seps}}
\newcommand{\Sigsn}{\sigma^2_{\Sn}}
\newcommand{\Sigshn}{\sigma^2_{\Shn}}

\newcommand{\Sigsheps}{\sigma^2_{\Sheps}}

\newcommand{\Tc}{T_{\rm ps}}
\newcommand{\Tsamp}{T_{\rm s}}
\newcommand{\Tsym}{T_{\rm sym}}
\newcommand{\Td}{p}

\newcommand{\cdf}[1]{F_{#1}}
\newcommand{\pdf}[1]{p_{#1}}

\newcommand{\fkn}{f_{k,n}}
\newcommand{\fkeps}{f_{k,\myEps}}

\newcommand{\Xnvec}[1][n]{{\myVec{X}}_{#1}}

\newcommand{\DC}{t_{\rm dc}}
\newcommand{\Trise}{t_{\rm rf}}

\newcommand{\ConvDist}[1]{\mathop{\longrightarrow}\limits^{(dist.)}_{#1}}
\newcommand{\Conv}[1]{\mathop{\longrightarrow}\limits_{#1}}

\newcommand{\plimsup}{{\rm p-}\mathop{\lim \sup}\limits_{k \rightarrow \infty}}
\newcommand*\mathinhead[2]{\texorpdfstring{${#1}$}{#2}}

\long\def\symbolfootnote[#1]#2{\begingroup\def\thefootnote{\fnsymbol{footnote}}\footnote[#1]{#2}\endgroup}

\acrodef{cdf}[CDF]{cumulative distribution function}
\acrodef{pdf}[PDF]{probability density function}
\acrodef{rv}[RV]{random variable}
\acrodef{dt}[DT]{discrete-time} 
\acrodef{ct}[CT]{continuous-time} 
\acrodef{sss}[SSS]{strict-sense stationary}
\acrodef{sscs}[SSCS]{strict-sense cyclostationary}
\acrodef{wss}[WSS]{wide-sense stationary}
\acrodef{wscs}[WSCS]{wide-sense cyclostationary}
\acrodef{wsacs}[WSACS]{wide-sense almost cyclostationary}
\acrodef{ofdm}[OFDM]{orthogonal frequency division multiplexing}
\acrodef{tdma}[TDMA]{time division multiple access}
\acrodef{cdma}[CDMA]{code division multiple access}
\acrodef{noma}[NOMA]{non-orthogonal multiple access}
\acrodef{asmcgn}[AS-MCGNC]{asynchronously-sampled memoryless cyclostationary Gaussian noise channel}
\acrodef{cmt}[CMT]{continuous mapping theorem}
\acrodef{cs}[CS]{cyclostationary}
\acrodef{dcd}[DCD]{decimated component decomposition}
\acrodef{iid}[IID]{independent and identically distributed}
\acrodef{rdf}[RDF]{rate-distortion function}
\acrodef{mimo}[MIMO]{multiple input multiple output}
\acrodef{siso}[SISO]{single-input single-output}
\acrodef{awgn}[AWGN]{additive white Gaussian noise}
\acrodef{csi}[CSI]{channel state information}
\acrodef{snr}[SNR]{signal to noise power ratio}
\acrodef{mac}[MAC]{multiple access channel}
\acrodef{bc}[BC]{broadcast channel}
\acrodef{su}[SU]{single-user}
\acrodef{mu}[MU]{multi-user}
\acrodef{plc}[PLC]{power line communications}
\acrodef{jscc}[JSCC]{joint source-channel coding}
\acrodef{acs}[ACS]{almost cyclostationary}
\acrodef{cf}[CF]{compress-and-forward}
\acrodef{psd}[PSD]{power spectral density}
\acrodef{cpsd}[CPSD]{cyclic power spectral density}
\acrodef{tpsd}[TPSD]{time-varying power spectral density}
\acrodef{ptp}[PTP]{point to point}
\acrodef{caf}[CAF]{cyclic autocorrelation function}
\acrodef{osi}[OSI]{open systems interconnection}
\acrodef{lptv}[LPTV]{linear periodically time-varying}
\acrodef{dtgc}[DTGC]{discrete time Gaussian channel}
\acrodef{isi}[ISI]{inter-symbol interference}
\acrodef{ncgc}[NCGC]{N-circular Gaussian channel}
\acrodef{ofdma}[OFDMA]{orthorgonal frequency division multiple access}
\acrodef{dsl}[DSL]{digital subscriber line}
\acrodef{dms}[DMS]{discrete memoryless source}
\acrodef{dmc}[DMC]{discrete memoryless channel}
\acrodef{aep}[AEP]{asymptotic equipartition property}
\acrodef{caf}[CAF]{cyclic autocorrelation function}
\acrodef{ic}[IC]{interference cancellation}
\acrodef{sic}[SIC]{successive interference cancellation}
\acrodef{ia}[IA]{interference alignment}
\acrodef{marc}[MARC]{multiple-access relay channels}
\acrodef{mabrc}[MABRC]{multiple-access broadcast relay channels}
\acrodef{pam}[PAM]{pulse amplitude modulation}
\acrodef{rhs}[RHS]{right hand side}
\acrodef{mse}[MSE]{mean squared error}

\newif\ifcomments
\definecolor{CmtColor}{rgb}{0,0.6,1}

\acrodef{cdf}[CDF]{cumulative distribution function}
\acrodef{pdf}[PDF]{probability density function}
\acrodef{rv}[RV]{random variable}
\acrodef{dt}[DT]{discrete-time} 
\acrodef{ct}[CT]{continuous-time} 
\acrodef{wscs}[WSCS]{wide-sense cyclostationary}
\acrodef{wsacs}[WSACS]{wide-sense almost cyclostationary}
\acrodef{ofdm}[OFDM]{orthogonal frequency division multiplexing}
\acrodef{tdma}[TDMA]{time division multiple access}
\acrodef{cdma}[CDMA]{code division multiple access}
\acrodef{noma}[NOMA]{non-orthogonal multiple access}
\acrodef{asmcgn}[AS-MCGNC]{asynchronously-sampled memoryless cyclostationary Gaussian noise channel}
\acrodef{cmt}[CMT]{continuous mapping theorem}

 \IEEEoverridecommandlockouts
\title{The Rate Distortion Function of Asynchronously Sampled Memoryless Cyclostationary Gaussian Processes
}

\vspace{-0.5cm}

\author{
	\IEEEauthorblockN{\vspace{-0.2cm} Emeka Abakasanga, Nir Shlezinger, Ron Dabora\\
	}

	\thanks{E. Abakasanga and R. Dabora  are with the department of ECE, Ben-Gurion University,  Israel (e-mail:  abakasan@post.bgu.ac.il; ron@ee.bgu.ac.il)
	N. Shlezinger is with the faculty of Math and CS, Weizmann Institute of Science,  Israel (e-mail: nirshlezinger1@gmail.com).  
		This work was supported by the Israel Science Foundation under Grants 1685/16 and 0100101, and by the Israeli Ministry of Economy through the HERON 5G consortium.}
	
	\vspace{-1.0cm}
	
}

\vspace{-0.75cm}
\begin{document}
	\maketitle
	\pagestyle{plain}
	\thispagestyle{plain}
	\vspace{-0.2cm}
\begin{abstract}
Man-made communications signals are typically modelled as continuous-time (CT) wide-sense cyclostationary (WSCS)  processes. As modern processing is digital, it operates on sampled versions of the CT  signals. When sampling is applied to a CT WSCS process, the statistics of the resulting discrete-time (DT)  process depends on the relationship between the sampling interval and the period of the statistics of the CT process: When these two parameters have a common integer factor, then the DT process is WSCS. This situation is referred to as synchronous sampling.  When this is not the case, which is referred to as asynchronous sampling, the resulting DT process is wide-sense almost cyclostationary (WSACS). Such acquired CT processes are commonly encoded using a source code to facilitate storage or transmission over multi-hop networks using compress-and-forward relaying.
In this work, we study the fundamental tradeoff of sources codes applied to sampled CT WSCS processes, namely, their rate-distortion function (RDF).  We note that while RDF characterization for the case of synchronous sampling directly follows from classic information-theoretic tools utilizing ergodicity and the law of large numbers, when sampling is asynchronous, the resulting process is not information stable. In such cases, commonly used information-theoretic tools are inapplicable to RDF analysis, which poses a major challenge. Using the information spectrum framework, we show that the RDF for asynchronous sampling in the low distortion regime can be expressed  as the limit superior of a sequence of RDFs in which each element corresponds to the RDF of a synchronously sampled WSCS process (but their limit is not guaranteed to exist). The resulting characterization allows us to introduce novel insights on the relationship between sampling synchronization and RDF. For example, we demonstrate that, differently from stationary processes,  small differences in the sampling rate and the sampling time offset can notably affect the RDF of sampled CT WSCS processes.
\end{abstract}


\section{Introduction}
	\label{sec:Intro}
Man-made signals are typically generated using a repetitive procedure, which takes place at fixed intervals. The resulting signals are thus commonly modeled as \ac{ct} random processes exhibiting periodic statistical properties \cite{gardner1987spectral,giannakis1998cyclostationary,gardner2006cyclostationarity}, which are referred to as {\em \ac{wscs}} processes.  In digital communications, where the transmitted waveforms commonly obey the \ac{wscs} model \cite{gardner2006cyclostationarity}, the received \ac{ct} signal is first sampled to obtain a \ac{dt} received signal. In the event that the sampling interval is commensurate with the period of the statistics of the \ac{ct} \ac{wscs} signal, cyclostationarity is preserved in \ac{dt} \cite[Sec. 3.9]{gardner2006cyclostationarity}. In this work, we refer to this situation as {\em synchronous sampling}. However, it is practically common to encounter scenarios in which the sampling rate at the receiver and symbol rate of the received \ac{ct} \ac{wscs} process are incommensurate, which is referred to as {\em asynchronous sampling}. The resulting sampled process in such cases is a \ac{dt} {\em \ac{wsacs}} stochastic process \cite[Sec. 3.9]{gardner2006cyclostationarity}. 

This research aims at investigating lossy source coding for asynchronously sampled \ac{ct} \ac{wscs} processes. 
In the source coding problem, every sequence of information symbols from the source is mapped into a sequence of code symbols, referred to as codewords, taken from a predefined codebook. 
In {\em lossy} source coding, the source sequence is recovered up to a predefined distortion constraint, within an arbitrary small tolerance of error. The figure-of-merit for lossy source coding is the \ac{rdf} which characterizes the minimum number of bits per symbol required to compress the source sequence such that it can be reconstructed at the  decoder within the specified maximal distortion \cite{berger1998lossy}. For an \ac{iid} random source process, the \ac{rdf} can be expressed as the minimum mutual information between the source variable and the reconstruction variable, such that for the corresponding conditional distribution of the reconstruction symbol given the source symbol, the distortion constraint is satisfied \cite[Ch. 10]{cover2006elements}. The source coding problem has been further studied in multiple different scenarios, including the reconstruction of a single source at multiple destinations \cite{wolf1980source} and the reconstruction of  multiple correlated stationary Gaussian sources at a single destination~\cite{wyner1976rate,oohama1997gaussian,pandya2004lossy}.

For {\em stationary} source processes, ergodicity theory and the \ac{aep} \cite[Ch. 3]{cover2006elements} were applied for characterizing the \ac{rdf} for different scenarios \cite[Ch. 9]{gallager1968information}, \cite[Sec. I]{berger1998lossy}, \cite{harrison2008generalized}. However, as in a broad range of applications, including  digital communication networks, most \ac{ct} signals  are \ac{wscs}, the sampling operation results in a \ac{dt} source signal whose statistics depends on the relationship between the sampling rate and the period of the statistics of the source signal.
When sampling is synchronous, the resulting \ac{dt} source signal is \ac{wscs} \cite[Sec. 3.9]{gardner2006cyclostationarity}. The \ac{rdf} for lossy compression of \ac{dt} \ac{wscs} Gaussian sources with memory was studied in \cite{kipnis2018distortion}. This used the fact that any \ac{wscs} signal can be transformed into a set of stationary subprocess \cite{giannakis1998cyclostationary}; thereby facilitating the application of information-theoretic results obtained for multivariate stationary sources to the derivation of the \ac{rdf}; Nonetheless, in many digital communications scenarios, the sampling rate and the symbol rate of the \ac{ct} \ac{wscs} process are not related in any way, and are possibly incommensurate, resulting in a sampled process which is a \ac{dt} \ac{wsacs} stochastic process \cite[Sec. 3.9]{gardner2006cyclostationarity}. Such situations can occur as a result of the a-priori determined values of the sampling interval and the symbol duration of the \ac{wscs} source signal, as well as due to sampling clock jitters resulting from hardware impairments. A comprehensive review of trends and applications for almost cyclostationary signals can be found in \cite{napolitano2016cyclostationarity}. Despite their apparent frequent occurrences, the \ac{rdf} for lossy compression of \ac{wsacs} sources was not characterized, which is the motivation for the current research. A major challenge associated with characterizing fundamental limits for asynchronously sampled \ac{wscs} processes stems from the fact that the resulting processes are {\em not information stable}, in the sense that its conditional distribution is not ergodic \cite[Page X]{han2003information}, \cite{verdu1994general}, \cite{zeng2006information}. As a result, the standard information-theoretic tools cannot be employed, making the characterization of the \ac{rdf} a very challenging problem. 

Our recent study in \cite{shlezinger2019capacity} on channel coding reveals that for the case of additive \ac{ct} \ac{wscs} Gaussian noise, capacity varies significantly with sampling rates, whether the Nyquist criterion is satisfied or not. In particular, it was observed that the capacity can change dramatically with minor variations in the sampling rate, causing it to switch from synchronous sampling to asynchronous sampling. This is in direct contrast to the results obtained for wide-sense stationary noise for which the capacity remains unchanged for any sampling rate above the Nyquist rate \cite{shannon1998communication}. A natural fundamental question that arises from this result is how the \ac{rdf} of a sampled Gaussian source process varies with the sampling rate. As a motivating example, one may consider \ac{cf} relaying, where the relay samples at a rate which can be incommensurate with the symbol rate of the incoming communications signal. 

In this work, we employ the information spectrum framework \cite{han2003information} in characterizing the \ac{rdf} of asynchronously sampled memoryless Gaussian \ac{wscs} processes, as this framework is applicable to the information-theoretic analysis of  non information-stable processes \cite[Page VII]{han2003information}. We further note that while rate characterizations obtained using information spectrum tools and its associated quantities may be difficult to evaluate \cite[Remark 1.7.3]{han2003information}, here we obtain a numerically computable characterization of the \ac{rdf}. 
In particular, we focus on the \ac{mse} distortion measure in the low distortion regime, namely, source codes for which the average \ac{mse} of the difference between the source and the reproduction process is not larger than the minimal source variance. 
The results of this research lead to accurate modelling of signal compression in current and future digital communications systems. Furthermore, we utilize our characterization of the \ac{rdf}  \ac{rdf} for a sampled \ac{ct} \ac{wscs} Gaussian source with different sampling rates and sampling time offsets. We demonstrate that, differently from stationary signals, when applying a lossy source code a sampled \ac{wscs} process, the achievable rate-distortion tradeoff can be significantly affected by minor variations in the sampling time offset and the sampling rate. Our results thus allow identifying the sampling rate and sampling time offsets which minimize the \ac{rdf} in systems involving asynchronously sampled \ac{wscs} processes. 


The rest of this work is organised as follows: Section \ref{sec:backgrnd_literature} provides a scientific background on cyclostationary processes, and on rate-distortion analysis of \ac{dt} \ac{wscs} Gaussian sources. Section \ref{sec:rsrch_obj} presents the problem formulation and auxiliary results, and Section \ref{subsec:rdf_Async_samp_wscs} details the main result of \ac{rdf} characterization for sampled \ac{wscs} Gaussian process. Numerical examples and discussions are addressed in Section \ref{sec:Simulations}, and Section \ref{conclusion} concludes the paper.

\section{Preliminaries and Background}
\label{sec:backgrnd_literature}
In the following we review the main tools and framework used in this work: In Subsection \ref{notations}, we detail the notations. In Subsection \ref{subsec:cyclosta_prelim} we review the basics of cyclostationary processes and the statistical properties of a \ac{dt} process resulting from sampling a \ac{ct} \ac{wscs} process. In Subsection \ref{subsec:RDF_stat}, we recall some preliminaries in rate-distortion theory, and present the \ac{rdf} for a \ac{dt} \ac{wscs} Gaussian source process. This background creates a premise for the statement of the main result provided in Section \ref{subsec:rdf_Async_samp_wscs} of this paper.

\subsection{Notations}
\label{notations}
 In this paper, random vectors are denoted by boldface uppercase letters, e.g., $\myVec{X}$; boldface lowercase letters denote deterministic column vectors, e.g., $\myVec{x}$. Scalar RVs and deterministic values are denoted via standard uppercase and lowercase fonts respectively, e.g., $X$ and $x$. Scalar random processes are denoted with $X(t), t\in \mathcal{R}$ for \ac{ct} and with $X[n], n\in \mathcal{Z}$ for \ac{dt}. Uppercase Sans-Serif fonts represent matrices, e.g., $\mathsf{A}$, and the element at the $i^{th}$ row and the $l^{th}$ column of $\mathsf{A}$ is denoted with $(\mathsf{A})_{i,l}$. We use $|\cdot|$ to denote the absolute value, $\lfloor d \rfloor, d\in \mR$, to denote the floor function and $d^+, d \in \mR$, to denote the $\mathop{\max}\{0,d\}$. $\delta [\cdot]$ denotes the Kronecker delta function: $\delta [n]=1$  for $n=0$ and  $\delta [n]=0$ otherwise, and $\mathbb{E}\{\cdot\}$ denotes the stochastic expectation. The sets of positive integers, integers, rational numbers, real numbers, positive numbers, and complex numbers are denoted by $\mathcal{N},\mathcal{Z}$, $\mySet{Q}$, $\mathcal{R}$, $\mR^{++}$, and $\mySet{C}$, respectively. The \ac{cdf} is denoted by $F_X(x)\triangleq\Pr{(X\leq x)}$ and the \ac{pdf} of a \ac{ct} \ac{rv} is denoted by $p_X (x)$. We represent a real Gaussian distribution with mean $\mu$ and variance $\sigma^2$ by the notation $\mN(\mu, \sigma^2)$. All logarithms are taken to base-2, and $j=\sqrt{-1}$. Lastly, for any sequence $y[i]$, $i \in \mySet{N}$, and positive integer $k \in \mathcal{N}$,  ${\myVec{y}}^{(k)}$ denotes the column vector $\big( { y}[1],\ldots, { y}[k]\big)^T$. 
 
\subsection{Wide-Sense Cyclostationary Random Processes}
\label{subsec:cyclosta_prelim}
Here, we review some preliminaries in the theory of cyclostationarity. 
We begin by recalling the definition of wide-sense cyclostationary processes:

\begin{definition}[Wide-sense cyclostationary processes {\cite[Sec. 17.2]{giannakis1998cyclostationary}}]
\label{Def:wscs}
A scalar stochastic process $\{S(t)\}_{t\in \mathcal{T}}$, where $\mathcal{T}$ is either discrete ($\mathcal{T} = \mZ$)  or continuous ($\mathcal{T} = \mR$)  is called \ac{wscs} if both its first-order and its second-order moments are periodic with respect to $t \in \mathcal{T}$ with some period $N_p\in \mT$.
\end{definition}

\ac{wscs} signal are thus random processes whose first and second-order moments are periodic functions. To define \ac{wsacs} signals, we first recall the definition of almost-periodic functions:


\begin{definition}[Almost-periodic-function \cite{GUAN20131165}] 
A function $x(t), t\in \mT$ where $\mT$ is either discrete ($\mT =\mZ$) or continuous ($\mT =\mR$), is called an almost-periodic function if for every $\epsilon > 0$ there exists a number $l(\epsilon) > 0$ with the property that any interval in $\mT$ of length $l(\epsilon)$ contains a $\tau$, 
such that

\begin{equation*}
|x(t+\tau)-x(t)|<\epsilon, \quad \forall t\in \mT.
\end{equation*}

\end{definition}

\begin{definition}[Wide-sense almost-cyclostationary processes {\cite[Def. 17.2]{giannakis1998cyclostationary}}] 
 A scalar stochastic process $\left\{S(t)\right\}_{t\in \mathcal{T}}$ where $\mathcal{T}$ is either discrete ($\mT =\mZ$) or continuous ($\mT =\mR$), is called \ac{wsacs} if its first and its second order moments are almost-periodic functions with respect to $t\in \mathcal{T}$. 
\end{definition}

The \ac{dt} \ac{wscs} model is commonly used in the communications literature, as it facilitates the the analysis of many problems of interest, such as fundamental rate limits analysis \cite{shlezinger2015capacity,shlezinger2016capacity,shlezinger2017secrecy}, channel identification \cite{heath1999exploiting}, synchronization \cite{shaked2017joint}, and noise mitigation \cite{shlezinger2014frequency}. However, in many scenarios, the considered signals are \ac{wsacs} rather than \ac{wscs}.
To see how the \ac{wsacs} model is obtained in the context of sampled signals, we briefly recall the discussion in \cite{shlezinger2019capacity} on sampled \ac{wscs}
 processes (please refer to \cite[Sec. II.B]{shlezinger2019capacity} for more details):
  Consider a \ac{ct} \ac{wscs} random process $S(t)$, which is sampled uniformly with a sampling interval of $\Tsamp$ and sampling time offset $\phi$, resulting in a \ac{dt} random process $S[i]=S(i\cdot \Tsamp+\phi )$. It is well known that contrary to stationary processes, which have a time-invariant statistical characteristics, the values of $\Tsamp$ and $\phi$ have a significant effect on the statistics of sampled \ac{wscs} processes \cite[Sec. II.B]{shlezinger2019capacity}. To demonstrate this point, consider a \ac{ct} \ac{wscs} process with variance $\sigma_s^2(t)=\frac{1}{2}\cdot \sin\left(2\pi t/\Tsym\right)+2$ for some $\Tsym > 0$.
      The sampled process for $\phi=0$ (no symbol time offset) and $\Tsamp=\frac{\Tsym}{3}$ has a  variance function whose period is $N_p=3$: $\sigma_s^2(i\Tsamp)=\{2,2.433,1.567,2,2.433,1.567,\ldots\}$, for $i=0,1,2,3,4,5,\ldots$; while the DT process obtained with the same sampling interval and the  sampling time offset of $\phi= \frac{\Tsamp}{2\pi}$ has a periodic variance with $N_p=3$ with values $\sigma_s^2(i\Tsamp +\phi) =\{2.155,2.335,1.510,2.155,2.335,1.510,\ldots\}$, for $i=0,1,2,3,4,5,\ldots$, which are different from the values of the \ac{dt} variance for $\phi=0$. It follows that both variances are periodic in discrete-time with the same period $N_p=3$, although with different values within the period, which is a result of the sampling time offset, yet, both \ac{dt} processes correspond to two instances of {\em synchronous} sampling. Lastly, consider the sampled variance obtained by sampling without a time offset (i.e., $\phi=0$) at a sampling interval of $\Tsamp=(1+\frac{1}{2\pi})\frac{\Tsym}{3}$. For this case, $\Tsamp$ is not an integer divisor of $\Tsym$ or of any of its integer multiples (i.e.,  $\frac{\Tsym}{\Tsamp}=2+\frac{2\pi-2}{2\pi+1}\equiv 2+\eps$; where $\eps \not \in \mySet{Q}$ and $\eps \in [0,1)$ ) resulting in the variance values $\sigma_s^2(i\Tsamp)=\{2, 2.335, 1.5027, 2.405,1.896,1.75,\ldots\}$, for $i=0,1,2,3,4,5\ldots$. For this scenario, the \ac{dt} variance is not periodic but is almost-periodic, corresponding to {\em asynchronous} sampling and the resulting \ac{dt} process is not \ac{wscs} but \ac{wsacs} \cite[Sec. 3.2]{gardner2006cyclostationarity}. The example above demonstrates that the statistical properties of sampled \ac{wscs} signals depend on the sampling rate and the sampling time offset, implying that the \ac{rdf} of such processes should also depend on these quantities, as we demonstrate in the sequel.

\subsection{The Rate-Distortion Function for DT WSCS Processes}
\label{subsec:RDF_stat}
\tikzstyle{codec}=[draw=black, rectangle, fill=white, minimum width=3cm, minimum height=1cm, text centered]

\tikzstyle{arrow} = [thick,->,>=stealth]

\begin{figure}
\centering
\begin{tikzpicture}[node distance=4.5cm]
\node (encoder)[codec]{Encoder $f_S$};
\node (decoder)[codec, right of=encoder, xshift=2cm]{Decoder $g_S$};

\draw [arrow] (encoder) -- node[anchor=south] {$\{1,2,\ldots,2^{lR}\}$} (decoder);
\draw [arrow] (-3.7cm,0) -- node[anchor=south] {$\{S[i]\}_{i=1}^l$} (encoder);
\draw [arrow] (decoder) -- node[anchor=south] {$\{\hat{S}[i]\}_{i=1}^l$} (10cm,0);
\end{tikzpicture}
    	\caption{Source coding block diagram}
		
		\label{fig:source_code_dgrm}
\end{figure}
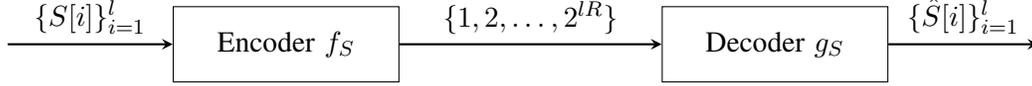

In this subsection we review the source coding problem and the existing results on the \ac{rdf} of \ac{wscs} processes.
We begin by recalling the definition of a source coding scheme, see, e.g., \cite[ch. 3]{el2011network}, \cite[Ch.10]{cover2006elements}: 
\begin{definition}[Source coding scheme]
\label{Def:source_coding_scheme}
A source coding scheme with  blocklength $l$ consists of:
\begin{enumerate}
    \item An encoder $f_S$ which maps a block of $l$ source samples $\{S[i]\}^{l}_{i=1}$ into an index from a set of $M = 2^{lR}$ indexes,
$f_S:\{S[i]\}_{i=1}^l \mapsto \{1,2,\ldots, M\}$.
\item A decoder $g_S$ which maps the received index into a reconstructed sequence of length $l$, $\left\{\hat S [i]\right\}_{i=1}^l$, $g_S:\{1,2,\ldots, M\}\mapsto \left\{\hat S [i]\right\}_{i=1}^l$
\end{enumerate}
The encoder-decoder pair is referred to as an $(R,l)$ source code, where $R$ is the rate of the code in bits per source sample, defined  as:
\begin{equation}
\label{eqn:Generic_rate_process}
    R=\frac{1}{l} \log_2 M  
\end{equation}
\end{definition} 
 The \ac{rdf} characterizes the minimal average number of bits per source sample, denoted $R(D)$, that can be used to encode a source process such that it can be reconstructed from its encoded representation with a recovery distortion not larger than $D >0$ \cite[Sec. 10.2]{cover2006elements}. In the current work, we use the \ac{mse} distortion measure, which measures the cost of decoding a source symbol $S$ into $\hat{S}$ via $d(S,\hat{S})=\left\|S-\hat{S}\right\|^2$. The distortion for a sequence of source samples $\myVec{S}^{(l)}$ decoded into a reproduction sequence $\hat{\myVec{S}}^{(l)}$ is given by $d\left(\myVec{S}^{(l)},\hat{\myVec{S}}^{(l)}\right)= \frac{1}{l} \sum\limits_{i=1}^{l} \left(S[i]-\hat S[i]\right)^2$ and the average distortion in decoding a random source sequence $\myVec{S}^{(l)}$ into a random reproduction sequence $\hat{\myVec{S}}^{(l)}$ is defined as:
\begin{equation}
\label{eqn:distortionn}
	  \bar{d}\left(\myVec{S}^{(l)},\hat{\myVec{S}}^{(l)}\right)\triangleq\E\left\{d\left(\myVec{S}^{(l)},\hat{\myVec{S}}^{(l)}\right)\right\}
	  =\frac{1}{l} \sum\limits_{i=1}^{l} \E\left\{\left(S[i]-\hat S[i]\right)^2 \right\},
\end{equation}
where the expectation in \eqref{eqn:distortionn} is taken with respect to the joint probability distributions on the source $S[i]$ and its reproduction $\hat{S}[i]$. 
Using Def. \ref{Def:source_coding_scheme} we can now formulate the achievable rate-distortion pair for a source $S[i]$, as stated in the following definition  \cite[Pg. 471]{gallager1968information}:

\begin{definition}[Achievable rate-distortion pair] 
\label{def:rate_dist_pair}
A rate-distortion pair $(R, D)$ is achievable for a process $\{S[i]\}_{i\in\mN}$ if for any $\eta>0$ and for all sufficiently large $l$ one can construct an  $\left(R_s, l\right)$ source code such that
\textcolor{black}{
\begin{equation}
\label{eqn:rate_bound}
 R_s\leq R+\eta.
\end{equation}}
and
\begin{equation}
\label{eqn:dist_bound}
    \bar{d}\left(\myVec{S}^{(l)},\hat{\myVec{S}}^{(l)}\right)\leq D+\eta.
\end{equation}
\end{definition}
\begin{definition}
The rate-distortion function $R(D)$ is defined as the infimum of all achievable rates $R$ for a given maximum allowed distortion $D$. 
\end{definition}

Def. \ref{def:rate_dist_pair} defines a rate-distortion pair to as that achievable using source codes with any sufficiently large blocklength. In the following lemma, which is required to characterize the \ac{rdf} of \ac{dt} \ac{wscs} signals, we state that it is sufficient to consider only source codes whose blocklength is an integer multiple of some fixed integer:
\begin{lemma}
\label{lemma:arbit_blocklength}
Consider the process$\{S[i]\}_{i\in\mN}$ with a finite and bounded variance. For a given maximum allowed distortion $D$, 
the optimal reproduction process $\{\hat{S}[i]\}_{i\in\mN}$ is also the optimal reproduction process when restricted to using source codes whose blocklengths are  integer multiples of some fixed positive integer $r$.
\end{lemma}
\begin{proof}
The proof of the lemma is detailed in Appendix \ref{app:arb_blocklength}.
\end{proof}
This lemma facilitates switching between multivariate and scalar representations of the source and the reproduction processes.

The \ac{rdf} obviously depends on the distribution of the source $\{S[i]\}_{i\in\mN}$. Thus, modifying the source yields a different \ac{rdf}. However, when a source is scaled by some positive constant, the \ac{rdf} of the scaled process with the \ac{mse} criterion can be inferred from that of the original process, as stated in the following theorem: 
\begin{theorem}
\label{thm:ThmScale1}
 \textcolor{black}{
Let $\{S[i]\}_{i\in\mN}$ be a source process for which the rate-distortion pair $(R,D)$ is achievable under the \ac{mse} distortion. 
%
Then, for every $\alpha \in \mR^{++}$, it holds that  the rate-distortion pair $(R, \alpha^2\cdot D)$ is achievable for the source $\{\alpha \cdot S[i]\}_{i\in\mN}$. }
\end{theorem}

\begin{proof}
The proof to the theorem is detailed in Appendix \ref{app:Proofthmm}.
\end{proof}
 
Lastly, in the proof of our main result, we make use of the \ac{rdf} for \ac{dt} \ac{wscs} sources derived in \cite[Thm. 1]{kipnis2018distortion}, repeated below for ease of reference. Prior to the statement of the theorem, we recall that for blocklenghts which are integer multiples of $N_p$,  a \ac{wscs} process $S[i]$ with period $N_p >0$ can be represented as an equivalent $N_p$-dimensional process $\myVec{S}^{(N_p)}[i]$ via the \ac{dcd} \cite[Sec. 17.2]{giannakis1998cyclostationary}. The \ac{psd} of the process $\myVec{S}^{(N_p)}$ is defined as \cite[Sec. II]{kipnis2018distortion}:
\begin{equation}
  \bigg(\myVec{\rho}_{\myVec{S}} \left(e^{j2 \pi f}\right)\bigg)_{u,v}  =\sum\limits_{\Delta \in \mathcal{Z}} \bigg(\mathsf{R}_{\myVec{S}}[\Delta]\bigg)_{u,v}e^{-j2 \pi f\Delta} \qquad -\frac{1}{2}\leq f\leq \frac{1}{2},\quad u,v \in \{1,2,\ldots N_p\}
\end{equation}
where $\mathsf{R}_{\myVec{S}}[\Delta]\triangleq \E\left\{\myVec{S}^{(N_p)}[i]\cdot\myVec{S}^{(N_p)}[i+\Delta]\right\}$ \cite[Sec. 17.2]{giannakis1998cyclostationary}.
We now proceed to the statement of \cite[Thm. 1]{kipnis2018distortion}:
\begin{theorem}
\cite[Thm. 1]{kipnis2018distortion} Consider a zero-mean real \ac{dt} \ac{wscs} Gaussian source $S[i], i\in \mN$ with memory, and let $N_p \in \mN$ denote the period of its statistics. The \ac{rdf} is expressed as: 

\begin{subequations}
\label{eqn:kipnis_RDF_WSCS}
\begin{equation}
\label{eqn:kipnis_Rate_WSCS}
    R(D)=\frac{1}{2N_p} \sum\limits_{m=1}^{N_p} \int_{f=-0.5}^{0.5} \left(\log\left(\frac{\lambda_m\left(e^{j2 \pi f}\right)}{\theta}\right)\right)^+ \mathrm{d}f,
\end{equation}
where $\lambda_m\left(e^{j2 \pi f}\right)$, $m=1,2,\ldots, N$ denote the eigenvalues of the \ac{psd} matrix of the process $\myVec{S}^{(N_p)}[i]$, which is  obtained from $S[i]$ by applying $N_p$-dimensional \ac{dcd}, and  $\theta$ is selected such that
\begin{equation}
\label{eqn:kipnis_Distortion_WSCS}
    D=\frac{1}{N_p} \sum\limits_{m=1}^{N_p} \int_{f=  -0.5}^{0.5} \min \left\{\lambda_m\left(e^{j2 \pi f}\right),\theta \right\}\mathrm{d}f.
\end{equation}
\end{subequations}

\end{theorem}
We note that $\myVec{S}^{(N_p)}[i]$ corresponds to a vector of stationary  processes whose elements are are not identically distributed; hence the variance is different for each element. 
Using \cite[Thm. 1]{kipnis2018distortion}, we can directly obtain the \ac{rdf} for the special case of a \ac{dt} memoryless \ac{wscs} Gaussian process. This is stated in the following corollary:

\begin{corollary}
\label{corollary:RDF_memoryless_WSCS}
Let $\{S[i]\}_{i \in \mN}$ be a zero-mean \ac{dt} memoryless real \ac{wscs} Gaussian source with period $N_p\in \mathcal{N}$, and set $\sigma^2_m = \E\{S^2[m]\}$ for $ m=1,2,\ldots, N_P$. 
%
\begin{subequations}
\label{eqn:RDF_WSCS}
The \ac{rdf} for compression of of $S[i]$ is stated as:
\begin{align}
\label{eqn:Rate_WSCS}
   R(D)= \begin{cases}
    \frac{1}{2N_p} \sum\limits_{m=1}^{N_p}\log \left(\frac{\sigma^2_m}{D_m}\right) & D\leq \frac{1}{N_p} \sum\limits_{m=1}^{N_p}\sigma^2_m\\
    0 & D>  \frac{1}{N_p} \sum\limits_{m=1}^{N_p}\sigma^2_m,
    \end{cases}
\end{align}
where $D_m\triangleq \min \left\{\sigma^2_m,\theta \right\}$, and $\theta$ is defined such that \begin{equation}
\label{eqn:Distortion_WSCS}
    D=\frac{1}{N_p} \sum\limits_{m=1}^{N_p} D_m.
\end{equation}
\end{subequations}
\end{corollary}

\begin{proof} 
Applying Equations \eqref{eqn:kipnis_Rate_WSCS} and \eqref{eqn:kipnis_Distortion_WSCS} to our specific case of a memoryless \ac{wscs} source, we obtain equations \eqref{eqn:Rate_WSCS} and \eqref{eqn:Distortion_WSCS} as follows: First note that the corresponding \ac{dcd} components for a zero-mean memoryless \ac{wscs} process are also zero-mean and memoryless; hence the \ac{psd} matrix for the multivariate process $\myVec{S}^{(N_p)}[i]$ is a diagonal matrix, whose eigenvalues are the constant diagonal elements such that the $m^{th}$ diagonal element is equal to the variance $\sigma_m^2$: $\lambda_m\left(e^{j2 \pi f}\right)= \sigma_m^2$. Now, writing Eqn. \eqref{eqn:kipnis_Rate_WSCS} for this case we obtain:
\begin{align}
    R(D)&=\frac{1}{2N_p} \sum\limits_{m=1}^{N_p} \int_{f=-0.5}^{0.5} \left(\log\left(\frac{\lambda_m\left(e^{j2 \pi f}\right)}{\theta}\right)\right)^+ \mathrm{d}f\nonumber\\
     &=\frac{1}{2N_p} \sum\limits_{m=1}^{N_p}\left(\log \left(\frac{\sigma^2_m}{\theta}\right)\right)^+.
     \label{eqn:Aid1}
\end{align}
Since $\left(\log \left(\frac{\sigma^2_m}{\theta}\right)\right)^+=\max\left\{0, \log \left(\frac{\sigma^2_m}{\theta}\right)\right\}\equiv\log\left( \frac{\sigma^2_m}{D_m}\right)$ it follows that \eqref{eqn:Aid1} coincides with  \eqref{eqn:Rate_WSCS}. Next, expressing Eqn. \eqref{eqn:kipnis_Distortion_WSCS} for the memoryless source process, we obtain:
\begin{equation}
    D=\frac{1}{N_p} \sum\limits_{m=1}^{N_p} \int_{f=  -0.5}^{0.5} \min \left\{\lambda_m\left(e^{j2 \pi f}\right),\theta \right\}\mathrm{d}f=\frac{1}{N_p} \sum\limits_{m=1}^{N_p} \min \left\{\sigma^2_m, \theta \right\},
\end{equation}
proving Eqn. \eqref{eqn:Distortion_WSCS}. 
\end{proof}

 Now, from Lemma \ref{lemma:arbit_blocklength}, we conclude that the \ac{rdf} for compression of source sequences whose blocklength is an integer multiple of $N_p$ is the same as the \ac{rdf} for compressing source sequences whose blocklength is arbitrary. We recall that from \cite[Ch. 10.3.3]{cover2006elements} it follows that for the zero-mean memoryless Gaussian \ac{dcd} vector source process $\myVec{S}^{(N_p)}[i]$ the optimal reproduction process which achieves the \ac{rdf} is an $N_p \times 1$ memoryless process whose covariance matrix is diagonal with non-identically distributed elements. From \cite{giannakis1998cyclostationary}, we can apply the inverse \ac{dcd} to obtain a \ac{wscs} process. Hence, from Lemma \ref{lemma:arbit_blocklength} we can conclude that the optimal reproduction process for the \ac{dt} \ac{wscs} Gaussian source is a \ac{dt} \ac{wscs} Gaussian process.
 

\section{Problem Formulation and Auxiliary Results}
\label{sec:rsrch_obj} 
   Our objective is to characterize the \ac{rdf} for compression of asynchronously sampled \ac{ct} \ac{wscs} Gaussian sources when the sampling interval is larger than the memory of the source. In particular, we focus on the minimal rate required to achieve a high fidelity reproduction, representing the \ac{rdf} curve for distortion values not larger than the variance of the source. Such  characterization of the \ac{rdf} for asynchronous sampling is essential for comprehending the relationship between the minimal required number of bits and the sampling rate at a given distortion. Our analysis constitutes an important step towards constructing joint source-channel coding schemes in scenarios in which the symbol rate of the transmitter is not necessarily synchronized with the sampling rate of the source to be transmitted. Such scenarios arise, for example, when recording a communications signal for storage or processing, or in compress-and-forward relaying \cite[Ch. 16.7]{el2011network},  \cite{wu2013optimal} in which the relay compresses the sampled received signal, which is then forwarded to the assisted receiver. As the relay operates with its own sampling clock, which need not necessarily be synchronized with the symbol rate of the assisted transmitter, sampling at the relay may result in a \ac{dt} \ac{wsacs} source signal. In the following we first characterize the sampled source model in Subsection \ref{subsec:problem_formulation}. Then, 
as a preliminary step to our characterization the \ac{rdf} for asynchronously sampled \ac{ct} \ac{wscs} Gaussian processes stated in Section \ref{subsec:rdf_Async_samp_wscs}, we  recall in Subsection \ref{subsec:Info_spec_definitions} the definitions of some information-spectrum quantities used in this study. Finally, in Subsection \ref{subsec:rate_info_spec_limit}, we recall an auxiliary result relating the information spectrum quantities of a collection of sequences of RVs to the information spectrum quantities of its limit sequence of RVs. This result will be applied in the derivation of the \ac{rdf} with asynchronous sampling. 
 
\subsection{Source Model}
\label{subsec:problem_formulation}
Consider a real \ac{ct}, zero-mean \ac{wscs} Gaussian random process $S_c (t)$ with period $\Tc$. Let the variance function of $S_c (t)$ be defined as $\Csc(t)\triangleq \E\big\{\Sc^2 (t)\big\}$, and assume it is both upper bounded and lower bounded away from zero, and that it is continuous in $t\in \mathcal{R}$. Let $\tau_m >0$ denote the maximal correlation length of $S_c(t)$, i.e., $r_{S_c}(t,\tau)\triangleq\E\big\{S_c(t)S_c(t-\tau)\big\}=0, \forall |\tau| > \tau_m$. By the cyclostationarity of $\Sc (t)$, we have that $\sigma_{S_c}^2 (t)=\sigma_{S_c}^2 (t+\Tc), \forall t\in \mR$.  
Let $S_c(t)$ be sampled uniformly with the sampling interval $\Tsamp>0$ such that $\Tc=(p+\epsilon)\cdot \Tsamp$ for $p\in \mN$ and $\epsilon \in [0,1)$ yielding $S_{\epsilon} [i]\triangleq\Sc (i\cdot \Tsamp)$, where $i \in \mZ$. The variance of $\Seps[i]$ is given by $\Sigseps[i]\triangleq \rseps [i,0]=\Csc \left(\frac{i\cdot\Tc}{p+\eps}\right)$. 

In this work, as in \cite{shlezinger2019capacity}, we assume that the duration of temporal correlation of the \ac{ct} signal is shorter than the sampling interval $\Tsamp$, namely, $\tau_m < \Tsamp$. Consequently, the \ac{dt} Gaussian process $S_{\epsilon} [i]$ is a memoryless zero-mean Gaussian process and its autocorrelation function is given by:
\begin{align}
\label{eqn:AsyncAutocorr}
    \rseps [i,\Delta]&=\E\bigg\{\Seps[i]\Seps[i+\Delta]\bigg\}\nonumber\\&= \E \left\{\Sc\left(\frac{i\cdot \Tc}{\Td+\eps}\right)\cdot \Sc\left(\frac{(i+\Delta)\cdot \Tc}{\Td+\eps}\right)\right\}=\Csc \left(\frac{i\cdot\Tc}{\Td+\eps}\right)\cdot \delta[\Delta]=\Sigseps[i]\cdot \delta[\Delta].
\end{align} 
While we do not explicitly account for sampling time offsets in our definition of the sampled process $\Seps [i]$, it can be incorporated by replacing $\Csc (t)$ with a time-shifted version, i.e., $\Csc (t-\phi)$, see also \cite[Sec. II.C]{shlezinger2019capacity}. 

It can be noted from \eqref{eqn:AsyncAutocorr} that if $\eps$ is a rational number, i.e., $\exists u,v\in \mN$, $u$ and $v$ are relatively prime,  such that $\eps=\frac{u}{v}$, then $\left\{\Seps [i]\right\}_{i\in\mZ}$ is a \ac{dt} memoryless \ac{wscs} process with the period $\Td_{u,v}=\Td\cdot v+ u \in \mN$ \cite[Sec. II.C]{shlezinger2019capacity}. For this class of processes, the \ac{rdf} can be obtained from \cite[Thm. 1]{kipnis2018distortion} as stated in Corollary  \ref{corollary:RDF_memoryless_WSCS}. On the other hand, if $\eps$ is an irrational number, then sampling becomes asynchronous and leads to a \ac{wsacs} process whose \ac{rdf} has not been characterized to date. 

\subsection{Definitions of Relevant Information Spectrum Quantities}
\label{subsec:Info_spec_definitions} 
Conventional information theoretic tools for characterizing \acp{rdf} are based on an underlying ergodicity of the source. Consequently, these techniques cannot be applied to characterize the \ac{rdf} of  of asynchronously sampled \ac{wscs} processes. To tackle this challenge, we use  information spectrum methods. 
  The information spectrum framework \cite{han2003information} can be utilized to obtain general formulas for rate limits for any arbitrary class of processes. The resulting expressions do are not restricted to specific statistical models of the considered processes, and in particular, do not require information stability or stationarity. In the following, we recall the definitions of several information-spectrum quantities used in this study, see also \cite[Def. 1.3.1-1.3.2]{han2003information}: 
\begin{definition}
		\label{def:pliminf}
		The {\em limit-inferior in probability} of a sequence of real \acp{rv} $\{\zk{k}\}_{k \in \mySet{N}}$ is defined as
	
		\begin{equation}
		\label{eqn:pliminf}
		{\rm p-}\mathop{\lim \inf}\limits_{k \rightarrow \infty} \zk{k} \triangleq \sup\left\{\alpha \in \mySet{R} \big| \mathop{\lim}\limits_{k \rightarrow \infty}\Pr \left(\zk{k} < \alpha \right) = 0   \right\}
		\triangleq \alpha_0.
		\end{equation}
\end{definition}
Hence, $\alpha_0$ is the largest real number satisfying that  $\forall \talpha < \alpha_0$ and  $\forall \mu>0$ there exists $k_0(\mu,\talpha) \in \mN$ such that $\Pr(Z_k<\talpha)<\mu$, $\forall k>k_0(\mu,\talpha)$.
	\begin{definition}
		\label{def:plimsup}
		The {\em limit-superior in probability} of a sequence of real \acp{rv} $\{\zk{k}\}_{k \in \mySet{N}}$ is defined as
		\begin{equation}
		\label{eqn:plimsup}
		{\rm p-}\mathop{\lim \sup}\limits_{k \rightarrow \infty} \zk{k} \triangleq \inf\left\{\beta  \in \mySet{R} \big| \mathop{\lim}\limits_{k \rightarrow \infty}\Pr \left(\zk{k} > \beta \right) = 0   \right\}
		\triangleq \beta_0.
		\end{equation}
	\end{definition}
Hence,  $\beta_0$ is the smallest real number satisfying that $\forall \tbeta > \beta_0$ and $\forall \mu>0$, there exists $k_0(\mu,\tbeta)\in\mN$, such that $\Pr(Z_k>\tbeta)<\mu$, $\forall k>k_0(\mu,\tbeta)$.

The notion of uniform integrability of a sequence of \acp{rv} is a basic property in probability \cite[Ch. 12]{zitkovic2015lecture}, which is not directly related to information spectrum methods. However, since it plays an important role in the information spectrum characterization of \acp{rdf}, we include its statement in the following definition:
\begin{definition}[Uniform Integrability {\cite[Def. 12.1]{zitkovic2015lecture},\cite[Eqn. (5.3.2)]{han2003information}}]
\label{def:uniform_integrability}
The sequence of real-valued random variables $\{\zk{k}\}_{k=1}^\infty$, is said to satisfy uniform integrability if
\begin{equation}
\label{Eqn:uniform_integrability}
    \mathop{\lim}\limits_{u\rightarrow \infty}\mathop{\sup}\limits_{k\geq 1}\mathop{\int}\limits_{z:|z|\geq u} \pdf{\zk{k}}\left(z\right)|z|\mathrm{d}z=0
\end{equation}
\end{definition} 

The aforementioned quantities allow characterizing the \ac{rdf} of arbitrary sources. 
Consider a general source process $\{S[i]\}_{i=1}^{\infty}$ (stationary or non-stationary) taking values from the source alphabet $S[i]\in \mS$ and a reproduction process $\{\hat{S}[i]\}_{i=1}^{\infty}$ with values from the reproduction alphabet $\hat{S}[i]\in \hat{\mS}$. It follows from \cite[Sec. 5.5]{han2003information} that for a distortion measure which satisfies the uniform
integrability criterion, i.e., that there exists a deterministic sequence  $\{r[i]\}_{i=1}^{\infty}$ such that the sequence of \acp{rv} $\{d\big(\myVec{S}^{(k)},\myVec{r}^{(k)}\big)\}_{k=1}^{\infty}$ satisfies Def. \ref{def:uniform_integrability} \cite[Pg. 336]{han2003information}, then the \ac{rdf} is expressed as \cite[Eqn. (5.4.2)]{han2003information}:
\begin{equation}
\label{eqn:general_rdf}
    R(D)=\mathop{\inf}\limits_{ \cdf{S,\hat{S}}:\bar{d}_S(\myVec{S}^{(k)},\hat{\myVec{S}}^{(k)})\leq D}\bar{I}\left(\myVec{S}^{(k)};\hat{\myVec{S}}^{(k)}\right),
\end{equation} 
where $\bar{d}_S(\myVec{S}^{(k)},\hat{\myVec{S}}^{(k)})=\mathop{\lim \sup}\limits_{k\rightarrow\infty}\frac{1}{k}\E\left\{d\left(\myVec{S}^{(k)},\hat{\myVec{S}}^{(k)}\right)\right\}$, $\cdf{S,\hat{S}}$ denotes the joint \ac{cdf} of $\{S[i]\}_{i=1}^{\infty}$ and $\{\hat{S}[i]\}_{i=1}^{\infty}$, and $\bar{I}\left(\myVec{S}^{(k)}:\hat{\myVec{S}}^{(k)}\right)$ represents the limit superior in probability of the mutual information rate of $\myVec{S}^{(k)}$ and $\Shvec^{(k)}$, given by:
\begin{equation}
\label{eqn:spec_sup_mutualinformation}
   \bar{I}\left(\myVec{S}^{(k)};\hat{\myVec{S}}^{(k)}\right)\triangleq \plimsup \frac{1}{k}\log \frac{\pdf{\myVec{S}^{(k)}|\hat{\myVec{S}}^{(k)}} \left(\myVec{S}^{(k)}|\hat{\myVec{S}}^{(k)} \right)}{\pdf{\myVec{S}^{(k)}}\left( \myVec{S}^{(k)} \right) }
\end{equation}
 
 In order to use the \ac{rdf} characterization in \eqref{eqn:general_rdf}, the distortion measure must satisfy the uniform integrability criterion. For the considered class of sources detailed in Subsection \ref{subsec:problem_formulation}, the \ac{mse} distortion satisfies this criterion, as stated in the following lemma:

\begin{lemma}
\label{lem:Uniform1}
For any real memoryless zero-mean Gaussian source $\{S[i]\}_{i=1}^{\infty}$  with bounded variance, i.e., $\exists \sigma_{\max}^2 <\infty$ such that $\E\{S^2[i]\} \leq \sigma_{\max}^2$  for all $i \in \mN$, the \ac{mse} distortion  satisfies the uniform integrability criterion.
\end{lemma}
 
 \begin{proof}
 Set the deterministic sequence $\{{r}[i]\}_{i=1}^{\infty}$ to be the all-zero sequence. Under this setting and the \ac{mse} distortion, it holds that $d\big(\myVec{S}^{(k)},\myVec{r}^{(k)}\big) = \frac{1}{k}\sum_{i=1}^{k}S^2[i]$.  
 To prove the lemma, we show that the sequence of \acp{rv} $\big\{d\big(\myVec{S}^{(k)},\myVec{r}^{(k)}\big)\big\}_{k=1}^{\infty}$ has a bounded $\ell_2$ norm, which proves that it is uniformly integrable by \cite[Cor. 12.8]{zitkovic2015lecture}. The $\ell_2$ norm of $d\big(\myVec{S}^{(k)},\myVec{r}^{(k)}\big)$ satisfies 
 \begin{align}
     \E\big\{d\big(\myVec{S}^{(k)},\myVec{r}^{(k)}\big)^2\big\} &= \frac{1}{k^2}\E\left\{\sum_{i=1}^{k}S^2[i] \sum_{j=1}^{k}S^2[j]\right\} \notag \\
     &=  \frac{1}{k^2}\sum_{i=1}^{k} \sum_{j=1}^{k}\E\left\{S^2[i]S^2[j]\right\} \stackrel{(a)}{\leq} \frac{1}{k^2}\sum_{i=1}^{k} \sum_{j=1}^{k} 3 \sigma_{\max}^4 =  3\sigma_{\max}^4,
     \label{eqn:Uniform1}
 \end{align}
 where $(a)$ follows since $\E\{S^2[i]S^2[j]\} = \E\{S^2[i]\}\E\{S^2[j]\} = \sigma_{\max}^4$ for $i\neq j$ while $\E\{S^4[i]\} = 3\sigma_{\max}^4$ \cite[Ch. 5.4]{Papoulis91}. Eqn. \eqref{eqn:Uniform1} proves that $d\big(\myVec{S}^{(k)},\myVec{r}^{(k)}\big)$ is $\ell_2$-bounded by $3\sigma_{\max}^4 < \infty$ for all $k \in \mN$, which in turn implies that the \ac{mse} distortion is uniformly integrable for the source $\{S[i]\}_{i=1}^{\infty}$.
 \end{proof}

Since, as detailed in Subsection \ref{subsec:problem_formulation}, we focus in the following on memoryless zero-mean Gaussian sources, Lemma \ref{lem:Uniform1} implies that the \ac{rdf} of the source can be characterized using \eqref{eqn:general_rdf}. However, \eqref{eqn:general_rdf} is in general difficult to evaluate, and thus does not lead to a meaningful understanding of how the  \ac{rdf} of sampled \ac{wscs} sources behaves, motivating our analysis in Section \ref{subsec:rdf_Async_samp_wscs}.

	\subsection{Information Spectrum Limits}
	\label{subsec:rate_info_spec_limit} 
The following theorem originally stated in \cite[Thm. 1]{shlezinger2019capacity}, presents a fundamental result which is directly useful for the derivation of the \ac{rdf}:
\begin{theorem}
\label{thm:plim}
 \cite[Thm. 1]{shlezinger2019capacity} Let $\left\{\zkn{k}{n}\right\}_{n,k \in \mN}$ be a set of sequences of real scalar RVs satisfying two assumptions:
\begin{enumerate}[label={\em AS\arabic*}]
\item \label{itm:assm1} For every fixed $n \in \mN$, every convergent subsequence of $\left\{\zkn{k}{n}\right\}_{k \in \mN}$ converges in distribution, as $k \rightarrow \infty$, to a finite deterministic scalar. Each subsequence may converge to a different scalar.
\item \label{itm:assm2} For every fixed $k \in \mySet{N}$, the sequence $\big\{\zkn{k}{n}\big\}_{n\in \mySet{N}}$ converges uniformly in distribution, as $n \rightarrow \infty$,  to a scalar real-valued \ac{rv} $\zkeps{k}$. Specifically,  letting $\Fkn{k}{n}(\alpha)$ and $\Fkeps{k}(\alpha)$, $\alpha \in \mySet{R}$, denote the \acp{cdf} of $\zkn{k}{n}$ and of $\zkeps{k}$, respectively, then by \ref{itm:assm2} it follows that $\forall \eta>0$, there exists $n_0(\eta)$ such that for every $ n > n_0(\eta)$ 
\begin{equation*} 
			\left|\Fkn{k}{n}(\alpha) - \Fkeps{k}(\alpha)\right| < \eta, 
			\end{equation*}
for each $\alpha \in \mySet{R}$, $ k \in \mySet{N}$.
\end{enumerate}
Then, for  $\big\{\zkn{k}{n}\big\}_{n,k \in \mySet{N}}$ it holds that 
		\begin{subequations}
			\label{eqn:plim}
			\begin{eqnarray}
			\label{eqn:plima}
			{\rm p-}\mathop{\lim \inf}\limits_{k \rightarrow \infty} \zkeps{k} &=&\mathop{\lim}\limits_{n \rightarrow \infty} \Big({\rm p-}\mathop{\lim \inf}\limits_{k \rightarrow \infty} \zkn{k}{n} \Big), \\
			\label{eqn:plimb}
			{\rm p-}\mathop{\lim \sup}\limits_{k \rightarrow \infty} \zkeps{k} &=&\mathop{\lim}\limits_{n \rightarrow \infty} \Big({\rm p-}\mathop{\lim \sup}\limits_{k \rightarrow \infty} \zkn{k}{n} \Big).
			\end{eqnarray}
		\end{subequations}
\end{theorem}
\begin{proof}
In Appendix \ref{app:Proof1} we explicitly prove Eqn. \eqref{eqn:plimb}. This complements the proof in \cite[Appendix A]{shlezinger2019capacity} which explicitly considers only \eqref{eqn:plima}.
\end{proof}

	\section{Rate-Distortion Characterization for Sampled CT WSCS Gaussian Sources}
	\label{subsec:rdf_Async_samp_wscs} 
	
\subsection{Main Result}
\label{subsec:MainResult}
Using the information spectrum based characterization of the \ac{rdf}  \eqref{eqn:general_rdf} combined with the characterization of the limit of a sequence of information spectrum quantities in Theorem \ref{thm:plim}, we now analyze the \ac{rdf} of asynchronously sampled \ac{wscs} processes. Our analysis is based on formulating a sequence of synchronously sampled \ac{wscs} processes, whose \ac{rdf} is given in Corollary \ref{corollary:RDF_memoryless_WSCS}. Then, we show that the \ac{rdf} of the  asynchronously sampled process can be obtained as the limit superior of the computable \acp{rdf} of the sequence of synchronously sampled processes.
We begin by letting $\myEps_n \triangleq \frac{\lfloor n \cdot \myEps \rfloor}{n}$ for $n\in \mN$ and defining a Gaussian source process $\Sn[i]=\Sc\left(\frac{i\cdot \Tc}{p+\myEpsn}\right)$. From the discussion in Subsection \ref{subsec:problem_formulation} (see also \cite[Sec. II.C]{shlezinger2019capacity}), it follows that since $\myEpsn$ is rational, $\Sn [i]$ is a \ac{wscs} process and its period is given by $p_n=p\cdot n +\lfloor n\cdot \myEps\rfloor$. Accordingly, the periodic correlation function of $\Sn[i]$  can be obtained similarly to \eqref{eqn:AsyncAutocorr} as: 
	\begin{equation}
	    \rsn[i,\Delta]=\E\bigg\{\Sn[i]\Sn[i+\Delta]\bigg\}=\Csc \left(\frac{i\cdot \Tc}{p+\myEpsn}\right)\cdot \delta [\Delta].
	\end{equation}
Due to cyclostationarity of $S_n[i]$, we have that $\rsn [i,\Delta]=\rsn [i+p_n,\Delta]$, $\forall i,\Delta \in \mZ$, and let $\Sigsn [i]\triangleq r_{S_n} [i,0]$ denote its periodic variance.

We next restate Corollary \ref{corollary:RDF_memoryless_WSCS} in terms of $\myEpsn$ as follows:
\begin{proposition}
\label{prop:RDF_WSCS}
Consider a \ac{dt}, memoryless, zero-mean, \ac{wscs} Gaussian random process $\Sn[i]$ with a variance $\Sigsn [i]$, obtained from $\Sc(t)$ by sampling with a sampling interval of $T_s (n)=\frac{\Tc}{\Td +\myEpsn}$. Let $\myVec{S}^{p_n}[i]$ denote the memoryless stationary multivariate random process obtained by applying the \ac{dcd} to $\Sn[i]$ and let $\Sigsn[m]$, $m=1,2, \ldots, p_n$, denote the variance of the $m^{th}$ component of $\myVec{S}^{p_n}[i]$. The rate-distortion function is given by:

\begin{subequations}
\label{rdf_n}

\begin{align}
\label{eqn:Rate_WSCS_u,v}
    R_{n}(D)= \begin{cases}
  \frac{1}{2\Td_{n}} \sum\limits_{m=1}^{\Td_{n}}\left(\log \left(\frac{\Sigsn[m]}{D_n[m]}\right)\right) & D\leq \frac{1}{\Td_{n}}\sum\limits_{m=1}^{\Td_{n}}\Sigsn[m]\\
   0 & D> \frac{1}{\Td_{n}}\sum\limits_{m=1}^{\Td_{n}}\Sigsn[m]
    \end{cases},
\end{align}
where for  $D\leq \frac{1}{\Td_{n}}\sum\limits_{m=1}^{\Td_{n}}\Sigsn[m]\;$ we let $D_n[m]\triangleq\min\big\{\Sigsn[m],\theta_n\big\}$, and $\theta_n$ is selected $s.t.$ 
\begin{equation}
\label{eqn:Distortion_WSCS_u,v}
   D=\frac{1}{\Td_{n}} \sum\limits_{m=1}^{\Td_{n}} D_n[m].
\end{equation}

\end{subequations}
\end{proposition}
We recall that the \ac{rdf} of $S_n[i]$ is characterized in Corollary \ref{corollary:RDF_memoryless_WSCS} via the \ac{rdf} of the multivariate stationary process $\myVec{S}_{n}^{(p_n)}[i]$ obtained via a $p_n$-dimensional \ac{dcd} applied to $S_n[i]$. Next, we recall that the relationship between the source process $\myVec{S}_{n}^{(p_n)}[i]$ and the optimal reconstruction process, denoted by $\hat{\myVec{S}}_{n}^{(p_n)}[i]$, is characterized in \cite[Ch. 10.3.3]{cover2006elements} via a linear, multivariate, time-invariant backward channel with a $p_n \times 1$ additive  vector noise process $\myVec{W}_{n}^{(p_n)}[i]$, and is given by:
\begin{equation}
\label{eqn:tst_chan_DCD_vct}
    \myVec{S}_{n}^{(p_n)}[i]=\hat{\myVec{S}}_n^{(p_n)}[i] + \myVec{W}_n^{(p_n)}[i], \quad i\in \mN.
\end{equation}
It also follows from \cite[Sec. 10.3.3]{cover2006elements} that for the \ac{iid}  Gaussian multivariate process whose entries are independent and distributed via  $\big(\myVec{S}_{n}^{(p_n)}[i]\big)_m\sim \mN(0,\Sigsn[m])$, $m \in \{1,2,\ldots,p_n\}$, the optimal reconstruction vector process $\hat{\myVec{S}}_n^{(p_n)}[i]$ and the corresponding noise vector process $\myVec{W}_n^{(p_n)}[i]$ each follow a multivariate Gaussian distribution:
\begin{equation*}
    \hat{\myVec{S}}_n^{(p_n)}[i]\sim \mN\left({\bf 0},  \left[\begin{matrix}
    \Sigshn[1] & \cdots & 0\\
    \vdots & \ddots & \vdots\\
    0 & \cdots & \Sigshn[p_n]
    \end{matrix}\right]\right) \quad \mathrm{and} \quad  \myVec{W}_n^{(p_n)}[i]\sim \mN\left({\bf 0},  \left[\begin{matrix}
   D_n[1] & \cdots & 0\\
    \vdots & \ddots & \vdots\\
    0 & \cdots & D_n[p_n]
    \end{matrix}\right]\right),
\end{equation*}
where $D_n[m]\triangleq \min \left\{\Sigsn[m],\theta_n\right\}$; $\theta_n$ denotes the reverse waterfilling threshold defined in Prop. \ref{prop:RDF_WSCS} for the index $n$, and is selected such that $D=\frac{1}{p_n} \sum\limits_{m=1}^{p_n} D_n[m]$. The optimal reconstruction process, $\hat{\myVec{S}}_n^{(p_n)}[i]$ and the noise process $ \myVec{W}_n^{(p_n)}[i]$ are mutually independent, and for each  $m \in \{1,2,\ldots,p_n\}$ it holds that $\E\left\{\left(S_{n}^{(p_n)}[i]-\hat{S}_n^{(p_n)}[i]\right)_m^2\right\}=D_n[m]$, see \cite[Ch. 10.3.2-10.3.3]{cover2006elements}. The multivariate relationship between stationary processes in \eqref{eqn:tst_chan_DCD_vct} can be transformed into an equivalent linear relationship between cyclostationary Gaussian memoryless processes via the inverse \ac{dcd} transformation \cite[Sec 17.2]{giannakis1998cyclostationary} applied to each of the processes, resulting in:
\begin{equation}
\label{eqn:testchannelrdf_synch}
    \Sn[i]=\hat{S}_n[i]+W_n[i] \quad i\in\mN.
\end{equation}

We are now ready to state our main result, which is the \ac{rdf} of the asynchronously sampled \ac{dt} source $\Seps[i], \eps \not \in \mySet{Q}$, in the low \ac{mse} regime, i.e., at a given distortion $D$ which is not larger than the source variance. The \ac{rdf} is stated in the following theorem, which  applies to both  synchronous sampling as well as asynchronous sampling:
\begin{theorem}
\label{Thm:rate_WSACS}
Consider a \ac{dt} source $\{S_{\eps}[i]\}_{i=1}^{\infty}$ obtained by sampling a \ac{ct} \ac{wscs} source, whose period of statistics is $\Tc$, at intervals $\Tsamp$. Then, for any distortion constraint $D$ such that $D< \mathop{\min}\limits_{0\leq t\leq \Tc}\Csc(t)$ and any  $\eps \in [0,1)$, the \ac{rdf} $R_\eps(D)$ for compressing $\{S_{\eps}[i]\}_{i=1}^{\infty}$ can be obtained as the limit:
\begin{equation}
\label{eqn:rdf_asyncsampling}
    R_\eps (D)=\mathop{\lim \sup}\limits_{n\rightarrow \infty} R_n(D),
\end{equation}
where $R_n(D)$ is defined Prop. \ref{prop:RDF_WSCS}.
\end{theorem}

\begin{proof}
The detailed proof is provided in Appendix \ref{app:Proof2}. Here, we give a brief outline: The derivation of the \ac{rdf} with asynchronous sampling follows three steps: First, we note that sampling rate  $\Tsamp(n)=\frac{\Tc}{p+\myEpsn}$ used to obtain  the sequence of \ac{dt} \ac{wscs} sources $\{S_n[i]\}_{i\in\mN, n\in\mN}$  asymptotically approaches the sampling interval for irrational $\myEps$ given by $\Tsamp=\frac{\Tc}{p+\myEps}$ as $n\rightarrow\infty$. We define a sequence of rational numbers $\myEpsn$ $s.t.$ $\myEpsn\rightarrow \myEps$ as $n\rightarrow\infty$; Building upon this insight, we prove that the \ac{rdf} with $\Tsamp$ can be stated as a double limit where the outer limit is with respect to the blocklength and the inner limit is with respect to $\myEpsn$. Lastly, we use Theorem \ref{thm:plim} to show that the limits can be exchanged, obtaining a limit of expressions which are computable.
\end{proof}

\begin{remark}
\label{stationary_noise}
{Theorem \ref{Thm:rate_WSACS} focuses on the low distortion regime, defined as the values of $D$ satisfying $D< \mathop{\min}\limits_{0\leq t\leq\Tc}\Csc(t)$. This implies that $\theta_n$ has to be smaller than $\mathop{\min}\limits_{0\leq t\leq\Tc}\Csc(t)$; hence, from Prop. \ref{prop:RDF_WSCS} it follows that for the corresponding stationary noise vector $\myVec{W}_n^{(p_n)}[i]$ in \eqref{eqn:tst_chan_DCD_vct},  $D_n[m]=\min \left\{\Sigsn[m], \theta_n \right\}=\theta_n$ and $D=\frac{1}{p_n}\mathop{\sum}\limits_{m=1}^{p_n}D_n[m]=\theta_n=D_n[m]$. We note that since every element of the vector $\left(\myVec{W}_n^{(p_n)}[i]\right)_m$ has the same variance $D_n[m]=D$ for all $n\in \mN$ and $m=1,2,\ldots, p_n$ then by applying the inverse \ac{dcd} to $\myVec{W}_n^{(p_n)}[i]$, the resulting  scalar \ac{dt} process $W_n[i]$ is {\em wide sense stationary}; and in fact \ac{iid} with $\E\left\{\big(W_n[i] \big)^2  \right\}=D$}.
\end{remark}
\subsection{Discussion and Relationship with Capacity Derivation in \cite{shlezinger2019capacity}}
\label{discussion}

Theorem \ref{Thm:rate_WSACS} provides a meaningful and computable characterization for the \ac{rdf} of sampled \ac{wscs} signals. 
We note that the proof of the main theorem uses some of the steps used in our recent study on the capacity of memoryless channels with sampled \ac{ct} \ac{wscs} Gaussian noise \cite{shlezinger2019capacity}. It should be emphasized, however, that there are several fundamental differences between the two studies, which require the introduction of new treatments and derivations original to the current work. 
First, it is important to note that in the study on capacity, a physical channel model exists, and therefore the conditional \ac{pdf} of the output signal given the input signal can be characterized explicitly for both synchronous sampling and asynchronous sampling for {\em every} input distribution. For the current study of the \ac{rdf} we note that the relationship \eqref{eqn:testchannelrdf_synch}, commonly referred to as the backward channel \cite{zamir2008achieving}, \cite[Ch.10.3.2]{cover2006elements}, characterizes the relationship between the source process and the {\em optimal} reproduction process, and hence is valid only for synchronous sampling and the optimal reproduction process. Consequently, in the \ac{rdf} analysis the limiting relationship \eqref{eqn:testchannelrdf_synch} as $n\rightarrow  \infty$ is not even known to exist and, in fact, we can show it exists under a rather strict condition on the distortion (namely, the condition $D< \mathop{\min}\limits_{0\leq t\leq \Tc}\Csc(t)$ stated in Theorem \ref{Thm:rate_WSACS}). In particular, to prove the statement in Theorem \ref{Thm:rate_WSACS}, we had to show that from the backward channel  \eqref{eqn:testchannelrdf_synch}, we can define an asymptotic relationship, as $n\rightarrow \infty$, which corresponds to the asynchronously sampled source process, denoted by $\Seps[i]$, and relates $\Seps[i]$ with its optimal reconstruction process $\Sheps[i]$. This is done by showing that the \acp{pdf} for the reproduction process $\hat{S}_n[i]$ and noise process $W_n[i]$ from \eqref{eqn:testchannelrdf_synch}, each converge uniformly as $n\rightarrow\infty$ to a respective limiting \ac{pdf}, which has to be defined as well. This enabled us to relate the \acp{rdf} for the synchronous sampling and for the asynchronous sampling cases using Theorem \ref{thm:plim}, eventually leading to \eqref{eqn:rdf_asyncsampling}. Accordingly, 
in our detailed proof of Theorem \ref{Thm:rate_WSACS} given in Appendix \ref{app:Proof2}, Lemmas \ref{lem:tightness} and \ref{lem:rdflowBound} as well as a significant part of Lemma \ref{lem:AsyncZk} are largely new, addressing the special aspects of the proof arising from the fundamental differences between current setup and the setup in \cite{shlezinger2019capacity}, while the derivations of Lemmas \ref{app:proof2A}  and \ref{lem:rdfUpBound} follow similarly to \cite[Lemma B.1]{shlezinger2019capacity} and \cite[Lemma  B.5]{shlezinger2019capacity}, respectively, and parts of Lemma \ref{lem:AsyncZk} coincide with \cite[Lemma B.2]{shlezinger2019capacity}.
 

	\section{Numerical Examples}
	\label{sec:Simulations}
	
	In this section we demonstrate the insights arising from our \ac{rdf} characterization via numerical examples. Recalling that Theorem \ref{Thm:rate_WSACS} states the \ac{rdf} for  asynchronously sampled \ac{ct} \ac{wscs} Gaussian process, $R_\eps(D)$, as the limit supremum of a sequence of \acp{rdf} corresponding to \ac{dt} memoryless \ac{wscs} Gaussian source processes $\left\{R_n(D)\right\}_{n \in \mN}$, we first consider the convergence of $\{R_n(D)\}_{n \in \mySet{N}}$ in Subsection \ref{RDF_VS_n}. Next, in Subsection \ref{RDF_VS_FREQ} we study the variation of the \ac{rdf} of the sampled \ac{ct} process due to changes in the sampling rate and in the sampling time offset.
	
	Similarly to \cite[Sec. IV]{shlezinger2019capacity}, define a periodic continuous pulse function, denoted by $\Pi_{\DC,\Trise}(t)$, with equal rise/fall time $\Trise = 0.01$, duty cycle $\DC \in [0,0.98]$, and period of $1$, i.e., $ \Pi_{\DC,\Trise}(t+1) = \Pi_{\DC,\Trise}(t)$ for all $t \in \mySet{R}$. Specifically, for $t \in [0,1)$ the function $\Pi_{\DC,\Trise}(t)$ is given by 
	
	\begin{equation}
	\label{eqn:PerWaveFunc}
	\Pi_{\DC,\Trise}(t) = \begin{cases}
	\frac{t}{\Trise} & t \in [0,\Trise] \\
	1	&				t \in (\Trise, \DC+\Trise ) \\
	1 - \frac{t-\DC-\Trise}{\Trise} & t \in[\DC +\Trise, \DC +2\cdot \Trise] \\
	0 & t\in (\DC +2\cdot \Trise, 1).
	\end{cases}
	\end{equation}
	In the following, we model the time varying variance of the \ac{wscs} source $\Csc(t)$ to be a linear periodic function of $\Pi_{\DC,\Trise}(t)$. To that aim, we 
	define a time offset between the first sample and the rise start time of the periodic continuous pulse function; we denote the time offset by $\phi \in [0,1)$. This corresponds to the sampling time offset normalized to the period $\Tc$. The variance of  $\Sc(t)$ is periodic function with period $\Tc$ which is given by
	\begin{equation}
	\label{eqn:ScVar}
	\Csc(t) = 0.2 + 4.8 \cdot \Pi_{\DC,\Trise}\left(\frac{t}{\Tc} - \phi\right), \qquad t \in [0, \Tc),
	\end{equation}
	with period of $\Tc = 5$ $\mu$secs.

 	\subsection{Convergence of $R_n(D)$ in $n$}
	\label{RDF_VS_n}
		\begin{figure}
	    \centering
	    \begin{minipage}{0.5\textwidth}
			{\includegraphics[width=\columnwidth]{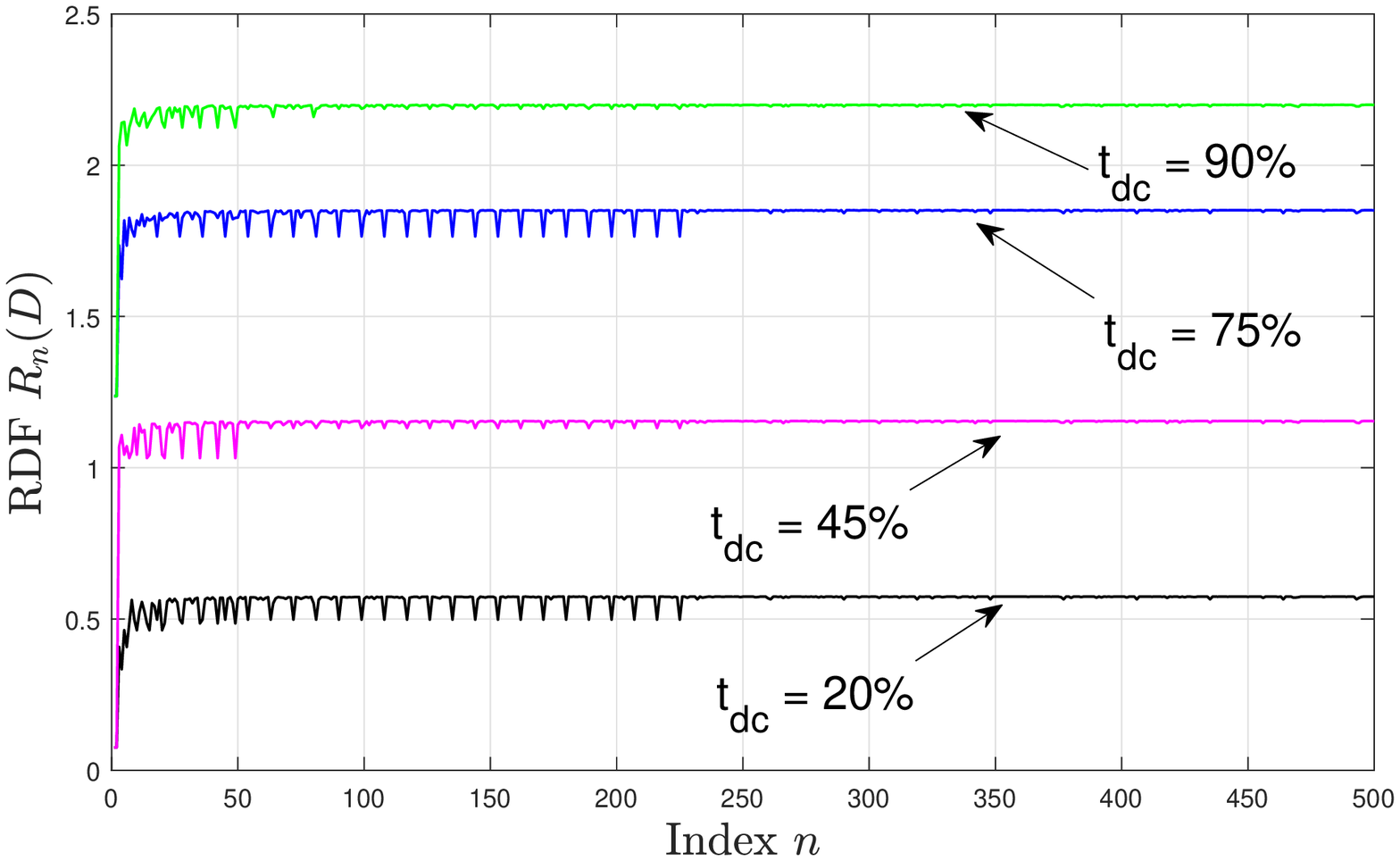}}
			\caption{$R_n(D)$ versus $n$;  offset $\phi=0$.
			}
			\label{fig:R_vs_n_0}		
		\end{minipage}
		\begin{minipage}{0.5\textwidth}
			{\includegraphics[width=\columnwidth]{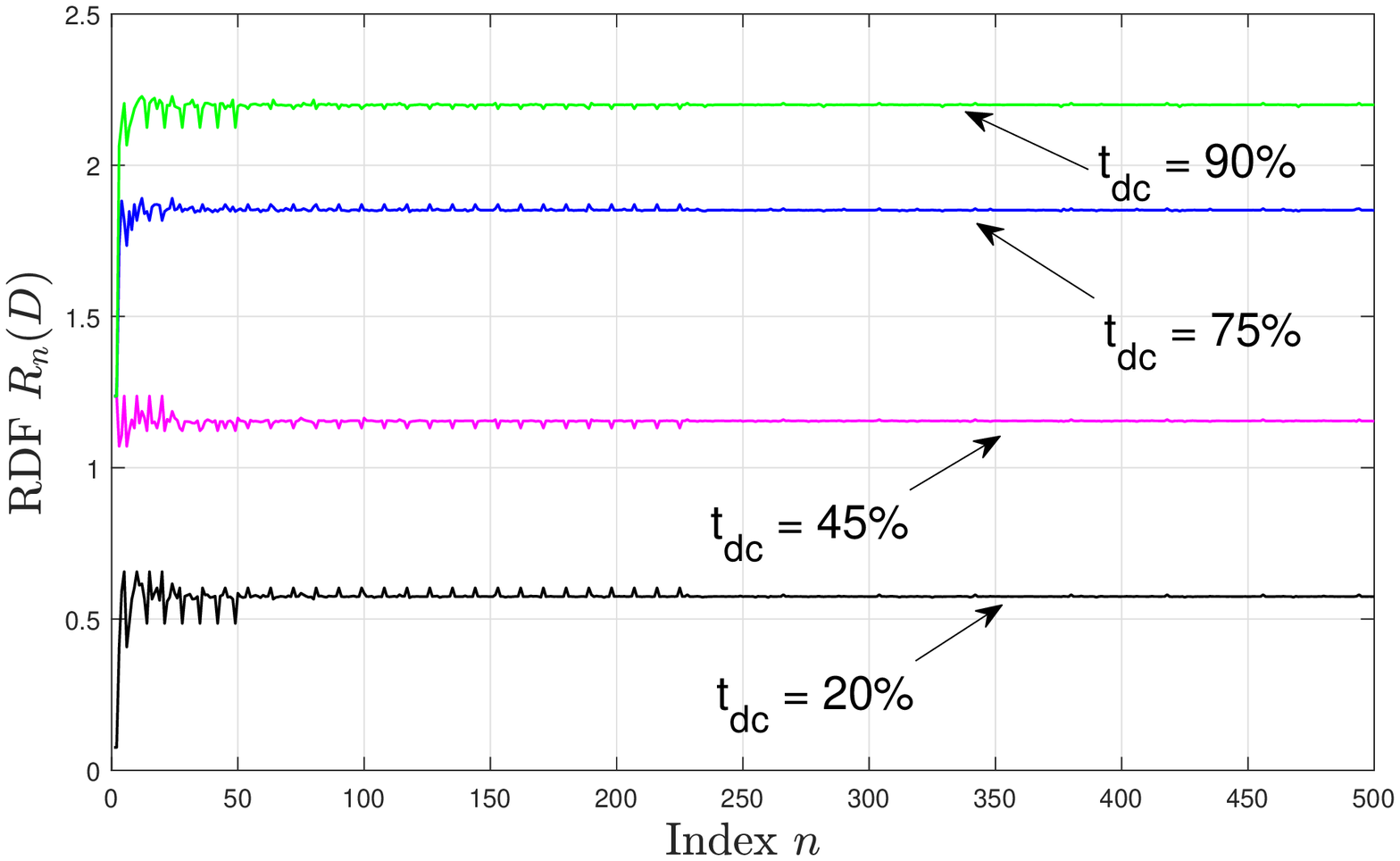}}
			\caption{$R_n(D)$ versus $n$; offset $\phi=\frac{1}{16}$.
			}
			\label{fig:R_vs_n_00625}
		\end{minipage}
		
	\end{figure}
	
	From Theorem \ref{Thm:rate_WSACS} it follows that if the distortion satisfies $D< \mathop{\min}\limits_{0\leq t\leq\Tc}\Csc(t)$, the \ac{rdf} of the asynchronously sampled \ac{ct} \ac{wscs} Gaussian process is given by the limit superior of the sequence $\{R_n(D)\}_{n \in \mySet{N}}$; where $R_n(D)$ is obtained via Corollary \ref{corollary:RDF_memoryless_WSCS}. In this subsection, we study the sequence of \acp{rdf} $\{R_n(D)\}_{n \in \mySet{N}}$ as $n$ increases.  For this evaluation setup, we fixed the distortion constraint at $D=0.18$ and set $\eps=\frac{\pi}{7}$ and $p=2$. Let the variance of the \ac{ct} \ac{wscs} Gaussian source process $\Csc(t)$ be modelled by Eq. \eqref{eqn:ScVar} for two sampling time offsets $\phi=\{0, \frac{1}{16}\}$. For each offset $\phi$, four duty cycle values were considered: $\DC=[20,45,75,98]\%$. For each $n$  we obtain the synchronous sampling mismatch $\myEpsn\triangleq \frac{\lfloor n\cdot \eps\rfloor}{n}$, which approaches $\myEps$ as $n \rightarrow \infty$, where $n\in \mN$. Since $\myEpsn$ is a rational number, corresponding to a sampling period of $\Tsamp(n) = \frac{\Tc}{\Td + \myEps_n}$, then for each $n$, the resulting \ac{dt} process is \ac{wscs} with the period $p_n=\Td\cdot n + \lfloor n\cdot \eps\rfloor$ and its \ac{rdf} follows from  Corollary \ref{corollary:RDF_memoryless_WSCS}.

	Figures \ref{fig:R_vs_n_0} and \ref{fig:R_vs_n_00625} depict $R_n(D)$ for $n \in [1,500]$ with the specified duty cycles and sampling time offsets, where  in Fig. \ref{fig:R_vs_n_0} there is no sampling time offset, i.e., $\phi = 0$, and in Fig. \ref{fig:R_vs_n_00625} the sampling time offset is set to $\phi = \frac{1}{16}$.
	We observe that in both figures the \ac{rdf} values are higher for higher $\DC$. This can be explained by noting that for higher $\DC$ values, the resulting time averaged variance of the \ac{dt} source process increases, hence, a higher number of bits per sample is required to encode the source process maintaining the same distortion value. Also, in all configurations, $R_n(D)$ varies significantly for smaller values of $n$. Comparing  Figures \ref{fig:R_vs_n_0} and  \ref{fig:R_vs_n_00625}, we see that the pattern of these variations depends on the sampling time offset $\phi$. For example, when $\DC = 45 \%$ at $n \in [4, 15]$, then  for $\phi =0$ the \ac{rdf} varies in the range $[1.032, 1.143]$ bits per sample, while  for $\phi =\frac{1}{16}$ the \ac{rdf} varies in the range $[1.071,1.237]$  bits per sample. However, as we increase $n$ above $230$, the variations in $\RnD$ become smaller and are less dependent on the sampling time offset, and the resulting values of $\RnD$ are approximately in the same range in both Figures \ref{fig:R_vs_n_0} and \ref{fig:R_vs_n_00625} for $n \ge 230$. This behaviour can be explained by noting that as $n$ varies, the period $p_n$ also varies and hence the statistics of the \ac{dt} variance differs over its respective period. This consequently affects the resulting \ac{rdf} (especially for small periods). As $n$ increases $\myEpsn$ approaches the asynchronous sampling mismatch $\eps$ and the period $p_n$ takes a sufficiently large value such that the samples of the \ac{dt} variance over the period are similarly distributed irrespective of the value of $\phi$; leading to a negligible variation in the  \ac{rdf} as seen in the above figures. 
	
	\subsection{The Variation of the RDF with the Sampling Rate}
	\label{RDF_VS_FREQ}
	\begin{figure}
	    \centering
	    \begin{minipage}{0.45\textwidth}
			\centering
			{\includegraphics[width=1.1\linewidth]{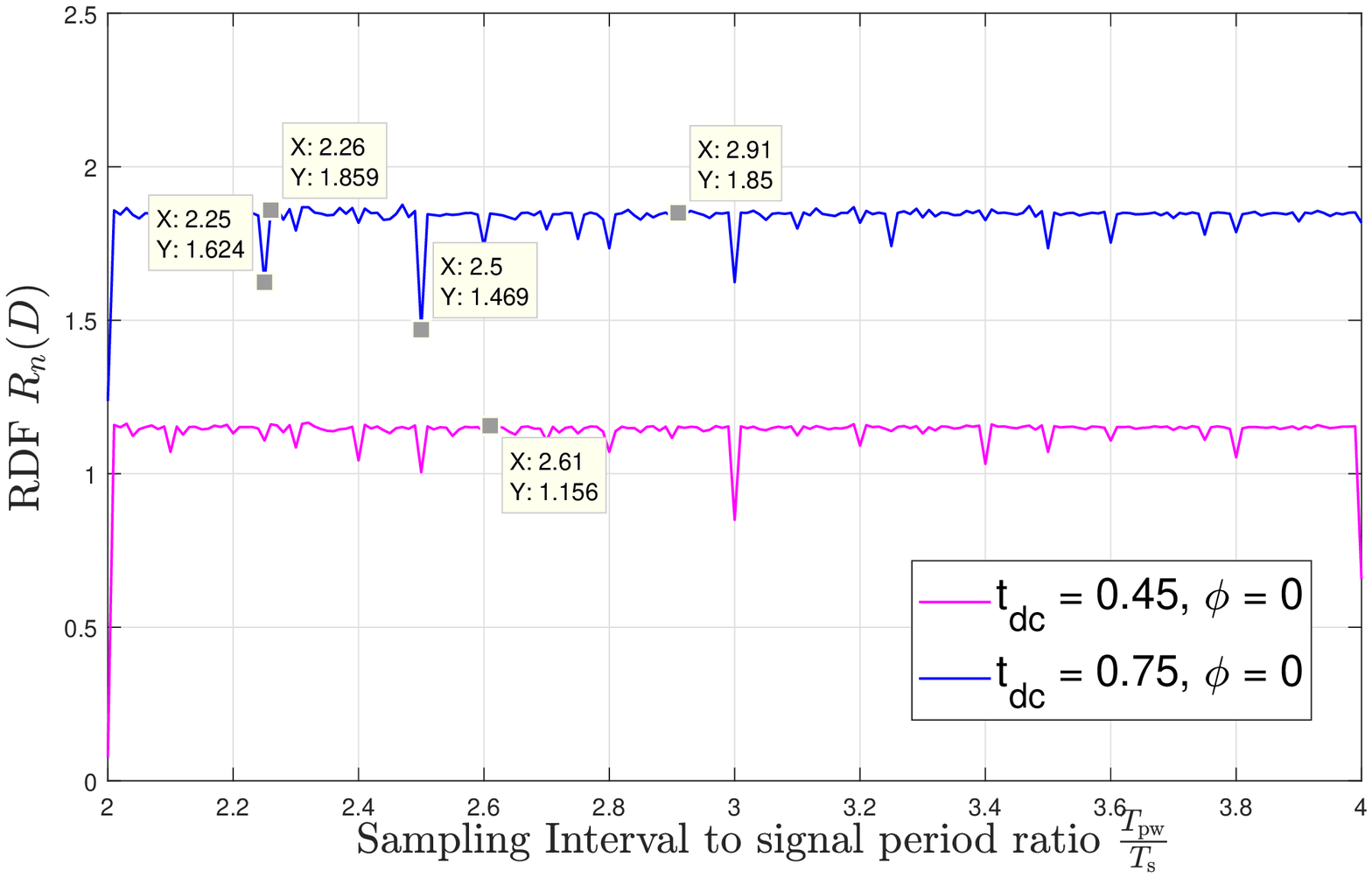}}
			\caption{$R_n(D)$ versus $\frac{\Tc}{\Tsamp}$; offset $\phi=0$.
			}
			\label{fig:R_nbyrate_0}		
		\end{minipage}
		$\quad$
		\begin{minipage}{0.45\textwidth}
			\centering
			{\includegraphics[width=1.1\linewidth]{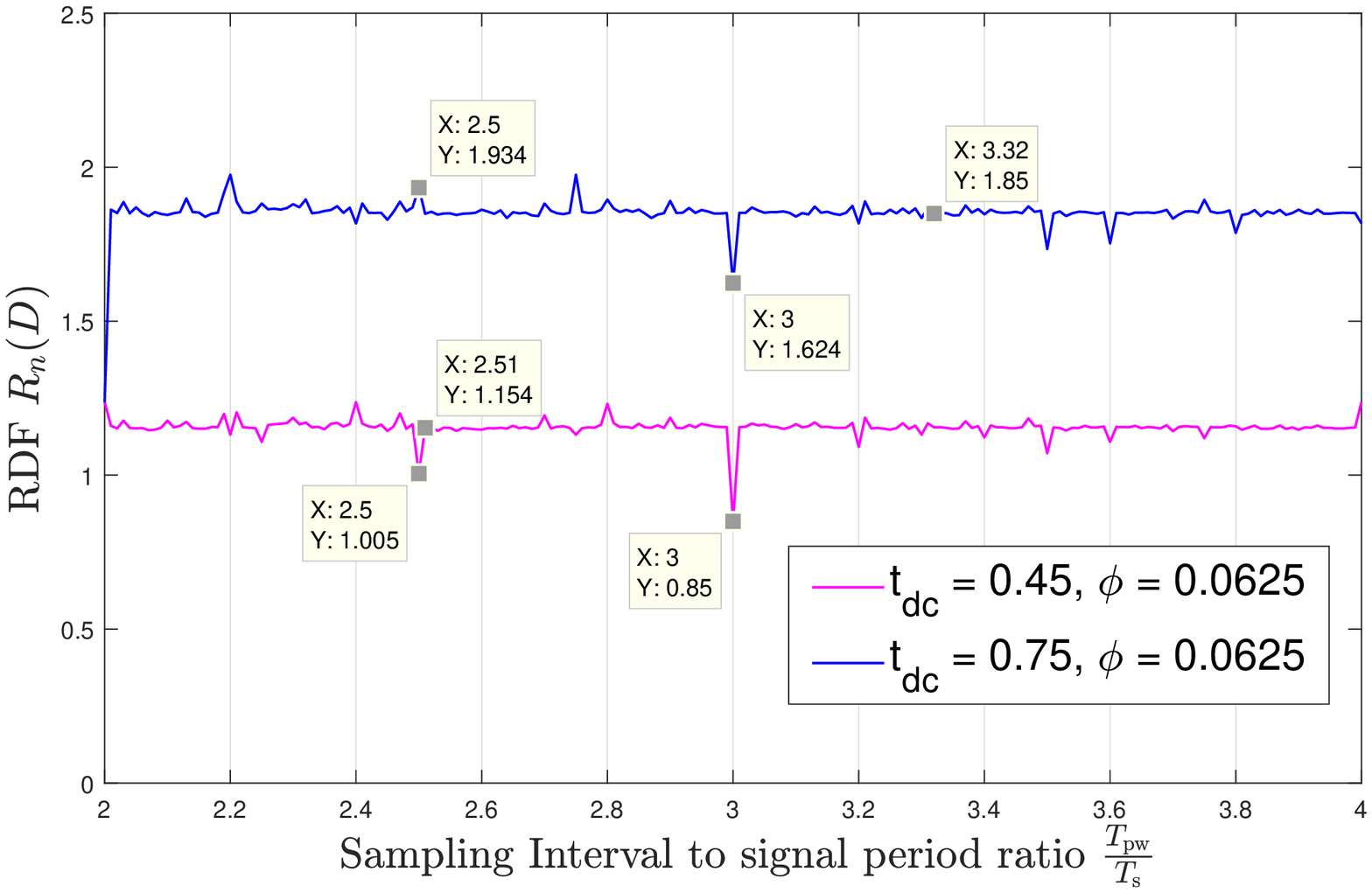}}
			\caption{$R_n(D)$ versus $\frac{\Tc}{\Tsamp}$; offset $\phi=\frac{1}{16}$.
			}
			\label{fig:R_nbyrate_0_0625}
		\end{minipage}
	\end{figure}
	
	Next, we observe the dependence of the \ac{rdf} for the sampled memoryless \ac{wscs} Gaussian process on the value of the sampling interval $\Tsamp$. For this setup,  we fix the distortion constraint to $D = 0.18$ and set the duty cycle in the source process  \eqref{eqn:ScVar} to $\DC = [45, 75] \%$. Figures \ref{fig:R_nbyrate_0}-\ref{fig:R_nbyrate_0_0625} demonstrate the numerically evaluated values for $\RnD$ at sampling intervals in the range $2 < \frac{\Tc}{\Tsamp} < 4 $ with the sampling time offsets $\phi = 0 $ and $\phi = \frac{1}{16}$, respectively. A very important insight which arises from the figures is that the sequence of \acp{rdf} $\RnD$ is not convergent; hence, for example, one cannot approach the \ac{rdf} for $\frac{\Tc}{\Tsamp}=2.5$ by simply taking rational values of $\frac{\Tc}{\Tsamp}$ which approach $2.5$. This verifies that the \ac{rdf} for asynchronous sampling cannot be obtained by straightforward application of previous results, and indeed, the entire analysis carried in the manuscript is necessary for the desired characterization. We observe in Figures \ref{fig:R_nbyrate_0}-\ref{fig:R_nbyrate_0_0625} that when $\frac{\Tc}{\Tsamp}$ has a fractional part with a relatively small integer denominator, the variations in the \ac{rdf} are significant, and the variations depend on the sampling time offset. However, when $\frac{\Tc}{\Tsamp}$ approaches an irrational number, the period of the sampled variance function becomes very long, and consequently, the  \ac{rdf} is approximately constant and independent of the sampling time offset. As an example, consider $\frac{\Tc}{\Tsamp} = 2.5$ and $\DC = 75\%$: For sampling time offset $\phi = 0$ the \ac{rdf} takes a value of  $1.469$ bits per sample, as shown in Figure \ref{fig:R_nbyrate_0} while for the offset of $\phi = \frac{1}{16}$  the \ac{rdf} peaks to $1.934$ bits per sample as we see in Figure \ref{fig:R_nbyrate_0_0625}. On the other hand, when approaching asynchronous sampling, the \ac{rdf} takes an approximately constant value $1.85$ bits per sample for all the considered values of $\frac{\Tc}{\Tsamp}$ and this value is invariant to the offsets of $\phi$. 
	This follows since when the denominator of the fractional part of $\frac{\Tc}{\Tsamp}$ increases, then the \ac{dt} period of the resulting sampled variance, $p_n$, increases and practically captures the entire set of values of the \ac{ct} variance regardless of the sampling time offset. 
	In a similar manner as with the study on capacity in \cite{shlezinger2019capacity}, we conjecture that since asynchronous sampling captures the entire set of values of the \ac{ct} variance, the respective \ac{rdf} represents the \ac{rdf} of the analog source, which does not depend on the specific sampling rate and offset.
Figures \ref{fig:R_nbyrate_0}-\ref{fig:R_nbyrate_0_0625} demonstrate how slight variations in the sampling rate can result in significant changes in the \ac{rdf}. For instance, at $\phi = 0$ we notice in Figure \ref{fig:R_nbyrate_0} that when the sampling rate switches from $\Tsamp = 2.25 \cdot \Tc$ to $\Tsamp = 2.26 \cdot \Tc$, i.e., the sampling rate switches from being synchronous to being nearly asynchronous,  and the \ac{rdf} changes from  $1.624$  bits per channel use to $1.859$  bits per sample for $\DC=75\%$; 
	also, we observe in Figure \ref{fig:R_nbyrate_0_0625} for $\DC=45\%$, that  when the sampling rate switches from $\Tsamp = 2.5 \cdot \Tc$ to $\Tsamp = 2.51 \cdot \Tc$, i.e., the sampling rate also switches from being synchronous to being nearly asynchronous,  and the \ac{rdf} changes from  $1.005$  bits per source sample to $1.154$ bits per source sample. 

	\begin{figure}
	    \centering
	    \begin{minipage}{0.5\textwidth}
			\centering	{\includegraphics[width=1\columnwidth]{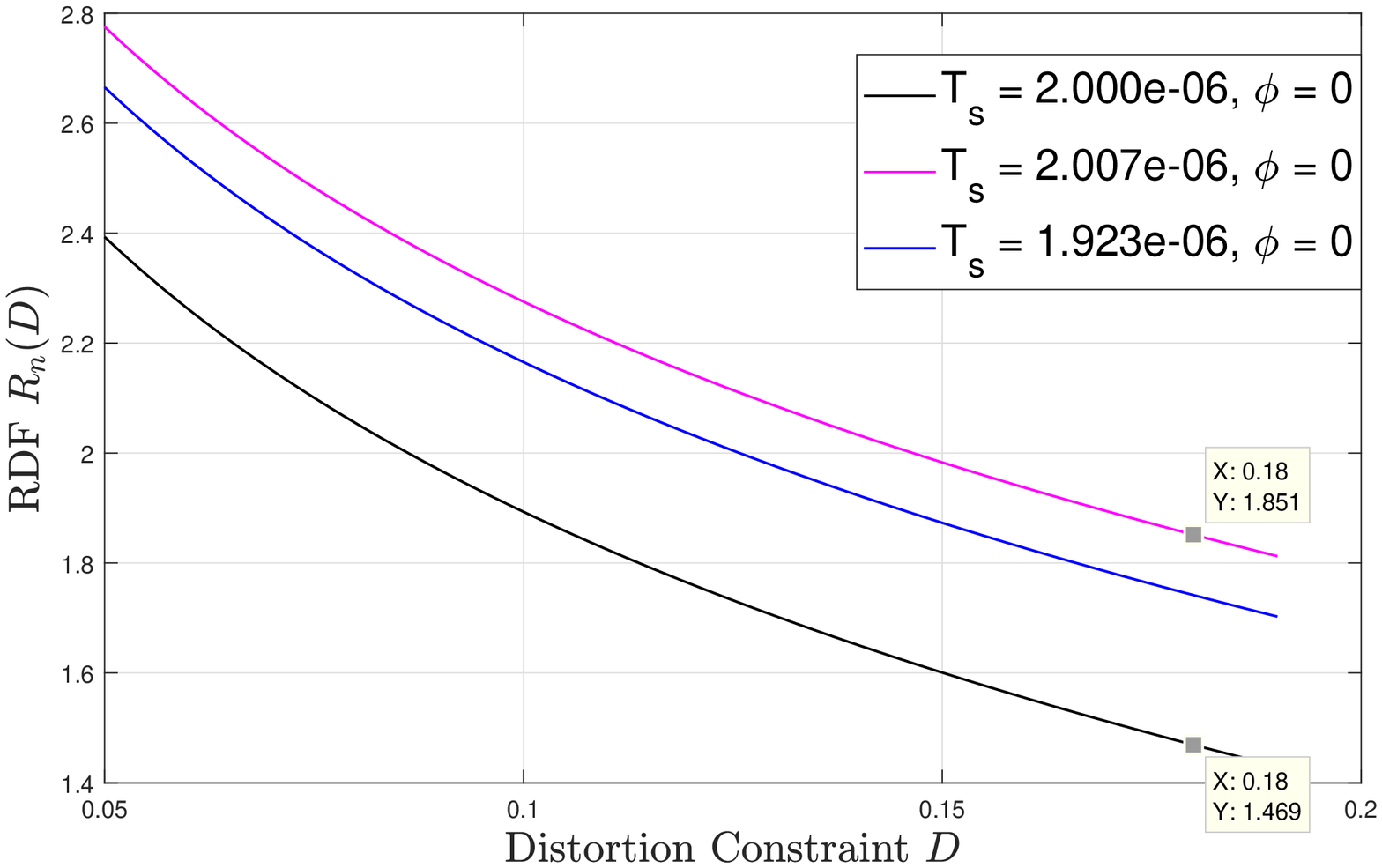}}
			\caption{$R_n(D)$ versus $D$; offset $\phi=0$.
			}
			\label{fig:R_nbyD_0}		
		\end{minipage}
		\begin{minipage}{0.5\textwidth}
			\centering	{\includegraphics[width=1\columnwidth]{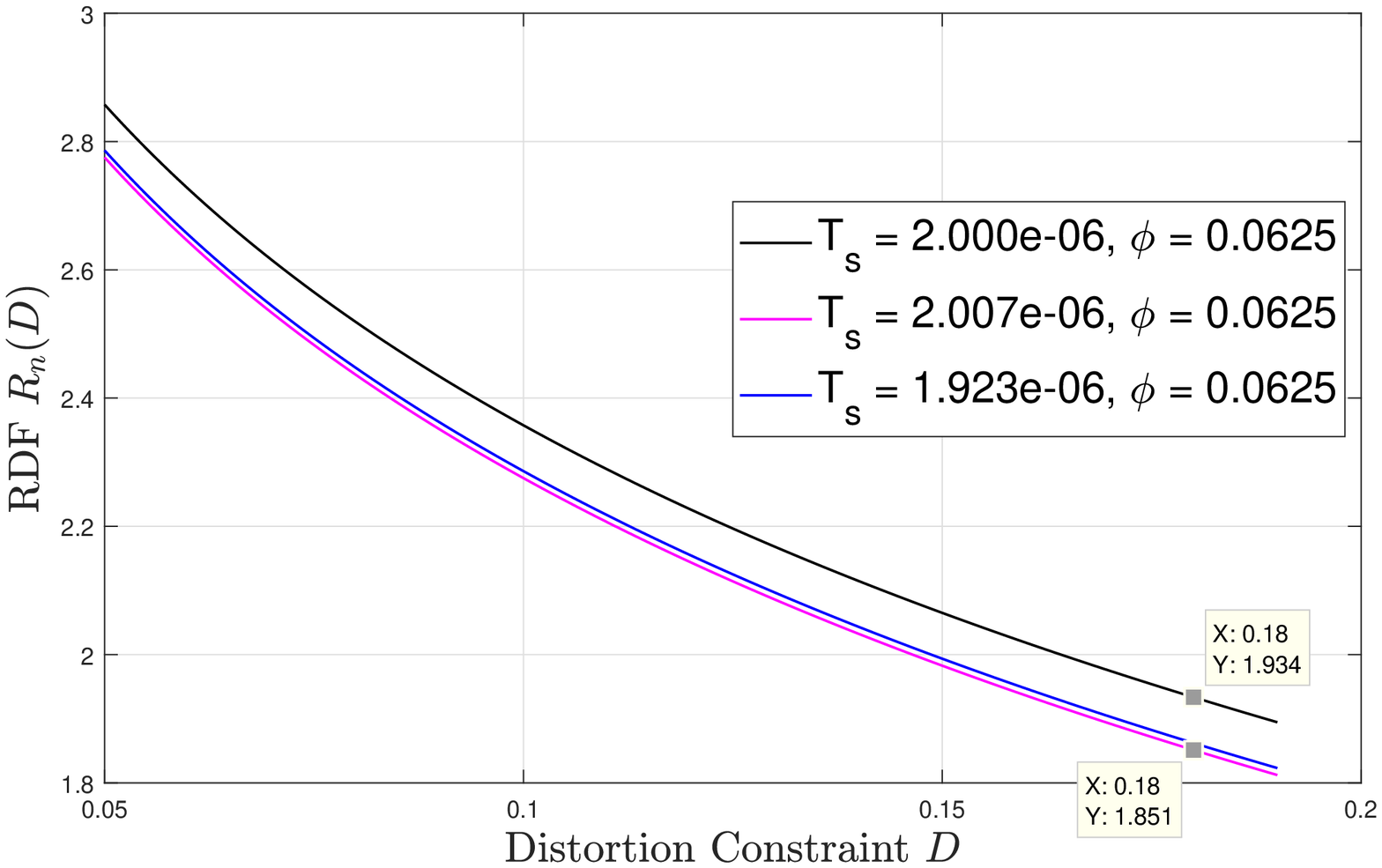}}
			\caption{$R_n(D)$ versus $D$; offset $\phi=\frac{1}{16}$.
			}
			\label{fig:R_nbyD_0_0625}
		\end{minipage}
	\end{figure}

Lastly, Figures \ref{fig:R_nbyD_0}-\ref{fig:R_nbyD_0_0625} numerically evaluate the \ac{rdf} versus the distortion constraint $D\in [0.05,0.19]$ for the sampling time offsets of $0$ and $\frac{1}{16}$ respectively. At each $\phi$, the result is evaluated at three different values of synchronization mismatch $\myEps$. For this setup, we fix $\DC=75\%$, $p=2$ and $\myEps \in \{0.5, \frac{5\pi}{32}, 0.6\}$. The only mismatch value that refers to the asynchronous sampling case is $\myEps=\frac{5\pi}{32}$ and its corresponding sampling interval is approximately $2.007$ $\mu$secs, which is a negligible variation from the sampling intervals corresponding to $\myEps \in \{0.5, 0.6\}$, which are $2.000$ $\mu$secs and $1.923$ $\mu$secs, respectively. Observing both figures, we see that the \ac{rdf} may vary significantly for very slight variation in the sampling rate. For instance, as shown in Figure \ref{fig:R_nbyD_0} for $\phi=0$, at $D=0.18$, a slight change in the synchronization mismatch from $\eps=\frac{5\pi}{32}$ (i.e., $\Tsamp\approx 2.007\mu$secs) to $\eps=0.5$ (i.e., $\Tsamp=2.000\mu$secs) results to approximately $20\%$ decrease in the \ac{rdf}. For $\phi=\frac{1}{16}$ the same change in the sampling synchronization mismatch at $D=0.18$ results to a rise in the \ac{rdf} by roughly $4\%$.
These results demonstrate the unique and counter-intuitive characteristics of the \ac{rdf} of sampled \ac{wscs} signals which arise from our derivation.

\section{Conclusions}
\label{conclusion}
In this work the \ac{rdf} of a sampled \ac{ct} \ac{wscs} Gaussian source process was characterized  for scenarios in which the resulting \ac{dt} process is memoryless. 
This characterization shows the relationship between the sampling rate and the minimal number of bits required for compression at a given distortion.
For cases in which the sampling rate is synchronized with the period of the statistics of the source process, the resulting \ac{dt} process is \ac{wscs} and  standard information theoretic framework can be used for deriving its \ac{rdf}. For asynchronous sampling, information stability does not hold, and hence we resort to the information spectrum framework to obtain a characterization. To that aim we derived a relationship between some relevant information spectrum quantities for uniformly convergent sequences of \acp{rv}. This relationship was further applied to characterize the \ac{rdf} of an asynchronously sampled  \ac{ct} \ac{wscs} Gaussian source process as the limit superior of a sequence of \acp{rdf}, each corresponding to the synchronous sampling of the \ac{ct} \ac{wscs} Gaussian process. The results were derived in the low distortion regime, i.e., under the condition that the distortion constraint $D$ is less than the minimum variance of the source, and for sampling intervals which are larger than the correlation length of the CT process.
Our numerical examples give rise to non-intuitive insights which follow from the derivations. 
In particular, the numerical evaluation  demonstrates  that the \ac{rdf} for a sampled \ac{ct} \ac{wscs} Gaussian source can change dramatically with minor variations in sampling rate and sampling time offset. In particular, when the sampling rate switches from being synchronous to being asynchronous and vice versa, the \ac{rdf} may change considerably as the statistical model of the source switches between \ac{wscs} to \ac{wsacs}. The resulting analysis enables determining the sampling system parameters in order to facilitate accurate and efficient source coding of acquired \ac{ct} signals. 
\begin{appendices}

	\numberwithin{proposition}{section} 
	\numberwithin{lemma}{section} 
	\numberwithin{corollary}{section} 
	\numberwithin{remark}{section} 
	\numberwithin{equation}{section}
	\numberwithin{assumption}{section}

\color{black}
\section{Proof of Lemma \ref{lemma:arbit_blocklength}}
\label{app:arb_blocklength}

\begin{proof}
To prove that the minimum achievable rate at a given maximum distortion for a code with arbitrary blocklength can be achieved by considering only codes whose blocklength is an integer multiple of $r$, we apply the following approach: We first show that every rate-distortion pair achievable when restricted to using source codes whose blocklength is an integer multiple of $r$ is also achievable when using arbitrary blocklenghts;
 We then prove that every achievable rate-distortion pair is also achievable when restricted to using codes whose blocklength is an integer multiple of $r$. 
    Combining these two assertions proves that the rate-distortion function of the source $\{S[i]\}_{i\in\mN}$ can be obtained when restricting the blocklengths to be an integer multiple of $r$. Consequently, a reproduction signal $\{\hat{S}[i]\}_{i\in\mN}$ which achieves the minimal rate for a given $D$ under the restriction to use only blocklengths which are an integer multiple of $r$ is also the reproduction signal achieving the minimal rate without this restriction, and vice versa, thus proving the lemma. 

To prove the first assertion, consider a rate-distortion pair $(R,D)$ which is achievable when using codes whose blocklength is an integer multiple of $r$. It thus follows directly from Def. \ref{def:rate_dist_pair} 
%
that for every  $\eta>0$, $\exists b_0\in \mN$ such that for all $b>b_0$ there exists a a source code $\left(R_{(b\cdot r)},b\cdot r\right)$ with rate $R_{(b\cdot r)}\leq R+\eta$ satisfying $\bar{d}\big(\myVec{S}^{(b\cdot r)},\hat{\myVec{S}}^{(b\cdot r)}\big)\leq D+\frac{\eta}{2}$. 
We now show that we can construct a  code with an arbitrary blocklength $l=b\cdot r+j$ where $0< j<r$ (i.e., the blocklength $l$ is not an integer multiple of $r$) satisfying Def. \ref{def:rate_dist_pair} for all $j\in \{1,\ldots,r-1\}$ as follows: Apply the code $\left(R_{(b\cdot r)},b\cdot r\right)$ to the first $b\cdot r$ samples of $S[i]$ and then concatenate each codeword by $j$ zeros to obtain a source code having codewords of length $b\cdot r+j$. The average distortion (i.e., see \eqref{eqn:distortionn}) of the resulting $\left(R_{(b\cdot r+j)},b\cdot r+j\right)$ code is given by:
\begin{align}
\label{eqn:distrtn_arb_lngt}
\bar{d}\left(\myVec{S}^{(b\cdot r +j)},\hat{\myVec{S}}^{(b\cdot r +j)}\right)&=\frac{1}{b\cdot r+j}\left(\sum\limits_{i=1}^{b\cdot r}\E\left\{\left(S[i]-\hat{S}[i]\right)^2\right\}+\sum\limits_{i=b\cdot r+1}^{b\cdot r+j}\E\left\{\left(S[i]\right)^2\right\}\right)\nonumber \\
&=\frac{1}{b\cdot r+j}\left(b\cdot r\cdot \bar{d}\left(\myVec{S}^{(b\cdot r)},\hat{\myVec{S}}^{(b\cdot r)}\right) + \sum\limits_{i=1}^{j} \sigma_S^2[i] \right)\nonumber
\\
&=\frac{b\cdot r}{b\cdot r+j}\cdot \bar{d}\left(\myVec{S}^{(b\cdot r)},\hat{\myVec{S}}^{(b\cdot r)}\right)+\frac{1}{b\cdot r+j}\sum\limits_{i=1}^{j} \sigma_S^2[i].
\end{align}
Thus $\exists b>b_o$ such that $\frac{1}{b\cdot r+j}\sum\limits_{i=1}^{j} \sigma_S^2[i]<\frac{\eta}{2}$ and
\begin{align}
   \bar{d}\left(\myVec{S}^{(b\cdot r+j)},\hat{\myVec{S}}^{(b\cdot r+j)}\right) 
    &=\frac{b\cdot r}{b\cdot r+j}\cdot \bar{d}\left(\myVec{S}^{(b\cdot r)},\hat{\myVec{S}}^{(b\cdot r)}\right)+\frac{1}{b\cdot r+j}\sum\limits_{i=1}^{j} \sigma_S^2[i] \nonumber\\
    &\leq \frac{b\cdot r}{b\cdot r+j}\cdot \bar{d}\left(\myVec{S}^{(b\cdot r)},\hat{\myVec{S}}^{(b\cdot r)}\right)+\frac{\eta}{2} \nonumber\\
    &\leq \bar{d}\left(\myVec{S}^{(b\cdot r)},\hat{\myVec{S}}^{(b\cdot r)}\right)+\frac{\eta}{2} \leq D+\eta.
\end{align}
The rate $R_{(b\cdot r+j)}$ satisfies:
\begin{equation}
\label{eqn:rate_arb_lngt}
    R_{(b\cdot r+j)}=\frac{1}{b\cdot r +j}\cdot \log_2 M = R_{(b\cdot r)}\cdot \frac{b\cdot r}{b\cdot r +j} \leq \left( R+\eta\right)\cdot \frac{b\cdot r}{b\cdot r +j} \leq R+\eta.
\end{equation}
 Consequently, any rate-distortion pair achievable with codes whose blocklength is an integer multiple of $r$ can be achieved by codes with arbitrary blocklengths.

Next, we prove that any achievable rate-distortion pair $(R,D)$ can be achieved by codes whose blocklength is an integer multiple of $r$. To that aim, we fix $\eta>0$. By Def. \ref{def:rate_dist_pair}, it holds that there exists a code of blocklength $l$ satisfying \eqref{eqn:rate_bound}-\eqref{eqn:dist_bound}. To show that $(R,D)$ is achievable using codes whose blocklength is an integer multiple of $r$, we assume here that $l$ is not an integer multiple of $r$, hence, there exist some positive integers $b$ and $j$ such that $j<r$ and $l = b\cdot r +j$. We denote this code by  $\left(R_{(b\cdot r+j)},b\cdot r+j\right)$. It follows from Def. \ref{def:rate_dist_pair} that $R_{(b\cdot r+j)}\leq R+\eta$ and $\bar{d}\left(\myVec{S}^{(b\cdot r+j)},\hat{\myVec{S}}^{(b\cdot r+j)}\right)\leq D+\frac{\eta}{2}$. 
%
Next, we construct a code $\left(R_{(b+1)\cdot r}, (b+1)\cdot r\right)$ with codewords whose length is $(b+1)\cdot r$, i.e., an integer multiple of $r$, by adding $r-j$ zeros at the end of each codeword of the code $\left(R_{(b\cdot r+j)},b\cdot r+j\right)$. The average distortion can now be computed as follows:
\begin{align}
\label{eqn:distrtn_integer_blngt}
\bar{d}\left(\myVec{S}^{((b+1)\cdot r)},\hat{\myVec{S}}^{((b+1)\cdot r)}\right)&=\frac{1}{(b+1)\cdot r}\left(\sum\limits_{i=1}^{b\cdot r+j}\E\left\{\left(S[i]-\hat{S}[i]\right)^2\right\}+\sum\limits_{i=b\cdot r+j+1}^{(b+1)\cdot r}\E\left\{\left(S[i]\right)^2\right\}\right)\nonumber \\
&=\frac{1}{(b+1)\cdot r}\left((b\cdot r+j)\cdot \bar{d}\left(\myVec{S}^{(b\cdot r+j)},\hat{\myVec{S}}^{(b\cdot r+j)}\right) + \sum\limits_{i=b\cdot r+j+1}^{(b+1)\cdot r} \sigma_S^2[i] \right)\nonumber
\\
&=\frac{b\cdot r+j}{(b+1)\cdot r}\cdot \bar{d}\left(\myVec{S}^{(b\cdot r+j)},\hat{\myVec{S}}^{(b\cdot r+j)}\right)+\frac{\sum\limits_{i=b\cdot r+j+1}^{(b+1)\cdot r} \sigma_S^2[i]}{(b+1)\cdot r}, 
\end{align}
and again $\exists b>b_o$ such that $\frac{\sum\limits_{i=b\cdot r+j+1}^{(b+1)\cdot r} \sigma_S^2[i]}{(b+1)\cdot r}<\frac{\eta}{2}$, hence 
\begin{align}
    \bar{d}\left(\myVec{S}^{((b+1)\cdot r)},\hat{\myVec{S}}^{((b+1)\cdot r}\right)
    &\leq  \frac{b\cdot r+j}{(b+1)\cdot r}\cdot \bar{d}\left(\myVec{S}^{(b\cdot r+j)},\hat{\myVec{S}}^{(b\cdot r+j)}\right) +\frac{\eta}{2} \notag \\
    &\leq  \bar{d}\left(\myVec{S}^{(b\cdot r+j)},\hat{\myVec{S}}^{(b\cdot r+j)}\right) + \frac{\eta}{2} \leq D + \eta.
\end{align}
 The rate $R_{(b+1)\cdot r}$ can be expressed as follows:
 \begin{equation}
\label{eqn:rate_integer_blngt}
    R_{(b+1)\cdot r}=\frac{1}{(b+1)\cdot r}\cdot \log_2 M = R_{(b\cdot r+j)}\cdot \frac{b\cdot r+j}{(b+1)\cdot r} \leq  \left( R+\eta\right)\cdot \frac{b\cdot r+j}{(b+1)\cdot r}<R+\eta.
\end{equation}
It follows that $ R_{(b+1)\cdot r}\leq  R+\eta$ for any arbitrary $\eta$ by selecting a sufficiently large $b$. This proves that every rate-distortion pair achievable with arbitrary blocklengths (e.g., $l=b\cdot r+ j, j<r$) is also achievable when considering source codes whose blocklength is an integer multiple of $r$ (e.g., $l=b\cdot r$). This concludes the proof.
\end{proof}

\section{Proof of Theorem \ref{thm:ThmScale1}}
\label{app:Proofthmm}
\vspace{-0.1cm}
\color{black}
Recall that $\alpha \in \mR^{++}$. To prove the theorem, we fix a rate-distortion pair $(R,D)$ that is achievable for the source $\{{S}[i]\}_{i\in\mN}$. By Def. \ref{def:rate_dist_pair} this implies that for all $\eta > 0$ there exists $l_0(\eta) \in \mN$ such that for all $l > l_0(\eta)$ there exists a source code $\mySet{C}_l$ with rate $R_l \leq R + \eta$ and \ac{mse} distortion $D_l = \E\big\{\frac{1}{l} \big\|\myVec{S}^{(l)} - \hat{\myVec{S}}^{(l)}\big\|^2 \big\}  \leq D +\eta$.  Next, we use the code $\mySet{C}_l$ to define the source code $\mySet{C}_{l}^{(\alpha)}$, which operates in the following manner: The encoder  first scales its input block by $1/\alpha$.  Then, the block is encoded using the source code $\mySet{C}_{l}$. Finally, the selected codeword is scaled by $\alpha$.
Since the $\mySet{C}_{l}^{(\alpha)}$ has the same number of codewords and the same blocklength as $\mySet{C}_{l}$, it follows that its rate, denote $R_l^{(\alpha)}$, satisfied $R_l^{(\alpha)} = R_l \leq R +\eta$. Furthermore, by the construction of  $\mySet{C}_{l}^{(\alpha)}$, it holds that its reproduction vector when applied to $\alpha \cdot \myVec{S}^{(l)}$ is equal to the output of $\mySet{C}_{l}$ applied to $\myVec{S}^{(l)}$ scaled by $\alpha$, i.e., $\alpha \cdot \hat{\myVec{S}}^{(l)}$. 
Consequently, the \ac{mse} of  $\mySet{C}_{l}^{(\alpha)}$ when applied to the source $\{\alpha \cdot {S}[i]\}_{i\in\mN}$, denoted $D_l^{(\alpha)}$, satisfies $D_l^{(\alpha)} =  \E\big\{\frac{1}{l} \big\|\alpha \cdot \myVec{S}^{(l)} - \alpha \cdot \hat{\myVec{S}}^{(l)}\big\|^2 \big\} = \alpha^2 D_l \leq \alpha^2 D + \alpha^2 \eta$.  
 
 It thus follows that for all $\tilde{\eta} > 0$ there exists $\tilde{l}_0(\tilde{\eta}) = l_0\big(\min(\tilde{\eta}, \alpha^2\tilde{\eta})\big)$ such that for all $l > \tilde{l}_0(\tilde{\eta}) $ there exists a code $\mySet{C}_{l}^{(\alpha)}$ with rate $R_l^{(\alpha)} \leq R +\tilde{\eta}$ which achieves an \ac{mse} distortion of $D_l^{(\alpha)}\leq \alpha^2\cdot D + \tilde{\eta}$ when applied to the compression of $\{\alpha\cdot {S}[i]\}_{i \in \mN}$. 
 Hence, $(R, \alpha^2 D)$ is achievable for compression of $\{\alpha \cdot {S}[i]\}_{i \in \mN}$ by Def. \ref{def:rate_dist_pair}, proving the theorem.

\color{black}

\section{Proof of Theorem \ref{thm:plim}}
\label{app:Proof1}
\vspace{-0.1cm}
In this appendix, we prove \eqref{eqn:plimb} by applying a similar approach as used for proving \eqref{eqn:plima} in \cite[Appendix A]{shlezinger2019capacity}. We first note that Def. \ref{def:plimsup} can also be written as follows:
\begin{equation}
\label{eqn:appendix_pliminf}
    {\rm p-}\mathop{\lim \sup}\limits_{k \rightarrow \infty} \zk{k} \!\stackrel{(a)}{=} \! \inf\left\{\beta  \in \mySet{R} \Big|| \mathop{\lim\sup}\limits_{k \rightarrow \infty}\Pr \left(\zk{k} > \beta \right) = 0   \right\} \!\stackrel{(b)}{=} \! \inf\left\{\beta  \in \mySet{R} \Big|| \mathop{\lim\inf}\limits_{k \rightarrow \infty} \Fkeps{k} (\beta)= 1  \right\}.
\end{equation}
For the equality $(a)$, we note that the set of probabilities $\{\Pr \left(\zkeps {k} > \beta\right)\}_{k\in \mN}$ is non-negative and bounded in $[0,1]$; hence, for any $\beta \in \mR$ for which $\mathop{\lim\sup}\limits_{k \rightarrow \infty}\Pr \left(\zk{k} > \beta \right)=0$, it also holds from \cite[Thm. 3.17]{rudin1976principles} that the limit of any subsequence of $\left\{\Pr \left(\zk{k} > \beta \right)\right\}_{k\in \mN}$ is also $0$, since non-negativity of the probability implies $\mathop{\lim \inf}\limits_{k\rightarrow\infty} \Pr \left(\zk{k} > \beta \right) \geq 0$. Then, combined with the relationship $\mathop{\lim \inf}\limits_{k\rightarrow\infty} \Pr \left(\zk{k} > \beta \right) \leq \mathop{\lim\sup}\limits_{k \rightarrow \infty}\Pr \left(\zk{k} > \beta \right)$, we conclude:
\begin{align*}
  &0\leq \mathop{\lim\inf}\limits_{k \rightarrow \infty}\Pr \left(\zk{k} > \beta \right)\leq \mathop{\lim\sup}\limits_{k \rightarrow \infty}\Pr \left(\zk{k} > \beta \right)=0 \notag \\
  &\implies \mathop{\lim\inf}\limits_{k \rightarrow \infty}\Pr \left(\zk{k} > \beta \right)=\mathop{\lim\sup}\limits_{k \rightarrow \infty}\Pr \left(\zk{k} > \beta \right)\stackrel{(a)}{=}\mathop{\lim}\limits_{k \rightarrow \infty}\Pr \left(\zk{k} > \beta \right)=0,  
\end{align*}
where $(a)$ follows from \cite[Example. 3.18(c)]{rudin1976principles}.
This implies $\mathop{\lim}\limits_{k \rightarrow \infty}\Pr \left(\zk{k} > \beta \right)$ exists and is equal to 0. 

In the opposite direction, if $\mathop{\lim}\limits_{k \rightarrow \infty}\Pr \left(\zk{k} > \beta \right)=0$ then it follows from \cite[Example. 3.18(c)]{rudin1976principles} that $\mathop{\lim\sup}\limits_{k \rightarrow \infty}\Pr \left(\zk{k} > \beta \right)=0$. Next, we note that since $\Fkeps{k} (\beta)$ is bounded in $[0,1]$ then $\mathop{\lim\inf}\limits_{k \rightarrow \infty} \Fkeps{k} (\beta)$ is finite $\forall \beta \in \mR$, even if $\mathop{\lim}\limits_{k \rightarrow \infty} \Fkeps{k} (\beta)$ does not exist. Equality $(b)$ follows since $\mathop{\lim \sup}\limits_{k \rightarrow \infty}\Pr \left(\zk{k} > \beta \right)= \mathop{\lim\sup}\limits_{k \rightarrow \infty}\left(1-\Pr \left(\zk{k} \leq \beta \right)\right)$ which according to \cite[Thm. 7.3.7]{dixmier2013general} is equal to $1+ \mathop{\lim\sup}\limits_{k \rightarrow \infty}\left(-\Pr\left(\zk{k} \leq \beta \right)\right)$. By \cite[Ch. 1, page 29]{stein2009real}, this quantity is also equal to $1- \mathop{\lim\inf}\limits_{k \rightarrow \infty}\left(\Pr\left(\zk{k} \leq \beta \right)\right)=1- \mathop{\lim\inf}\limits_{k \rightarrow \infty}\Fkeps{k} (\beta)$. 

Next, we state the following lemma:

\begin{lemma}
			\label{lem:AidLemma2}
			Given assumption \ref{itm:assm2},  for all $\beta \in \mySet{R}$ it holds that
			\begin{equation}
			\label{eqn:AidLemma2}
			\mathop{\lim \inf}\limits_{k \rightarrow \infty} \Fkeps{k}(\beta) = \mathop{\lim }\limits_{n \rightarrow \infty}\mathop{\lim\inf}\limits_{k \rightarrow \infty} \Fkn{k}{n}(\beta).
			\end{equation}  
		\end{lemma}
		
		\begin{proof}
			To prove the lemma we first show that $\mathop{\lim \inf}\limits_{k \rightarrow \infty} \Fkeps{k}(\beta) \le \mathop{\lim }\limits_{n \rightarrow \infty}\mathop{\lim \inf}\limits_{k \rightarrow \infty} \Fkn{k}{n}(\beta)$, and then we show 	$\mathop{\lim \inf}\limits_{k \rightarrow \infty} \Fkeps{k}(\beta) \ge \mathop{\lim }\limits_{n \rightarrow \infty}\mathop{\lim \inf}\limits_{k \rightarrow \infty} \Fkn{k}{n}(\beta)$.
			Recall that by \ref{itm:assm2}, for all $\beta \in \mySet{R}$ and $k \in \mySet{N}$, $ \Fkn{k}{n}(\beta)$ converges as $n \rightarrow \infty$ to $\Fkeps{k}(\beta)$, uniformly over $k$ and $\beta$, i.e., for all $\eta >  0$ there exists $n_0(\eta) \in \mySet{N}$, $k_0\big(n_0(\eta), \eta\big) \in \mySet{N}$ such that for every $n > n_0(\eta)$, $\beta \in \mySet{R}$ and  $k  > k_0\big(n_0(\eta), \eta\big)$, it holds that $\big|\Fkn{k}{n}(\beta) -  \Fkeps{k}(\beta)\big| < \eta$. 
			Consequently, for every subsequence $0<k_1< k_2<\ldots$ such that $\mathop{\lim}\limits_{l \rightarrow \infty} \Fkn{k_l}{n}(\beta)$  exists for any $n > n_0(\eta)$, it follows from \cite[Thm. 7.11]{rudin1976principles}\myFtn{\cite[Thm. 7.11]{rudin1976principles}: Suppose $f_n \rightarrow f$ uniformly in a  set $E$ in a metric space. Let $x$ be a limit point of $E$, and suppose that $\mathop{\lim}\limits_{t\rightarrow x} f_n(t)=A_n$, $(n=1,2,3,\ldots)$. Then $A_n$ converges, and $\mathop{\lim}\limits_{t \rightarrow x}\mathop{\lim}\limits_{n \rightarrow \infty} f_n(t)=\mathop{\lim}\limits_{n \rightarrow \infty}\mathop{\lim}\limits_{t \rightarrow x}f_n(t)$} that, as the convergence over $k$ is uniform, the limits over $n$ and $l$ are interchangeable:
			\begin{equation}
			\mathop{\lim}\limits_{n \rightarrow \infty}\mathop{\lim}\limits_{l \rightarrow \infty} \Fkn{k_l}{n}(\beta) 
			=\mathop{\lim}\limits_{l \rightarrow \infty}  \mathop{\lim}\limits_{n \rightarrow \infty}\Fkn{k_l}{n}(\beta)
			= \mathop{\lim}\limits_{l \rightarrow \infty}  \Fkeps{k_l}(\beta). 
			\label{eqn:AidLemma21}
			\end{equation}	
			
			The existence of such a convergent subsequence is guaranteed by the Bolzano-Weierstrass theorem \cite[Thm. 2.42]{rudin1976principles} \myFtn{\cite[Thm. 2.42]{rudin1976principles}:Every bounded infinite subset of $\mR^k$ has a limit point in $\mR^k$} as $\Fkn{k}{n}(\beta) \in [0,1]$.    
			
			From the properties of the limit inferior \cite[Thm. 3.17]{rudin1976principles}\myFtn{\cite[Thm. 3.17]{rudin1976principles}: Let $\{s_n\}$ be a sequence of real numbers; Let $E$ be the set of numbers $x$ (in the extended real number system) containing all limits of all subsequences of $\{s_n\}$. Then...... $\mathop{\lim\inf}\limits_{n \rightarrow \infty} s_n\in E$.} it follows that there exists a subsequence of $\big\{\Fkeps{k}(\beta)\big\}_{k \in \mySet{N}}$, denoted $\big\{\Fkeps{k_m}(\beta)\big\}_{m \in \mySet{N}}$, such that $\mathop{\lim}\limits_{m \rightarrow \infty}  \Fkeps{k_m}(\beta) = \mathop{\lim\inf}\limits_{k \rightarrow \infty}  \Fkeps{k}(\beta)$. Consequently, 
		
			\begin{align}
			\mathop{\lim\inf}\limits_{k \rightarrow \infty}  \Fkeps{k}(\beta) 
			&= \mathop{\lim}\limits_{m \rightarrow \infty}  \Fkeps{k_m}(\beta) 
			\stackrel{(a)}{=} \mathop{\lim}\limits_{n \rightarrow \infty}\mathop{\lim}\limits_{m \rightarrow \infty} \Fkn{k_m}{n}(\beta) \notag \\
			&\stackrel{(b)}{\geq} \mathop{\lim}\limits_{n \rightarrow \infty}\mathop{\lim\inf}\limits_{k \rightarrow \infty} \Fkn{k}{n}(\beta),
			\label{eqn:AidLemma23}
			\end{align}
			where $(a)$ follows from \eqref{eqn:AidLemma21}, and $(b)$ follows from the definition of the limit inferior \cite[Def. 3.16]{rudin1976principles}.  
			Similarly,  by  \cite[Thm. 3.17]{rudin1976principles}, for any $n \in \mySet{N}$  there exists a subsequence of $\{\Fkn{k}{n}(\beta)\}_{k\in \mN}$ which we denote by $\big\{\Fkn{k_l}{n}(\beta)\big\}_{l \in \mySet{N}}$ where  $\{k_l\}_{l \in \mySet{N}}$ satisfy  $0<k_1< k_2< \ldots$, 
			such that $\mathop{\lim}\limits_{l \rightarrow \infty}  \Fkn{k_l}{n}(\beta) = \mathop{\lim\inf}\limits_{k \rightarrow \infty}  \Fkn{k}{n}(\beta)$. Therefore,
			\begin{align}
			\mathop{\lim}\limits_{n \rightarrow \infty} \mathop{\lim\inf}\limits_{k \rightarrow \infty}  \Fkn{k}{n}(\beta)
			&=  \mathop{\lim}\limits_{n \rightarrow \infty} \mathop{\lim}\limits_{l \rightarrow \infty}  \Fkn{k_l}{n}(\beta)   \notag \\ 
			&\stackrel{(a)}{=} \mathop{\lim}\limits_{l \rightarrow \infty}  \Fkeps{k_l}(\beta)
			\stackrel{(b)}{\ge} \mathop{\lim\inf}\limits_{k \rightarrow \infty}  \Fkeps{k}(\beta),
			\label{eqn:AidLemma24}
			\end{align}
			where $(a)$ follows  
			from \eqref{eqn:AidLemma21}, and $(b)$ follows from the definition of the limit inferior \cite[Def. 3.16]{rudin1976principles}. 
			Therefore, $\mathop{\lim \inf}\limits_{k \rightarrow \infty} \Fkeps{k}(\beta) \le \mathop{\lim }\limits_{n \rightarrow \infty}\mathop{\lim \inf}\limits_{k \rightarrow \infty} \Fkn{k}{n}(\beta)$. Combining \eqref{eqn:AidLemma23} and \eqref{eqn:AidLemma24} proves \eqref{eqn:AidLemma2} in the statement of the lemma. 
		\end{proof}	
		
		\begin{lemma}
			\label{lem:AidLemma3} 
			Given assumptions \ref{itm:assm1}-\ref{itm:assm2},  the sequence of \acp{rv} $\big\{\zkn{k}{n} \big\}_{k,n \in \mySet{N}}$ satisfies
			\begin{align}
			\mathop{\lim }\limits_{n \rightarrow \infty}\left( {\rm p-}\mathop{\lim \sup}\limits_{k \rightarrow \infty} \zkn{k}{n} \right)  
			&=   \inf\left\{\beta \in \mySet{R} \Big|\mathop{\lim }\limits_{n \rightarrow \infty}\mathop{\lim \inf}\limits_{k \rightarrow \infty} \Fkn{k}{n}(\beta) =1   \right\}.
			\label{eqn:pliminfEq4}
			\end{align}
		\end{lemma}
		\begin{proof}
			Since by assumption \ref{itm:assm1}, for every $n \in \mySet{N}$,  every convergent subsequence of $\big\{\zkn{k}{n} \big\}_{k\in \mySet{N}}$ converges in distribution as $k \rightarrow \infty$ to a deterministic scalar, it follows that every convergent subsequence of $\Fkn{k}{n}(\beta)$ converges as $k \rightarrow \infty$ to a step function, which is the \ac{cdf} of the corresponding sublimit of $\zkn{k}{n}$. In particular, the limit
			$\liminfk\Fkn{k}{n}(\beta)$ is a step function representing the \ac{cdf} of the deterministic scalar $\zetan{n}$, i.e.,
			
			\begin{equation}
			\mathop{\lim \inf}\limits_{k \rightarrow \infty} \Fkn{k}{n}(\beta)=
			\begin{cases}
			0 & \beta <\zetan{n} \\
			1 & \beta \ge \zetan{n}.
			\end{cases}
			\label{eqn:IntProof1}
			\end{equation}
			
			Since, by Lemma \ref{lem:AidLemma2},  \ref{itm:assm2} implies that the limit $\mathop{\lim }\limits_{n \rightarrow \infty}\mathop{\lim \inf}\limits_{k \rightarrow \infty} \Fkn{k}{n}(\beta)$ exists\footnote{The convergence to a discontinuous function is in the sense of \cite[Ex. 7.3]{rudin1976principles}}, then $\mathop{\lim }\limits_{n \rightarrow \infty} \zetan{n} $ exists. Hence, we obtain that 
			
			\begin{equation}
				\mathop{\lim }\limits_{n \rightarrow \infty}\mathop{\lim \inf}\limits_{k \rightarrow \infty} \Fkn{k}{n}(\beta)=
			\begin{cases}
			0 & \beta < \mathop{\lim }\limits_{n \rightarrow \infty} \zetan{n} \\
			1 &  \beta \ge \mathop{\lim }\limits_{n \rightarrow \infty} \zetan{n},
			\end{cases}
			\end{equation}
			and from the right-hand side of \eqref{eqn:pliminfEq4} we have that 
			\begin{equation}
			\label{eqn:RHS_LEMMA_A2}
			    \inf\left\{\beta \in \mySet{R} \Big|\mathop{\lim }\limits_{n \rightarrow \infty}\mathop{\lim \inf}\limits_{k \rightarrow \infty} \Fkn{k}{n}(\beta) =1   \right\}=\mathop{\lim }\limits_{n \rightarrow \infty} \zetan{n}.
			\end{equation}
			Next, from \eqref{eqn:appendix_pliminf} and \eqref{eqn:IntProof1} we note that 
			\begin{align*}
			{\rm p-}\mathop{\lim \sup}\limits_{k \rightarrow \infty} \zkn{k}{n}
			= \inf\left\{\beta \in \mySet{R} \Big|\mathop{\lim \inf}\limits_{k \rightarrow \infty} \Fkn{k}{n}(\beta) =1   \right\} 	  =  \zetan{n}.
			\end{align*} 
			Consequently, the left-hand side of \eqref{eqn:pliminfEq4} is equal to $\mathop{\lim }\limits_{n \rightarrow \infty} \zetan{n} $. Combining with \eqref{eqn:RHS_LEMMA_A2} we arrive at the equality \eqref{eqn:pliminfEq4} in the statement of the lemma.  
		\end{proof}

		%
		
		Substituting \eqref{eqn:AidLemma2} into \eqref{eqn:appendix_pliminf} results in 
		\begin{align}
		{\rm p-}\mathop{\lim \sup}\limits_{k \rightarrow \infty} \zk{k}  
		&=\inf\left\{\beta\in \mySet{R} \Big|\mathop{\lim }\limits_{n \rightarrow \infty}\mathop{\lim \inf}\limits_{k \rightarrow \infty} \Fkn{k}{n}(\beta) =1   \right\} 
		\stackrel{(a)}{=}\mathop{\lim }\limits_{n \rightarrow \infty}\left( {\rm p-}\mathop{\lim \sup}\limits_{k \rightarrow \infty} \zkn{k}{n} \right), 
		\label{eqn:pliminfEq3}
		\end{align}
		where $(a)$ follows from \eqref{eqn:pliminfEq4}. Eq. \eqref{eqn:pliminfEq3} concludes the proof for \eqref{eqn:plimb}.
		
\section{Proof of Theorem \ref{Thm:rate_WSACS}}
\label{app:Proof2}
In this appendix we detail the proof of Theorem \ref{Thm:rate_WSACS}. The outline of the proof is given as follows:
\begin{itemize}
    \item We first show in Subsection \ref{app:proof2A} that for any 
    $k \in \mN$, 
    the \ac{pdf} of the random vector $\VecSn^{(k)}$, representing the first $k$ samples of the \ac{ct} \ac{wscs} source $\Sc(t)$ sampled at time instants $\Tsamp(n)=\frac{\Tc}{p+\myEpsn}$, converges in the limit as $n\rightarrow \infty$ and for any $k \in \mN$ to the \ac{pdf} of  $\VecSeps^{(k)}$, which represents the first $k$ samples of the \ac{ct} \ac{wscs} source $\Sc(t)$, sampled at time instants $\Tsamp=\frac{\Tc}{p+\myEps}$.  We prove that this convergence is uniform in $k \in \mN$ and in the realization vector $\myVec{s}^{(k)} \in \mR^k$. This is stated in Lemma \ref {Lem:PDF-convergence}.
    
    \smallskip  
    
    \item Next, in Subsection \ref{app:proof2b} we apply Theorem \ref{thm:plim} to relate the mutual information density rates for the random source vector $\VecSn^{(k)}$  and its reproduction $\Shnk$ with that of the random source vector  $\VecSeps^{(k)}$ and its reproduction $\Shepsk$. To that aim, let the functions $\cdf{\Sn,\Shn}$ and $ \cdf{\Seps,\Sheps}$ denote the joint distributions of an arbitrary dimensional source and reproduction vectors corresponding to the synchronously sampled and to the asynchronously sampled source process respectively. We define the following mutual information density rates:
    \begin{subequations}
    \label{eqn:zkdefs}
				\begin{equation}
				\label{eqn:zkdefs1}
				\zkn{k}{n}'\left( \cdf{\Sn,\Shn} \right)  \triangleq \frac{1}{k}\log \frac{\pdf{\VecSn^{(k)}|\Shnk} \left(\VecSn^{(k)}|\Shnk \right)}{\pdf{\VecSn^{(k)}}\left( \VecSn^{(k)} \right) },
				\end{equation}
				and
				\begin{equation}
				\label{eqn:zkdefs2}
				\zkeps{k,\eps}'\left( \cdf{\Seps,\Sheps} \right) \triangleq \frac{1}{k}\log \frac{\pdf{\VecSeps^{(k)}|\Shepsk}  \left( \VecSeps^{(k)}\big | \Shepsk  \right)}{\pdf{\VecSeps^{(k)}}\left( \Sepsk\right) },
				\end{equation}	
			\end{subequations}
			$k, n \in \mN$. The  \acp{rv} $\zkn{k}{n}'\left( \cdf{\Sn,\hat{S}_n} \right)$ and $\zkeps{k,\eps}'\left( \cdf{\Seps,\Sheps} \right)$ in \eqref{eqn:zkdefs} denote the mutual information density rates \cite[Def. 3.2.1]{han2003information} between the \ac{dt} source process and the corresponding reproduction process for the case of synchronous sampling and for the case of asynchronous sampling, respectively. 
			
			We then show that if the pairs of source process and optimal reproduction process  $\big\{\Sn[i],\Shn[i]\big\}_{i\in\mN}$ and $\big\{\Seps[i], \Sheps[i]\big\}_{i\in\mN}$
			satisfy that $\pdf{\Shnk}\left( \hat{\myVec{s}}^{(k)} \right) \Conv{n \rightarrow \infty} \pdf{\Shepsk}\left( \hat{\myVec{s}}^{(k)} \right)$ uniformly with respect to $\hat{\myVec{s}}^{(k)} \in \mR^k$ and $k \in \mN$, and that $\pdf{\Snk | \Shnk }\left(\myVec{s}^{(k)}  \big| \hat{\myVec{s}}^{(k)} \right)\Conv{n\rightarrow\infty}\pdf{\Sepsk | \Shepsk }\left( \myVec{s}^{(k)}  \big| \hat{\myVec{s}}^{(k)} \right)$ uniformly in $\left(\big(\hat{\myVec{s}}^{(k)}\big)^T, \big(\myVec{s}^{(k)}\big)^T \right)^T \in \mR^{2k}$ and $k \in \mySet{N}$, then 
			$\zkn{k}{n}'\left( \cdf{\Sn,\Shn} \right) \ConvDist{n \rightarrow \infty} \zkeps{k,\eps}'\left( \cdf{\Seps,\Sheps} \right)$ uniformly in $k \in \mySet{N}$. 
			In addition,  Lemma  \ref{lem:AsyncZk2} proves that every subsequence of $\left\{\zkn{k}{n}'\left( \cdf{\Sn,\Shn} \right)\right\}_{k\in \mN}$ $w.r.t.$ $k$, indexed as $k_l$ converges in distribution, in the limit $l \rightarrow \infty$ to a  deterministic scalar. 
			
			\smallskip
			
            \item Lastly, in Subsection \ref{app:proof2c} we combine the above results to show in Lemmas \ref{lem:rdfUpBound} and \ref{lem:rdflowBound} that $R_\eps (D) \leq \mathop{\lim \sup\limits}\limits_{n\rightarrow\infty} R_n (D)$ and $R_\eps (D) \geq \mathop{\lim \sup}\limits_{n\rightarrow\infty} R_n (D)$ respectively; implying that $R_\eps (D) = \mathop{\lim \sup}\limits_{n\rightarrow\infty} R_n (D)$, which proves the theorem.
\end{itemize}

To facilitate our proof we will need uniform convergence in $k\in \mN$, of $\pdf{\Snk}\left( \myVec{s}^{(k)} \right)$, $\pdf{\Shnk}\left( \hat{\myVec{s}}^{(k)} \right)$ and
$\pdf{\Snk | \Shnk }\left(\myVec{s}^{(k)}  \big| \hat{\myVec{s}}^{(k)} \right)$ to $\pdf{\Sepsk}\left( \myVec{s}^{(k)} \right)$, $\pdf{\Shepsk}\left( \hat{\myVec{s}}^{(k)} \right)$ and
$\pdf{\Sepsk | \Shepsk }\left(\myVec{s}^{(k)}  \big| \hat{\myVec{s}}^{(k)} \right)$, respectively. To that aim, we will make the following scaling assumption  w.l.o.g.:

\begin{assumption}
\label{Assumption_unif_conv}  The variance of the source and the allowed distortion are scaled by some factor $\alpha^2$ such that
\begin{equation}
    \alpha^2\cdot\min\left\{D,\left(\mathop{\min}\limits_{0\leq t\leq\Tc}\Csc(t)-D\right)\right\}>\frac{1}{2\pi}.
\end{equation}
\end{assumption}

 Note that this assumption has no effect on the generality of the \ac{rdf} for  multivariate stationary processes detailed in \cite[Sec. 10.3.3]{cover2006elements}, \cite[Sec. IV]{kolmogorov1956shannon}. 
{
 Moreover, by Theorem \ref{thm:ThmScale1}, for every $\alpha > 0$ it holds that when
any rate $R$ achievable when compressing the original source $\Sc(t)$ with distortion not larger that $D$ is achievable when compressing the scaled source $\alpha\cdot \Sc(t)$ with distortion not larger than $\alpha^2\cdot D$.} Note that if for the source $\Sc(t)$ the distortion satisfies  $D< \mathop{\min}\limits_{0\leq t\leq \Tc}\Csc(t)$, then for the scaled source and  distortion we have $\alpha^2 \cdot D< \mathop{\min}\limits_{0\leq t\leq \Tc}\alpha^2 \cdot\Csc(t)$.

\subsection{Convergence in Distribution of \mathinhead{\VecSn^{(k)}}{Sn} to  \mathinhead{\VecSeps^{(k)}}{Seps} Uniformly with Respect to \mathinhead{k\in \mySet{N}}{set}}
\label{app:proof2A}
In order to prove the uniform convergence in distribution, $\VecSn^{(k)} \ConvDist{n\rightarrow\infty} \VecSeps^{(k)}$, uniformly with respect to $ k\in \mN$, we first prove, in Lemma \ref{Lem:PDF-convergence}, that as $n\rightarrow\infty$ the sequence of \acp{pdf} of $\Snk$, $\pdf{\Snk}\left(\myVec{s}^{(k)}\right)$, converges to the \ac{pdf} of $\Sepsk$, $\pdf{\Sepsk}\left(\myVec{s}^{(k)}\right)$, uniformly in $\myVec{s}^{(k)} \in \mR^k$ and in  $ k \in \mN$. Next, we show in Corollary \ref{lem: convDist} that $\Snk \ConvDist{n\rightarrow\infty} \Sepsk $ uniformly in $k\in \mN$.

To that aim, let us define the set $\mK \triangleq \{1,2,\ldots,k\}$ and consider the $k$- dimensional zero-mean, memoryless random vectors $\Snk$ and $\Sepsk$ with their respective diagonal correlation matrices expressed below: 
\begin{subequations}
\begin{equation}
    \mathsf{R}_n^{(k)}\triangleq\E\big\{\big(\Snk\big)\big(\Snk\big)^T\big\}=\textrm{diag}\big(\Sigsn[1], \ldots, \Sigsn[k]\big),
\end{equation}
     \begin{equation}
     \mathsf{R}_{\eps}^{(k)}\triangleq\E\big\{\big(\Sepsk\big)\big(\Sepsk\big)^T\big\}=\textrm{diag}\big(\Sigseps[1], \ldots, \Sigseps[k]\big).
\end{equation}
\end{subequations}
Since $\myEps_n \triangleq \frac{\lfloor n \cdot \myEps \rfloor}{n}$ it holds that $\frac{n\cdot \myEps-1}{n} \leq \myEps_n \leq\frac{n\cdot \myEps}{n}$; therefore
\begin{equation}
\label{Eqn:limit_epsilon_n}
     \mathop{\lim}\limits_{n\rightarrow \infty}\myEps_n=\myEps.
\end{equation}
 Now we note that since $\sigma_{\Sc}^{2}(t)$ is uniformly continuous, then by the definition of a uniformly continuous function, for each $i\in \mN$, the limit in \eqref{Eqn:limit_epsilon_n} implies that
 \begin{equation}
 \label{eqn:limit_sigmas_n}
       \mathop{\lim}\limits_{n\rightarrow \infty} \Sigsn[i]\equiv  \mathop{\lim}\limits_{n\rightarrow \infty} \Csc\left(i\cdot \frac{\Tc}{p+\eps_n} \right) =  \Csc\left(i\cdot \frac{\Tc}{p+\eps} \right) \equiv \Sigseps[i].
 \end{equation}

From Assumption \ref{Assumption_unif_conv}, it follows that $\Sigsn[i]$ satisfies $\Sigsn[i]>\frac{1}{2\pi}$; Hence, we can state the following lemma:

\begin{lemma}
\label{Lem:PDF-convergence}
The \ac{pdf} of $\Snk$, $\pdf{\Snk}\left(\myVec{s}^{(k)}\right)$, converges as $n\rightarrow\infty$ to the \ac{pdf} of $\Sepsk$, $\pdf{\Sepsk}\left(\myVec{s}^{(k)}\right)$, uniformly in $\myVec{s}^{(k)}\in \mR^k$ and in $k\in \mN$: 
\begin{equation*}
    \mathop{\lim}\limits_{n\rightarrow \infty} \pdf{\Snk}\left(\myVec{s}^{(k)}\right)=\pdf{\Sepsk}\left(\myVec{s}^{(k)}\right), \quad \forall \myVec{s}^{(k)} \in \mR^k, \forall k\in \mN.
\end{equation*}
\end{lemma}
\begin{proof}
The proof of the lemma directly follows from the steps in the proof of \cite[Lemma B.1]{shlezinger2019capacity}\myFtn{\cite[Lemma B.1]{shlezinger2019capacity} considers a \ac{ct} memoryless \ac{wscs} Gaussian noise process $\Wi_c(t)$ sampled synchronously and asynchronously to yield $\Wi_{n}^{(k)}$ and $\Wi_{\myEps}^{(k)}$ respectively, both having independent entries. This Lemma proves that, for both Gaussian vectors of independent entries, if the variance of $\Wi_{n}^{(k)}$ converges as $n\rightarrow\infty$ to the variance of $\Wi_{\myEps}^{(k)}$, then the PDF (which in this case is a continuous mapping of the variance) of $\Wi_{n}^{(k)}$ also converges to the PDF of $\Wi_{\myEps}^{(k)}$. Also, for simplicity, the assumption $\frac{1}{2\pi} < \Cwc(t) < \infty $ for all $t \in \mySet{R}$ was used to prove uniform convergence of the PDF of $\Wi_{n}^{(k)}$. A similar setup is applied in this work with $\frac{1}{2\pi}<\sigma^2_{S_c}(t)<\infty$ for all $t\in \mR$, for the memoryless \ac{ct} \ac{wscs} process $\Sc(t)$; $\Snk$ and $\Sepsk$ also have independent entries, for the synchronously and for the asynchronously sampled case respectively. The proof for uniform convergence in $k \in \mN$ from \cite[Lemma B.1]{shlezinger2019capacity} also applies to this current work.}, which was applied for a Gaussian  noise process with independent entries and variance above $\frac{1}{2\pi}$. 
\end{proof}
Lemma \ref{Lem:PDF-convergence} gives rise to the following corollary:
\begin{corollary}
\label{lem: convDist}
	For any $k \in \mySet{N}$ it holds that $\Snk \ConvDist{n \rightarrow \infty} \Sepsk$, and convergence is uniform over $k$.
\end{corollary}
\begin{proof}
The corollary holds due to \cite[Thm.1]{scheffe1947useful}\myFtn{\cite[Thm.1]{scheffe1947useful}\label{fn:convpdf_imply_dist}: If, for a sequence $\{ p_n(x)\}_{n\in \mN}$ of densities, $ \mathop{\lim}\limits_{n\rightarrow \infty} p_n(x) = p(x)$ for almost all $x$ in $\mR$, then a sufficient condition that $ \mathop{\lim}\limits_{n\rightarrow \infty} \mathop{\int}\limits_{\mS} p_n(x) dx = \mathop{\int}\limits_{\mS}p(x) dx$, uniformly for all Borel sets $\mS$ in $\mR$, is that $p(x)$ be a density.}: Since $\pdf{\Snk}\left(\myVec{s}^{(k)}\right)$ converges to $\pdf{\Sepsk}\left(\myVec{s}^{(k)}\right)$ then $\Snk \ConvDist{n \rightarrow \infty} \Sepsk$. In addition, since the convergence of the \acp{pdf} is uniform in $k \in \mN$, the convergence of the \acp{cdf} is also uniform in $k \in \mN$.
\end{proof}

\subsection{Showing that \mathinhead{\zkn{k}{n}'\left( \cdf{\Sn,\Shn}\opt \right)}{zkn} and \mathinhead{\zkeps{k,\eps}'\left( \cdf{\Seps,\Sheps} \right)}{zk} Satisfy the Conditions of Thm. \ref{thm:plim}}
\label{app:proof2b}

	Let $\cdf{\Sn,\Shn}\opt$ denote the joint distribution for the source process and the corresponding optimal reproduction process satisfying the distortion constraint $D$.
		We next prove that for $\cdf{\Sn,\Shn}\opt \ConvDist{n\rightarrow\infty}\cdf{\Seps,\Sheps}$, then $\zkn{k}{n}'\left( \cdf{\Sn,\Shn}\opt \right)$ and $\zkeps{k,\eps}'\left( \cdf{\Seps,\Sheps} \right)$ satisfy \ref{itm:assm1}-\ref{itm:assm2}. In particular, in Lemma \ref{lem:AsyncZk} we prove  that $\tilde{Z}_{k,n}'\left( \cdf{\Sn,\Shn}\opt\right)  \ConvDist{n \rightarrow \infty} \zkeps{k,\eps}'\left(\cdf{\Seps,\Sheps}\right)$ uniformly in $k \in \mySet{N}$ for  the optimal zero-mean Gaussian reproduction vectors with independent entries. Lemma \ref{lem:AsyncZk2} proves that for any fixed $n$, $\zkn{k}{n}'\left( \cdf{\Sn,\Shn}\opt \right)$ converges in distribution	to a deterministic scalar as  $k \rightarrow \infty$. 
		\begin{lemma}
			\label{lem:AsyncZk}
			Let $\{\Shnk\}_{n \in \mySet{N}}$ and $\{\Wnk\}_{n \in \mySet{N}}$ be two sets of mutually independent sequences of $k\times 1$ zero-mean Gaussian random vectors  related via the backward channel \eqref{eqn:tst_chan_DCD_vct}, each having independent entries and let \acp{pdf} $\pdf{\Shnk}\left( \hat{\myVec{s}}^{(k)}\right)$ and $\pdf{\Wnk}\left( \myVec{w}^{(k)}\right)$, respectively, denote their \acp{pdf}. Consider two other zero-mean Gaussian random vectors $\Shepsk$ and $\Wepsk$ each having independent entries with the \acp{pdf} $\pdf{\Shepsk}\left( \hat{\myVec{s}}^{(k)} \right)$ and $\pdf{\Wepsk}\left( \myVec{w}^{(k)} \right)$, respectively, such that  $\mathop{\lim}\limits_{n \rightarrow \infty}\pdf{\Shnk}\left( \hat{\myVec{s}}^{(k)} \right)= \pdf{\Shepsk}\left( \hat{\myVec{s}}^{(k)} \right)$ uniformly in $\hat{\myVec{s}}^{(k)}\in \mR^{k}$ and uniformly with respect to $k\in \mN$, and $\mathop{\lim}\limits_{n \rightarrow \infty}\pdf{\Wnk}\left( \myVec{w}^{(k)} \right)= \pdf{\Wepsk}\left( \myVec{w}^{(k)} \right)$ uniformly in $\myVec{w}^{(k)} \in \mR^{k}$ and uniformly with respect to $k \in \mySet{N}$. 
			Then, the \acp{rv} $\zkn{k}{n}'\left( \cdf{\Sn, \Shn}\opt\right)$ and $\zkeps{k,\eps}'\left( \cdf{\Seps, \Sheps}\right)$, defined via \eqref{eqn:zkdefs} satisfy $\zkn{k}{n}'\left( \cdf{\Sn, \Shn}\opt\right)  \ConvDist{n \rightarrow \infty} \zkeps{k,\epsilon}'\left( \cdf{\Seps, \Sheps}\right)$ uniformly over $k \in \mySet{N}$.
		\end{lemma}

		\color{black}
		
		\begin{proof}
		To begin the proof, for $\left(\myVec{s}^{(k)}, \hat{\myVec{s}}^{(k)}\right) \in \mR^{2k}$, define
	\begin{equation}
			\label{eqn:fkdef}
			\fkn\left(\myVec{s}^{(k)}, \hat{\myVec{s}}^{(k)} \right)  \triangleq \frac{\pdf{\Snk | \Shnk}\left( \myVec{s}^{(k)} \big| \hat{\myVec{s}}^{(k)}\right) }{\pdf{\Snk}\left( \myVec{s}^{(k)} \right) },
			\qquad
			\fkeps\left(\myVec{s}^{(k)}, \hat{\myVec{s}} ^{(k)}\right)  \triangleq \frac{\pdf{\Sepsk | \Shepsk}\left( \myVec{s}^{(k)} \big| \hat{\myVec{s}}^{(k)}\right) }{\pdf{\Sepsk}\left( \myVec{s}^{(k)} \right)}.	
			\end{equation}

Now, we recall the backward channel relationship \eqref{eqn:tst_chan_DCD_vct}:
	    \begin{equation}
	    \label{eqn:testchannelrdf_synch2}
    \Snk=\Shnk+\Wnk,
\end{equation}
where $\Shnk$ and $\Wnk$ are mutually independent zero-mean, Gaussian random vectors with independent entries, corresponding to the optimal compression process and its respective distortion.
From this relationship we obtain
			\begin{align}
			\pdf{\Snk | \Shnk }\left( \myVec{s}^{(k)}  \big| \hat{\myVec{s}}^{(k)} \right)&\stackrel{(a)}{=}\pdf{\Shnk+\Wnk | \Shnk }\left( \myVec{s}^{(k)}  \big| \hat{\myVec{s}}^{(k)} \right)\nonumber \\
			&= \pdf{\Wnk | \Shnk }\left(\myVec{s}^{(k)}  -\hat{\myVec{s}}^{(k)}\big|\hat{\myVec{s}}^{(k)}\right)\nonumber\\
			&\stackrel{(b)}{=}\pdf{\Wnk }\left(\myVec{s}^{(k)}-\hat{\myVec{s}}^{(k)}\right),	\label{eqn:conditional_pdf}
			\end{align}	
			where $(a)$ follows since $\Snk = \Shnk + \Wnk$, see \eqref{eqn:testchannelrdf_synch2}, and $(b)$ follows since $\Wnk$ and $\Shnk$ are mutually independent. 
 The joint \ac{pdf} of $\Snk$ and $\Shnk$ can be expressed via the conditional \ac{pdf} as:
 \begin{align}
 \label{eqn:joint_pdf}
     \pdf{\Snk,\Shnk }\left( \myVec{s}^{(k)},\hat{\myVec{s}}^{(k)} \right)\!=\!\pdf{\Snk | \Shnk }\left( \myVec{s}^{(k)}  \big| \hat{\myVec{s}}^{(k)} \right)\cdot \pdf{\Shnk}\left(\hat{\myVec{s}}^{(k)}\right)\!\stackrel{(a)}{=}\!\pdf{\Wnk }\left(\myVec{s}^{(k)}\!-\!\hat{\myVec{s}}^{(k)}\right)\cdot \pdf{\Shnk}\left(\hat{\myVec{s}}^{(k)}\right),
 \end{align}
 where $(a)$ follows from \eqref{eqn:conditional_pdf}.
 Since $\Shnk$ and $\Wnk$ are Gaussian and mutually independent and since the product of two multivariate Gaussian \acp{pdf} is also a multivariate Gaussian \ac{pdf} \cite[Sec. 3]{bromiley2003products}, it follows from \eqref{eqn:joint_pdf} that $\Snk$ and $\Shnk$ are jointly Gaussian.
Following the mutual independence of $\Wnk$ and $\Shnk$, the \ac{rhs} of \eqref{eqn:joint_pdf} is also equivalent to the joint \ac{pdf} of $\left[\left(\Wnk\right)^T, \left(\Shnk\right)^T\right]^T$ denoted by $\pdf{\Wnk,\Shnk }\left(\myVec{s}^{(k)}-\hat{\myVec{s}}^{(k)},\hat{\myVec{s}}^{(k)} \right)$.
Now, from \eqref{eqn:conditional_pdf}, the assumption $\mathop{\lim}\limits_{n \rightarrow \infty}\pdf{\Wnk}\left( \myVec{w}^{(k)} \right)= \pdf{\Wepsk}\left( \myVec{w}^{(k)} \right)$ implies that a limit exists for the conditional \ac{pdf} $\pdf{\Snk\mid\Shnk }\left( \myVec{s}^{(k)}\mid\hat{\myVec{s}}^{(k)} \right)$, this we denote by $\pdf{\Sepsk\mid\Shepsk }\left( \myVec{s}^{(k)}\mid\hat{\myVec{s}}^{(k)} \right)$. Combining this with the assumption  $\mathop{\lim}\limits_{n \rightarrow \infty}\pdf{\Shnk}\left( \hat{\myVec{s}}^{(k)} \right)= \pdf{\Shepsk}\left( \hat{\myVec{s}}^{(k)} \right)$, we have that,

\begin{align}
\label{eqn:limit_joint_pdf}
   \mathop{\lim}\limits_{n \rightarrow \infty}\pdf{\Snk,\Shnk }\left(\myVec{s}^{(k)},\hat{\myVec{s}}^{(k)} \right)
   &=\mathop{\lim}\limits_{n \rightarrow \infty}\left(\pdf{\Snk | \Shnk }\left(\myVec{s}^{(k)}  \big| \hat{\myVec{s}}^{(k)} \right)\cdot \pdf{\Shnk}\left(\hat{\myVec{s}}^{(k)}\right)\right) \nonumber \\
   &\stackrel{(a)}{=} \mathop{\lim}\limits_{n \rightarrow \infty}\left( \pdf{\Wnk }\left(\myVec{s}^{(k)}-\hat{\myVec{s}}^{(k)}\right)\cdot \pdf{\Shnk}\left(\hat{\myVec{s}}^{(k)}\right)\right)\nonumber\\
   &\stackrel{(b)}{=}\mathop{\lim}\limits_{n \rightarrow \infty}\left( \pdf{\Wnk }\left(\myVec{s}^{(k)}-\hat{\myVec{s}}^{(k)}\right)\right)\cdot\mathop{\lim}\limits_{n \rightarrow \infty}\left( \pdf{\Shnk}\left(\hat{\myVec{s}}^{(k)}\right)\right)\nonumber\\
   &= \pdf{\Sepsk | \Shepsk }\left(\myVec{s}^{(k)}  \big| \hat{\myVec{s}}^{(k)}\right) \cdot \pdf{\Shepsk}\left(\hat{\myVec{s}}^{(k)}\right)\nonumber\\
   &= \pdf{\Sepsk,\Shepsk }\left(\myVec{s}^{(k)},\hat{\myVec{s}}^{(k)} \right),
\end{align}
where $(a)$ follows from \eqref{eqn:conditional_pdf}, and $(b)$ follows since the limit for each sequence in the product exists \cite[Thm. 3.3]{rudin1976principles}; Convergence is uniform in $\left(\left(\hat{\myVec{s}}^{(k)}\right)^T, \left(\myVec{s}^{(k)}\right)^T\right)^T \in \mR^{2k}$ and $k \in \mySet{N}$, as each sequence converges uniformly in $k \in \mySet{N}$ \cite[Page 165]{rudin1976principles}\myFtn{\cite[Page 165, Ex 2]{rudin1976principles}: The solution to this exercise shows that if two functions $\{f_n\}$ and $\{g_n\}$ converge uniformly on a set $E$ and both $\{f_n\}$ and $\{g_n\}$ are sequences of bounded functions then $\{f_ng_n\}$ converges uniformly on $E$.}.  Observe that
the joint \ac{pdf} for the zero-mean Gaussian random vectors $\left[\Snk, \Shnk\right]$ is given by the general expression:
	    \begin{equation}
	    \label{eqn:joint_pdf_gaus}
	         \pdf{\Snk,\Shnk }\!\left(\myVec{s}^{(k)},\hat{\myVec{s}}^{(k)} \right) \! =\!\Big(\Det\big(2\pi\tilde{\Cmat}_{n}^{(2k)}\big)\Big)^{-\frac{1}{2}} \! \exp\left(-\frac{1}{2}\left[\left(\hat{\myVec{s}}^{(k)}\right)^T\!\!, \left(\myVec{s}^{(k)}\right)^T\right] \big(\Cjmat_{n}^{(2k)}\big)^{-1}\!\left[\left(\hat{\myVec{s}}^{(k)}\right)^T\!\!, \left(\myVec{s}^{(k)}\right)^T\right]^T\right), 
	    \end{equation}
	   where $\Cjmat_{n}^{(2k)}$ denotes the joint covariance matrix of $\left[\left(\Shnk\right)^T,\left(\Snk\right)^T\right]^T$. From \eqref{eqn:joint_pdf_gaus} we note that $\pdf{\Snk,\Shnk }\left(\myVec{s}^{(k)},\hat{\myVec{s}}^{(k)} \right)$ is a continuous mapping of $\Cjmat_{n}^{(2k)}$ with respect to the index $n$, see \cite[Lemma B.1]{shlezinger2019capacity}. Hence the convergence in \eqref{eqn:limit_joint_pdf} of $\pdf{\Snk,\Shnk }\left(\myVec{s}^{(k)},\hat{\myVec{s}}^{(k)} \right)$ as $n\rightarrow \infty$ directly implies the convergence of $\Cjmat_{n}^{(2k)}$ as $n\rightarrow\infty$ to a limit which we denote by $\Cjmat_{\eps}^{(2k)}$. It therefore  follows that the limit function $\pdf{\Sepsk,\Shepsk}\left(\myVec{s}^{(k)},\hat{\myVec{s}}^{(k)} \right)$ corresponds to the \ac{pdf} of a Gaussian vector with the covariance matrix $\Cjmat_{\eps}^{(2k)}$.
The joint \ac{pdf} for the zero-mean Gaussian random vectors $\left[\Wnk, \Shnk\right]$ can be obtained using their mutual independence as: 
\begin{align}
	    \label{eqn:joint_pdf_gaus_ind}
	         &\pdf{\Wnk,\Shnk }\left(\myVec{s}^{(k)}-\hat{\myVec{s}}^{(k)},\hat{\myVec{s}}^{(k)} \right)\nonumber\\ &=\!\Big(\Det\big(2\pi\Sigma_{n}^{(2k)}\big)\Big)^{\frac{1}{2}}\! \exp\left(-\frac{1}{2}\left[\left(\myVec{s}^{(k)}\!-\!\hat{\myVec{s}}^{(k)}\right)^T, \left(\hat{\myVec{s}}^{(k)}\right)^T\right] \big(\Sigma_{n}^{(2k)}\big)^{-\!1}\!\left[\left(\myVec{s}^{(k)}\!-\!\hat{\myVec{s}}^{(k)}\right)^T, \left(\hat{\myVec{s}}^{(k)}\right)^T\right]^T\right), 
	    \end{align}
where $\Sigma_{n}^{(2k)}$ denotes the joint covariance matrix of $\left[\left(\Wnk\right)^T,\left(\Shnk\right)^T\right]^T$. Since the vectors $\Wnk$ and $\Shnk$ are zero-mean, mutually independent and, by the relationship \eqref{eqn:tst_chan_DCD_vct}, each vector has independent entries, it follows that $\Sigma_{n}^{(2k)}$ is a diagonal matrix with each diagonal element taking the value of the corresponding temporal variance at the respective index $i\in \{1,2,\ldots k\}$. i.e.,
\begin{align}
\label{covmatrix_ind_vectors}
    \Sigma_{n}^{(2k)}
    &\triangleq\E\left\{\left(\left(\Wnk\right)^T,\left(\Shnk\right)^T\right)^T\cdot\left(\left(\Wnk\right)^T,\left(\Shnk\right)^T\right)\right\}\nonumber\\
    &=\textrm{diag}\big(\E\left\{\left(\Wn[1]\right)^2\right\}, \E\left\{\left(\Wn[2]\right)^2\right\}, \ldots,
    \E\left\{\left(\Wn[k]\right)^2\right\},\Sigshn[1],\Sigshn[2]\ldots, \Sigshn[k]\big).
\end{align}
The convergence of $\pdf{\Wnk,\Shnk }\left(\myVec{s}^{(k)}-\hat{\myVec{s}}^{(k)},\hat{\myVec{s}}^{(k)} \right)$, from \eqref{eqn:limit_joint_pdf}, implies a convergence of the diagonal elements in \eqref{covmatrix_ind_vectors} as $n\rightarrow\infty$. Hence $\Sigma_{n}^{(2k)}$ converges as $n\rightarrow\infty$ to a diagonal joint covariance matrix which we denote by $\Sigma_{\eps}^{(2k)}$. This further implies that the limiting vectors $\Wepsk$ and $\Shepsk$ are zero-mean, mutually independent and each vector has independent entries in $i\in[1,2,\ldots, k]$.

Relationship \eqref{eqn:limit_joint_pdf} implies that the joint limit distribution satisfies $\pdf{\Sepsk,\Shepsk}\left(\myVec{s}^{(k)},\hat{\myVec{s}}^{(k)}\right)=\pdf{\Shepsk}\left(\hat{\myVec{s}}^{(k)}\right)\cdot \pdf{\Wepsk}\left(\myVec{s}^{(k)}-\hat{\myVec{s}}^{(k)}\right)$. Consequently, we can define an asymptotic backward channel that satisfies \eqref{eqn:limit_joint_pdf} via the expression:
\begin{equation}
    \label{eqn:testchannelrdf_asynchron}
	   \Sepsk[i]=\Shepsk[i]+\Wepsk[i]. 
	    \end{equation}

	   Next, by convergence of the joint \ac{pdf} $\pdf{\Wnk }\left(\myVec{s}^{(k)}-\hat{\myVec{s}}^{(k)}\right)\cdot \pdf{\Shnk}\left(\hat{\myVec{s}}^{(k)}\right)$ uniformly in $k \in \mN$ and in $\left(\left(\myVec{s}^{(k)}\right)^T, \left(\hat{\myVec{s}}^{(k)}\right)^T\right)^T\in \mR^{2k}$, it follows from \cite[Thm.1]{scheffe1947useful}\myFtn{Please refer to the footnote \ref{fn:convpdf_imply_dist} on page \pageref{fn:convpdf_imply_dist}.} that $ \left[\big( \Shnk\big) ^T , \big( \Wnk\big) ^T \right]^T \ConvDist{n \rightarrow \infty} \left[\big( \Shepsk\big) ^T , \big( \Wepsk\big) ^T \right]^T$  and the convergence is uniform in $k \in \mN$ and in $\left(\left(\myVec{s}^{(k)}\right)^T, \left(\hat{\myVec{s}}^{(k)}\right)^T\right)^T\in \mR^{2k}$.  Then, by the \ac{cmt} \cite[Thm. 7.7]{kosorok2008introduction}, we have 
	    \begin{align*}
	        \left[\big( \Snk\big) ^T , \big( \Shnk\big) ^T \right]^T\!=\! \left[\big( \Shnk\!+\!\Wnk\big) ^T , \big( \Shnk\big) ^T \right]^T \ConvDist{n \rightarrow \infty} \left[\big( \Shepsk\!+\!\Wepsk\big) ^T , \big( \Shepsk\big) ^T \right]^T
	        \!=\!\left[\big( \Sepsk\big) ^T , \big( \Shepsk\big) ^T \right]^T.
	    \end{align*}
%

	    Now, using the extended \ac{cmt} \cite[Thm. 7.24]{kosorok2008introduction}\myFtn{\label{fn:cmt} \cite[Thm. 7.24]{kosorok2008introduction}: (Extended continuous mapping). Let $\mathbb{D}_n \subset \mathbb{D}$ and $g_n$ :
$\mathbb{D}_n \mapsto \mathbb{E}$ satisfy the following: If $x_n\rightarrow x$ with $x_n \in \mathbb{D}_n$ for all $n \geq 1$ and
$x \in \mathbb{D}_0$, then $g_n(x_n) \rightarrow g(x)$, where $\mathbb{D}_0 \subset \mathbb{D}$ and $g : \mathbb{D}_0 \mapsto \mathbb{E}$. Let $X_n$ be
maps taking values in $\mathbb{D}_n$, and let $X$ be Borel measurable and separable.
Then
(i) $X_n \rightsquigarrow X$ implies $g_n(X_n) \rightsquigarrow g(X)$.
(ii) $X_n \mathop{\rightarrow}\limits^{P}
X$ implies $g_n(X_n)\mathop{\rightarrow}\limits^{P}g(X)$.
(iii) $X_n
\mathop{\rightarrow}\limits^{as*} X$ implies $g_n(X_n)\mathop{\rightarrow}\limits^{as*}g(X)$.}, we will show that 	$\fkn\big(\Snk , \Shnk  \big)  \ConvDist{n \rightarrow \infty}  \fkeps\big(\Sepsk , \Shepsk \big)$ for each $k \in \mySet{N}$, following the same approach of the proof for \cite[Lemma B.2]{shlezinger2019capacity} \myFtn{\cite[Lemma B.2]{shlezinger2019capacity}: Consider a sequence of $k\times 1$ zero-mean Gaussian random vectors with independent entries $\{\Xin_n^{(k)}\}_{n \in \mySet{N}}$ and a zero-mean Gaussian random vector  with independent entries $\Xin^{(k)}$, such that  $\Xin_n^{(k)} \ConvDist{n \rightarrow \infty} \Xin^{(k)}$ uniformly with respect to $k \in \mySet{N}$. Then, the \acp{rv} $\zkn{k}{n}'\left( \cdf{\Xnvec}\right)$ and $\zkeps{k}'\left( \cdf{\Xin}\right)$ defined in \cite[Eqn. (B.1)]{shlezinger2019capacity} satisfy $\zkn{k}{n}'\left( \cdf{\Xnvec}\right)  \ConvDist{n \rightarrow \infty} \zkeps{k}'\left( \cdf{\Xin}\right)$ uniformly over $k \in \mySet{N}$.}. Then, since $\zkn{k}{n}'\left( \cdf{\Sn,\Shn}\opt\right)=\frac{1}{k}\log \fkn\left(\Snk, \Shnk \right) $ and $\zk{k}'\left( \cdf{\Seps,\Sheps}\right)=\frac{1}{k}\log \fkeps\left(\Sepsk, \Shepsk \right)$, we conclude that $\zkn{k}{n}'\left( \cdf{\Sn,\Shn}\opt\right)\ConvDist{n \rightarrow \infty}\zk{k}'\left( \cdf{\Seps,\Sheps}\right)$, where it also follows from the proof of \cite[Lemma B.2]{shlezinger2019capacity} that the convergence is uniform in $k \in \mN$. Specifically, to prove that  $\fkn\left(\Snk , \Shnk  \right)  \ConvDist{n \rightarrow \infty}  \fkeps\left(\Sepsk , \Shepsk  \right)$, we will show that the following two properties hold:
			\begin{enumerate}[label={\em P\arabic*}]
				\item \label{itm:Tight} The distribution of $\left[\left( \Sepsk\right) ^T , \left( \Shepsk\right) ^T \right]^T$ is separable\myFtn{By \cite[Pg. 101]{kosorok2008introduction}, an \ac{rv} $X \in \mySet{X}$ is separable  if $\forall \eta > 0$ there exists a compact set $\mySet{K}(\eta) \subset \mySet{X}$ such that $\Pr \left( X \in \mySet{K}(\eta)\right) \ge 1 - \eta$.}.
				
				\item \label{itm:Conv} For any convergent sequence $\left(\left(\myVec{s}_n^{(k)}\right)^T, \left(\hat{\myVec{s}}_n^{(k)}\right)^T \right)^T \in \mySet{R}^{2k}$ such that $\mathop{\lim}\limits_{n \rightarrow \infty}\left(  \myVec{\myVec{s}}_n^{(k)},\hat{\myVec{s}}_n^{(k)}\right)  = \left(  \myVec{\myVec{s}}_\eps^{(k)}, \hat{\myVec{s}}_\eps^{(k)}\right) $, then $\mathop{\lim}\limits_{n \rightarrow \infty}  \fkn\left(\myVec{s}_n^{(k)}, \hat{\myVec{s}}_n^{(k)}\right) = \fkeps\left(\myVec{s}_\eps^{(k)}, \hat{\myVec{s}}_\eps^{(k)} \right) $.
			\end{enumerate}

To prove property \ref{itm:Tight}, we show that  ${ U}^{(k)} \triangleq \left[\big( \Sepsk\big) ^T , \big( \Shepsk\big) ^T \right]^T$ is \footnote{We point out that here, we misuse use the dimension notation as $U^{(k)}$ which denotes a $2k$ dimensional vector. Here, $k$ refers to the dimension of the compression problem and not of the vector.} separable \cite[Pg. 101]{kosorok2008introduction}, i.e., we show that $\forall \eta > 0$, there exists $\beta > 0$ such that $\Pr\left(\|U^{(k)}\|^2 > \beta \right) < \eta$.
			To that aim, recall first that by Markov's inequality \cite[Pg. 114]{Papoulis91}, it follows that $\Pr\left(\right\|U^{(k)}\left\|^2 > \beta \right) < \frac{1}{\beta}\E\left\{\left\|U^{(k)}\right\|^2 \right\}$. For the asynchronously sampled source process, we note that $ \Sigseps[i]\triangleq \E\left\{\left(\Seps[i]\right)^2\right\}\in [0, \mathop{\max}\limits_{0\leq t\leq\Tc}\Csc(t)] $. By the independence of $\Wepsk$ and $\Shepsk$, and by the fact that their mean is zero, we have, from \eqref{eqn:testchannelrdf_asynchron} that $\E\left\{\left(\Seps[i]\right)^2\right\}=\E\left\{\left(\Sheps[i]\right)^2\right\}+\E\left\{\left(\Weps[i]\right)^2\right\}\leq \mathop{\max}\limits_{0\leq t\leq\Tc}\Csc(t)$; Hence $\E\left\{\left(\Sheps[i]\right)^2\right\}\leq \mathop{\max}\limits_{0\leq t\leq\Tc}\Csc(t)$, and $\E\left\{\left(\Weps[i]\right)^2\right\}\leq \mathop{\max}\limits_{0\leq t\leq\Tc}\Csc(t)$. This further implies that $\E\left\{\left\| U^{(k)}\right\|^2 \right\}=\E\left\{\left\| \left[\big( \Sepsk\big) ^T , \big( \Shepsk\big) ^T \right]^T\right\|^2 \right\}\leq 2\cdot k\cdot \mathop{\max}\limits_{0\leq t\leq\Tc}\Csc(t)$ ; therefore for each $\beta > \frac{1}{\eta}\E\left\{\left\|U^{(k)}\right\|^2 \right\}$ we have that $\Pr\left(\left\|U^{(k)}\right\|^2 > \beta \right) < \eta$, and thus $U^{(k)}$ is separable.

    		    By the assumption in this lemma it follows that $\forall \eta > 0$ there exists $n_0(\eta) >0$ such that for all $n > n_0(\eta)$ we have that $\forall \myVec{w}^{(k)} \in \mR^{k}$, $\big| \pdf{\Wnk}\left( \myVec{w}^{(k)}  \right)  -  \pdf{\Wepsk}\left(w ^{(k)}\right) \big| < \eta$, for all sufficiently large $k \in \mySet{N}$. Consequently, for all $\left( \left(\myVec{s}^{(k)}\right)^T ,  \left(\hat{\myVec{s}}^{(k)}\right)^T\right)^T  \in \mySet{R}^{2k}$, $n > n_0(\eta)$ and a sufficiently large $k \in \mySet{N}$, it follows from \eqref{eqn:conditional_pdf} that
			\begin{align}
			\hspace{-1cm}\left|\pdf{\Snk | \Shnk }\left(\myVec{s}^{(k)} \big|  \hat{\myVec{s}}^{(k)} \right) - \pdf{\Sepsk | \Shepsk}\left(\myVec{s}^{(k)} \big|  \hat{\myVec{s}}^{(k)} \right)  \right|  
			&=  \left|	\pdf{\Wnk}\left( \myVec{s}^{(k)} - \hat{\myVec{s}}^{(k)} \right) -  \pdf{\Wepsk}\left(\myVec{s}^{(k)} -  \hat{\myVec{s}}^{(k)} \right)\right| <\eta.  
			\label{eqn:ContFproof3}
			\end{align} 
			Following the continuity of $\pdf{\Snk | \Shnk }\left( s ^{(k)}\big|  \hat{s} ^{(k)}\right)$ and of $\pdf{\Snk}(\myVec{s}^{(k)})$, $\fkn\left(\myVec{s}^{(k)}, \hat{\myVec{s}}^{(k)} \right)$ is also continuous \cite[Thm. 4.9]{rudin1976principles}\myFtn{\cite[Thm. 4.9]{rudin1976principles}: Let $f$ and $g$ be complex continuous functions on a metric space $X$. Then $f+g$, $fg$ and $f/g$ are continuous on $X$. In the last case, we must assume that $g(x)\ne 0$, for all $x\in X$}; hence, when $\mathop{\lim}\limits_{n \rightarrow \infty}\big( \myVec{s}_n^{(k)},\hat{\myVec{s}}_n^{(k)}\big)  = \big(\myVec{s}^{(k)}, \hat{\myVec{s}}^{(k)}\big) $, then $\mathop{\lim}\limits_{n \rightarrow \infty}  \fkn\left(\myVec{s}_n^{(k)},\hat{\myVec{s}}_n^{(k)}\right) = \fkeps\left(\myVec{s}^{(k)}, \hat{\myVec{s}}^{(k)}\right) $.
			This satisfies condition \ref{itm:Conv} for the extended \ac{cmt}; Therefore, by the extended \ac{cmt}, we have that $\fkn\left(\Snk, \Shnk\right) \ConvDist{n\rightarrow\infty} \fkeps\left(\Sepsk, \Shepsk \right)$. Since the \acp{rv} $\zkn{k}{n}'\left( \cdf{\Sn,\Shn}\opt\right)$ and $\zkeps{k,\eps}'\left( \cdf{\Seps, \Sheps}\right)$, defined in \eqref{eqn:zkdefs}, are also continuous mappings of $\fkn\left(\Snk, \Shnk\right)$ and of $\fkeps\left(\Sepsk, \Shepsk \right)$, respectively, it follows from  the \ac{cmt} \cite[Thm. 7.7]{kosorok2008introduction} that $\zkn{k}{n}'\left( \cdf{\Sn,\Shn}\opt\right) \ConvDist{n \rightarrow \infty} \zkeps{k,\eps}'\left( \cdf{\Seps, \Sheps}\right)$.
			
			Finally, to prove that the convergence $\zkn{k}{n}'\left( \cdf{\Sn,\Shn}\opt\right) \ConvDist{n \rightarrow \infty} \zkeps{k,\eps}'\left( \cdf{\Seps, \Sheps}\right)$ is uniform in $k\in \mN$, we note that as $\Shnk$ and $\Shepsk$ have independent entries, and the backward channels \eqref{eqn:testchannelrdf_synch} and \eqref{eqn:testchannelrdf_asynchron} are memoryless. Hence, it follows from the proof of \cite[Lemma B.2]{shlezinger2019capacity}, that the characteristic function of the \ac{rv} $k\cdot\zkn{k}{n}'\left( \cdf{\Sn,\Shn}\opt\right)$ which is denoted by $\Charac_{k\cdot\zkn{k}{n}}(\alpha) \triangleq \E \left\{e^{j\cdot\alpha\cdot k\cdot\zkn{k}{n}}\right\}$ converges to the characteristic function of $k\cdot \zkeps{k,\eps}'\left( \cdf{\Seps, \Sheps}\right)$, denoted by $\Charac_{k\cdot\zkeps{k,\eps}}(\alpha)$, uniformly over $k\in \mN$. Thus, for all sufficiently small $\eta>0$, $\exists k_0\in \mN, n_0(\eta, k_0) \in \mySet{N} $ such that $\forall n > n_0(\eta, k_0)$, and $\forall k>k_0$ 
			\begin{equation}
			    \big|\Charac_{k\cdot\zkn{k}{n}}(\alpha) - \Charac_{k\cdot\zkeps{k,\eps}}(\alpha) \big| < \eta, \quad \forall \alpha \cdot  \in \mySet{R}. 
			\end{equation}  
			Hence, following Lévy's convergence theorem \cite[Thm. 18.1]{williams1991probability} \myFtn{\cite[Thm. 18.1]{williams1991probability}: Let $(F_n)$ be a sequence of density functions and let $\phi_n$ denote the characteristic function of $F_n$. Suppose that $g(\theta)\coloneqq \lim \phi_n(\theta)$ exists for all $\theta \in \mR$, and that $g(\cdot)$ is continuous at $0$. Then $g=\phi F$ for some distribution function $F$, and $F_n \ConvDist{n\rightarrow \infty} F$.} we conclude that $k\cdot\zkn{k}{n}'\left( \cdf{\Sn,\Shn}\opt\right) \ConvDist{n \rightarrow \infty} k\cdot\zkeps{k,\eps}'\left( \cdf{\Seps, \Sheps}\right)$ and that this convergence is uniform for sufficiently large $k$. Finally, since the \acp{cdf} of $k\cdot\zkn{k}{n}'\left( \cdf{\Sn,\Shn}\opt\right)$ and $k\cdot\zkeps{k,\eps}'\left( \cdf{\Seps, \Sheps}\right)$ obtained at $\alpha \in \mR$ are equivalent to the \acp{cdf} of $\zkn{k}{n}'\left( \cdf{\Sn,\Shn}\opt\right)$ and $\zkeps{k,\eps}'\left( \cdf{\Seps, \Sheps}\right)$ obtained at $\frac{\alpha}{k} \in \mR$ respectively, we can conclude that $\zkn{k}{n}'\left( \cdf{\Sn,\Shn}\opt\right) \ConvDist{n \rightarrow \infty} \zkeps{k,\eps}'\left( \cdf{\Seps, \Sheps}\right)$, uniformly in $k\in \mN$.
		\end{proof}
		
		The following convergence lemma \ref{lem:AsyncZk2} corresponds to \cite[Lemma. B.3]{shlezinger2019capacity},
			\begin{lemma}
			\label{lem:AsyncZk2}
			  Let $n\in \mN$ be given. Every subsequence of $\left\{\zkn{k}{n}'\left( \cdf{\Shn, \Sn}\opt\right)\right\}_{k \in \mySet{N}}$, indexed by $k_l$, converges in distribution, in the limit as $l \rightarrow \infty$, to a finite deterministic  scalar.
%
		\end{lemma}
		\begin{proof}
			Recall that the \acp{rv} $\zkn{k}{n}'\left( \cdf{\Shn, \Sn}\opt\right)$  represent the mutual information density rate between $k$ samples of the source process $S_n[i]$ and the corresponding samples of its reproduction process $\hat{S}_n[i]$, where these processes are jointly distributed via the Gaussian distribution measure $ \cdf{\Shn, \Sn}\opt$. Further, recall that the relationship between the source signal and the reproduction process which achieves the \ac{rdf} can be  described via the backward channel in \eqref{eqn:testchannelrdf_synch} for a Gaussian source. 
			The channel \eqref{eqn:testchannelrdf_synch} is a memoryless additive \ac{wscs} Gaussian noise channel with period $p_n$, thus, by \cite{shlezinger2016capacity}, it can be equivalently represented as a $p_n \times 1$ multivariate memoryless additive stationary Gaussian noise channel, which is an {\em information stable} channel \cite[Sec. 1.5]{dobrushin1959general}\footnote{Information stable channels can
be described as having the property that the input
that maximizes mutual information and its corresponding
output behave ergodically \cite{verdu1994general}. Also, the information stability was further defined in \cite[Sec. IV]{zeng2006information} by applying the fact that ergodic theory is consequential to the law of large numbers. \cite[Eq. (3.9.2)]{han2003information}: A
general source $V =\left \{V^n\right\}_{n=1}^\infty$ is said to be information-stable if $\frac{\frac{1}{n}\log \frac{1}{P_{v^n}\left(V^n\right)}}{H_n\left(V^n\right)} \rightarrow 1$. where $H_n\left(V^n\right)=\frac{1}{n}H\left(V^n\right)$ and $H\left(V^n\right)$ stands for the entropy of $V^n$.}.
			For such channels in which the source and its reproduction  obey the \ac{rdf}-achieving joint distribution $\cdf{\Sn,\Shn}\opt$, the mutual information density rate converges as $k$ increases almost surely to the finite and deterministic  mutual information rate \cite[Thm. 5.9.1]{han2003information} \footnote{\cite[Thm. 5.9.1]{han2003information} holds for a subadditive distortion measure \cite[Eqn. (5.9.2)]{han2003information}; The \ac{mse} distortion measure, which was used in this research, is additive (and thus also subadditive).}. Since almost sure convergence implies convergence in distribution \cite[Lemma 7.21]{kosorok2008introduction}, this proves the lemma. 
		\end{proof}

		
\subsection{Showing that \mathinhead{R_\eps (D) = \mathop{\lim \sup}\limits_{n\rightarrow\infty} R_n (D)}{RntoReps}}
\label{app:proof2c}
This section completes the proof to Theorem \ref{Thm:rate_WSACS}. We note from \eqref{eqn:general_rdf} that the \ac{rdf} for the source process $S_n[i]$ (for fixed length coding and \ac{mse} distortion measure) is given by:

\begin{equation}
\label{eqn:inf_spec_Rn(D)}
    R_n(D)=\mathop{\inf}\limits_{ \cdf{\Shn, \Sn}:\bar{d}_S\left(\cdf{\Shn, \Sn}\right)\leq D}\left\{\plimsup \zkn{k}{n}'\left(\cdf{\Shn, \Sn}\opt\right)\right\},
\end{equation}
 where $\bar{d}_S\left(\cdf{\Shn, \Sn}\right)=\mathop{\lim \sup}\limits_{k\rightarrow\infty}\frac{1}{k}\E\big\{ \big\|\myVec{S}_{n}^{(k)}-\hat{\myVec{S}}_{n}^{(k)}\big\|^2\big\}$.  

We now state the following lemma characterizing the asymptotic statistics of the optimal reconstruction $\Shnk$ process and the respective noise process $\Wnk$ used in the backward channel relationship \eqref{eqn:testchannelrdf_synch}:
\begin{lemma}
\label{lem:tightness}
Consider the \ac{rdf}-achieving  distribution with distortion $D$ for compression of a vector Gaussian source process $\Snk$ characterized by the backward channel \eqref{eqn:testchannelrdf_synch}. Then, there exists a subsequence in the index $n \in \mN$ denoted $n_1 < n_2 < \ldots$, such that for the \ac{rdf}-achieving distribution, the sequences of reproduction vectors $\{\hat{\myVec{S}}_{n_l}^{(k)}\}_{l\in\mN}$ and backward channel noise vectors $\{\myVec{W}_{n_l}^{(k)}\}_{l\in\mN}$ satisfy that $\mathop{\lim}\limits_{l \rightarrow \infty}\pdf{\hat{\myVec{S}}_{n_l}^{(k)}}\left( \hat{\myVec{s}}^{(k)} \right)= \pdf{\Shepsk}\left( \hat{\myVec{s}}^{(k)} \right)$ uniformly in $\myVec{\hat{s}}^{(k)} \in \mR^{k}$ and uniformly with respect to $k \in \mySet{N}$, as well as $\mathop{\lim}\limits_{l \rightarrow \infty}\pdf{\myVec{W}_{n_l}^{(k)}}\left( \myVec{w}^{(k)} \right)= \pdf{\Wepsk}\left( \myVec{w}^{(k)} \right)$ uniformly in $\myVec{w}^{(k)} \in \mR^{k}$ and uniformly with respect to $k \in \mySet{N}$, where $\pdf{\Shepsk}\left( \hat{\myVec{s}}^{(k)} \right)$ and $\pdf{\Wepsk}\left( \myVec{w}^{(k)} \right)$ are Gaussian \acp{pdf}. 
\end{lemma}
\begin{proof}

\color{black}

Recall from the analysis of the \ac{rdf} for \ac{wscs} processes that for each $n \in \mN$, the marginal distributions of the \ac{rdf}-achieving reproduction process $\Shn[i]$ and the backward channel noise $W_n[i]$ is Gaussian, memoryless, zero-mean, and with variances $\sigma_{\hat{S}_n}^2[i] \triangleq \E\left\{\big(\hat{S}_n[i] \big)^2  \right\}$ and 
\begin{equation}
    \label{eqn:variance_source_sig}
    \E\left\{\big(W_n[i] \big)^2  \right\} = \sigma_{{S}_n}^2[i] - \sigma_{\hat{S}_n}^2[i],
\end{equation}
 respectively. 
Consequently, the sequences of reproduction vectors $\{\hat{\myVec{S}}_{n}^{(k)}\}_{n\in\mN}$ and backward channel noise vectors $\{\myVec{W}_{n}^{(k)}\}_{n\in\mN}$ are zero-mean Gaussian with independent entries for each $k \in \mN$.
Since $\sigma_{{S}_n}^2[i]\leq \mathop{\max}\limits_{t\in \mR}\sigma_{S_c}^2(t)$, then, from \eqref{eqn:variance_source_sig}, it follows that $\Sigshn[i]$ is also bounded in the interval  $[0,\mathop{\max}\limits_{t \in \mR}\sigma_{S_c}^2(t)]$ for all $n \in \mN$. Therefore, by Bolzano-Weierstrass theorem \cite[Thm. 2.42]{rudin1976principles} \myFtn{Every bounded infinite subset of $\mR^k$ has a limit point in $\mR^k$.}, $\Sigshn[i]$ has a convergent subsequence, and we let $n_1 < n_2 < \ldots$ denote the indexes of this convergent subsequence and let the limit of the subsequence be denoted by $\sigma_{\Sheps}^2[i]$. From the \ac{cmt}, as applied in the proof of \cite[Lemma B.1]{shlezinger2019capacity}, the convergence $\sigma_{\hat{S}_{n_l}}^2[i]\mathop{\longrightarrow}\limits_{l\rightarrow \infty}\sigma_{\Sheps}^2[i]$ for each $i\in \mySet{N}$ implies that the subsequence of \acp{pdf} $\pdf{\hat{\myVec{S}}_{n_l}^{(k)}}\left(\hat{\myVec{s}}^{(k)}\right)$ corresponding to   the memoryless Gaussian random vectors $\{\hat{\myVec{S}}_{n_l}^{(k)}\}_{l \in \mySet{N}}$ converges as $l \rightarrow \infty$ to a Gaussian \ac{pdf} which we denote by $\pdf{\Shepsk}\left(\hat{\myVec{s}}^{(k)}\right)$, and the convergence of $\pdf{\hat{\myVec{S}}_{n_l}^{(k)}}\left(\hat{\myVec{s}}^{(k)}\right)$ is uniform in $\myVec{s}^{(k)}$ for any fixed $k \in \mySet{N}$. 
By Remark \ref{stationary_noise}, it holds that $\Wn[i]$ is a memoryless stationary process with variance  $\E\left\{\left(\Wn[i]\right)^2\right\}=D$ and by Eq. \eqref{eqn:variance_source_sig}, $\Sigshn[i]=\Sigsn[i]-D$. Hence by Assumption \ref{Assumption_unif_conv}  and by the proof of \cite[Lemma B.1]{shlezinger2019capacity}, it follows that for a fixed $\eta > 0$ and $k_0 \in \mySet{N}$, $\exists n_0 (\eta, k_0)$ such that for all $n > n_0(\eta, k_0)$ and for all sufficiently large $k$, it holds that $\big|\pdf{\hat{\myVec{S}}_{n_l}^{(k)}}\big(\hat{\myVec{s}}^{(k)}\big) - \pdf{\Shepsk}\big(\hat{\myVec{s}}^{(k)}\big)\big| < \eta$ for every $\hat{\myVec{s}}^{(k)} \in \mySet{R}^k$. Since  $n_0(\eta, k_0)$ does not depend on $k$ (only on the fixed $k_0$), this implies that the convergence is uniform with respect to $k \in \mySet{N}$.

The fact that $W_n[i]$ is a zero-mean stationary Gaussian process with variance $D$ for each $n \in \mN$, implies that the sequence of \acp{pdf} $\pdf{\Wnk}\left(\myVec{w}^{(k)}\right)$ converges as $n \rightarrow \infty$ to a Gaussian \ac{pdf} which we denote by $\pdf{\myVec{W}^{(k)}}\left(\myVec{w}^{(k)}\right)$, hence its subsequence with indices $n_1 < n_2 < \ldots$ also converges to $\pdf{\myVec{W}^{(k)}}\left(\myVec{w}^{(k)}\right)$. Since $D > \frac{1}{2\pi}$ by Assumption \ref{Assumption_unif_conv} combined with the proof of \cite[Lemma B.1]{shlezinger2019capacity} it follows that this convergence is  uniform in $\myVec{w}^{(k)}$ and in $k\in \mN$ to $\pdf{\Wepsk}\left(\myVec{w}^{(k)}\right)$.

Following the proof of Corollary \ref{lem: convDist}, it holds that the subsequences of the memoryless Gaussian random vectors $\left\{\hat{\myVec{S}}_{n_l}^{(k)}\right\}$ and $\left\{\myVec{W}_{n_l}^{(k)}\right\}$ converge in distribution as $l\rightarrow\infty$ to a Gaussian distribution, and the convergence is uniform in  $k \in \mN$ for any fixed $k\in \mN$. Hence, as shown in Lemma \ref{lem:AsyncZk} the joint distribution 
$\left[\big( \myVec{S}_{n_l}^{(k)}\big) ^T , \big( \hat{\myVec{S}}_{n_l}^{(k)}\big) ^T \right]^T\ConvDist{n \rightarrow \infty}\left[\big( \Sepsk\big) ^T , \big( \Shepsk\big) ^T \right]^T$, and the limit distribution is jointly Gaussian.
\end{proof}

\begin{lemma}
\label{lem:rdfUpBound}
The \ac{rdf} of $\{S_\epsilon[i]\}$ satisfies $R_\eps(D)\leq \mathop{\lim \sup}\limits_{n\rightarrow \infty} R_n(D)$, and the rate $\mathop{\lim \sup}\limits_{n\rightarrow \infty} R_n(D)$ is achievable for the source $\{S_\epsilon[i]\}$ with distortion $D$ when the reproduction process which obeys a Gaussian distribution.
\end{lemma}
\begin{proof}
According to Lemma \ref{lem:tightness}, we note that the sequence of joint distributions $\{\cdf{\Sn, \Shn}\opt \}_{n \in \mySet{N}}$  has a convergent subsequence, i.e., there exists a set of indexes $n_1 < n_2 < \ldots$ such that the sequence of distributions with independent entries $\{\cdf{\Snl,\hat{S}_{n_l}}\opt \}_{l \in \mySet{N}}$ converges in the limit $l \rightarrow \infty$ to a joint  Gaussian distribution $\cdf{\Seps,\Sheps}'$ and the convergence is uniform in $k \in \mN$. Hence, this satisfies the condition of Lemma \ref{lem:AsyncZk};
 This implies that $\zkn{k}{n_l}'\left(\cdf{\Snl,\Shnl}\opt\right)\ConvDist{l\rightarrow\infty}\zk{k}'\left(\cdf{\Seps,\Sheps}'\right)$ uniformly in $k\in \mN $. Also, by Lemma \ref{lem:AsyncZk2} every subsequence of $\big\{\zkn{k}{n_l}'\big( \cdf{\Snl,\Shnl}\opt\big)\big\}_{l \in \mySet{N}}$ converges in distribution to a finite deterministic scalar as $k \rightarrow \infty$. Therefore, by Theorem~\ref{thm:plim} it holds that
\begin{align}
\label{eqn:RDFuppBound1}
    	\mathop{\lim}\limits_{l \rightarrow \infty}\left( {\rm p-}\mathop{\lim \sup}\limits_{k \rightarrow \infty} \zkn{k}{n_l}'\left( \cdf{\Snl, \Shnl}\opt\right) \right)
			&= {\rm p-}\mathop{\lim \sup}\limits_{k \rightarrow \infty} \zkeps{k,\eps}'\left( \cdf{\Seps,\Sheps}'\right) \notag \\
			& \ge \mathop{\inf} \limits_{\cdf{\Seps,\Sheps}}\left\{  {\rm p-}\mathop{\lim \sup}\limits_{k \rightarrow \infty} \zkeps{k,\eps}'\left( \cdf{\Seps,\Sheps}\right)\right\} = R_\myEps(D).
\end{align}

From  \eqref{eqn:general_rdf} we have that $R_n(D)=\plimsup \zkn{k}{n}'\left(\cdf{\Sn,\Shn}\opt\right)$, then from \eqref{eqn:RDFuppBound1}, it follows that
\begin{equation}
    \label{eqn:rdfuppBound2}
    R_\eps(D)\leq \mathop{\lim}\limits_{l\rightarrow\infty}R_{n_l}(D)\stackrel{(a)}{\le} \mathop{\lim \sup}\limits_{n\rightarrow\infty} R_n(D),
\end{equation}
where $(a)$ follows since, by \cite[Def. 3.16]{rudin1976principles}, the limit of every subsequence is not greater than the limit superior. Noting that $\cdf{\Seps,\Sheps}'$ is Gaussian by Lemma \ref{lem:tightness} concludes the proof. 
\end{proof}

\begin{lemma}
\label{lem:rdflowBound} 
The \ac{rdf} of $\{S_\epsilon[i]\}$ satisfies 
$R_\myEps(D)\ge \mathop{\lim \sup}\limits_{n\rightarrow\infty}R_n(D)$.
\end{lemma}
\begin{proof}
To prove this lemma, we first show that for a joint distribution $\cdf{\Seps,\Sheps}$ which achieves a rate-distortion pair $(R_{\eps}, D)$ it holds that $R_{\eps} \ge \E \{\zkeps{k,\eps}'(\cdf{\Seps,\Sheps}')\}$: 
Recall that $(R_\myEps, D)$ is an achievable rate-distortion pair for the source $\{\Seps[i]\}$, namely, there exists a sequence of codes $\{\mySet{C}_l\}$ whose rate-distortion  approach $(R_\myEps, D)$ when applied to  $\{S_\epsilon[i]\}$, 
This implies that for any $\eta >0$ there exists $l_0(\eta)$ such that $\forall l > l_0(\eta)$ it holds that $\mySet{C}_l$ has a code rate $R_l = \frac{1}{l}\log_2 M_l$ satisfying $R_l \le R_{\epsilon} + \eta$ by \eqref{eqn:rate_bound}. 
Recalling Def. \ref{Def:source_coding_scheme}, the source code maps $\myVec{S}_\epsilon^{(l)}$ into a discrete index $J_l \in \{1,2,\ldots, M_l\}$, which is in turn mapped into  $\hat{\myVec{S}}_\epsilon^{(l)}$, i.e., ${\myVec{S}}_\epsilon^{(l)} \mapsto J_l \mapsto \hat{\myVec{S}}_\epsilon^{(l)}$ form a Markov chain. Since $J_l$ is a discrete random variable taking values in $\{1,2,\ldots, M_l\}$, it holds that 
\begin{align}
    \log_2 M_l &\ge H(J_l) \notag \\
    &\stackrel{(a)}{\ge} I({\myVec{S}}_\epsilon^{(l)} ; J_l) \notag \\ &\stackrel{(b)}{\ge} I({\myVec{S}}_\epsilon^{(l)} ; \hat{\myVec{S}}_\epsilon^{(l)}),
    \label{eqn:Bound1Rate}
\end{align}
where $(a)$ follows since $ I({\myVec{S}}_\epsilon^{(l)} ; J_l) = H( J_l) - H( J_l|{\myVec{S}}_\epsilon^{(l)})$ which is not larger than $H( J_l)$ as $J_l$ takes discrete values; while $(b)$ follows from the data processing inequality \cite[Ch. 2.8]{cover2006elements}. 
Now, \eqref{eqn:Bound1Rate} implies that for each $l > l_0(\eta)$, the reproduction obtained using the code $\mySet{C}_l$ satisfies $\frac{1}{l}I({\myVec{S}}_\epsilon^{(l)} ; \hat{\myVec{S}}_\epsilon^{(l)}) \le \frac{1}{l}\log M_l \le R_{\epsilon} + \eta$. Since for every arbitrarily small $\eta \rightarrow 0$, this inequality holds for all $l > l_0(\eta)$, i.e., for all sufficiently large $l$, it  follows that $R_\epsilon \ge \mathop{\lim \sup}\limits_{k\rightarrow\infty} \frac{1}{l} I({\myVec{S}}_\epsilon^{(l)} ; \hat{\myVec{S}}_\epsilon^{(l)})$.
Hence, replacing the blocklength symbol from $l$ to $k$, as $\frac{1}{k} I(\Sepsk, \Shepsk)=\E \{\zkeps{k,\eps}'(\cdf{\Seps,\Sheps}')\}$\cite[Eqn. (2.3)]{cover2006elements}, we conclude that
\begin{equation}
\label{eqn: rdfbound2}
    R_\eps(D)\geq \mathop{\lim \sup}\limits_{k\rightarrow\infty} \E \{\zkeps{k,\eps}'(\cdf{\Seps,\Sheps}')\}.
\end{equation}

Next, we consider $\mathop{\lim \sup}\limits_{k\rightarrow\infty} \E \{\zkeps{k_l,\eps}'(\cdf{\Seps,\Sheps}')\}$: Let $\zkeps{k_l,\eps}' \left(\cdf{\Seps,\Sheps}' \right) $ be a subsequence of $\E \left\{\zkeps{k,\eps}'(\cdf{\Seps,\Sheps}')\right\}$ with the indexes $k_1 < k_2 < \ldots$ such that its limit equals the limit superior. i.e., $ \mathop{\lim}\limits_{l \rightarrow \infty} \E \left\{\zkeps{k_l,\eps}' \left(\cdf{\Seps,\Sheps}' \right)  \right\} = \mathop{\lim \sup}\limits_{k \rightarrow \infty} \E \left\{\zkeps{k,\eps}' \left(\cdf{\Seps,\Sheps}' \right)  \right\}$. Since by Lemma \ref{lem:AsyncZk}, the sequence of non-negative \acp{rv} $\left\{\zkn{k_l}{n}' \left(\cdf{\Sn,\Shn}\opt \right)  \right\}_{n \in \mySet{N}}$ convergences in distribution
to $\zkeps{k_l,\eps}' \left(\cdf{\Seps,\Sheps}' \right)$ as $n \rightarrow \infty$ uniformly in $k \in \mySet{N}$, it follows from\myFtn{\cite[Thm. 3.5]{billingsley1999convergence} states that if $X_n$ are uniformly integrable and $X_n \ConvDist{n \rightarrow\infty} X$ then $\E\{X_n\} \mathop{\longrightarrow}\limits_{n\rightarrow\infty} \E\{X\}$.} \cite[Thm. 3.5]{billingsley1999convergence} that $\E \left\{\zkeps{k_l,\eps}' \left(\cdf{\Seps,\Sheps}' \right)  \right\} =  \mathop{\lim}\limits_{n \rightarrow \infty} \E \left\{\zkn{k_l}{n}' \left(\cdf{\Sn,\Shn}\opt \right)  \right\}$. Also, we define a family of distributions $\mySet{F}(D)$ such that ${\mySet{F}(D)=\{\cdf{S,\hat{S}}:\mathsf{D}\left(\cdf{S, \hat{S}}\right)\leq D\}}$. 
Consequently, Eq. \eqref{eqn: rdfbound2} can now be written as:
\begin{align}
\label{eqn:rdfbound3}
    R_\eps(D)\geq \mathop{\lim \sup}\limits_{k \rightarrow \infty} \E \left\{\zkeps{k,\eps}' \left(\cdf{\Seps,\Sheps}' \right)  \right\}
    &=\mathop{\lim}\limits_{l \rightarrow \infty}\mathop{\lim}\limits_{n \rightarrow \infty} \E \left\{\zkn{k_l}{n}' \left(\cdf{\Sn,\Shn}\opt \right)  \right\}  \notag
    \\
    &\stackrel{(a)}{=}\mathop{\lim}\limits_{n \rightarrow \infty}\mathop{\lim}\limits_{l \rightarrow \infty}\E \left\{\zkn{k_l}{n}' \left(\cdf{\Sn,\Shn}\opt \right)  \right\} \notag
    \\
    &\stackrel{(b)}{=} \mathop{\lim \sup}\limits_{n \rightarrow \infty}\mathop{\lim}\limits_{l \rightarrow \infty}\E \left\{\zkn{k_l}{n}' \left(\cdf{\Sn,\Shn}\opt \right)  \right\} \notag
    \\ 
   & \ge \mathop{\lim \sup}\limits_{n \rightarrow \infty}\mathop{\lim}\limits_{l \rightarrow \infty}  \mathop{\inf}\limits_{\cdf{S,\hat{S}} \in \mySet{F}(D)} \E \left\{\zkn{k_l}{n}' \left(\cdf{S,\hat{S}}\right)  \right\}\notag
   \\
	&\stackrel{(c)}{=}  \mathop{\lim \sup}\limits_{n \rightarrow \infty} \mathop{\lim}\limits_{l \rightarrow \infty}  \mathop{\inf}\limits_{ \cdf{S,\hat{S}} \in \mySet{F}(D)}\frac{1}{k_l} I\left(\hat{\myVec{S}}_{n}^{(k_l)} ; \myVec{S}_n^{(k_l)} \right),
\end{align}
	where $(a)$ follows since the convergence  $\zkn{k_l}{n}' \left(\cdf{\Sn,\Shn}\opt \right)\ConvDist{n\rightarrow \infty}\zkeps{k_l,\eps}' \left(\cdf{\Seps,\hat{S}}'\right)$ is uniform with respect to $k_l$, thus the limits are interchangeable \cite[Thm. 7.11]{rudin1976principles}\myFtn{Rudin: Thm. 7.11: Suppose $f_n \rightarrow f$ uniformly in a  set $E$ in a metric space. Let $x$ be a limit point of $E$.... $\mathop{\lim}\limits_{t \rightarrow x}\mathop{\lim}\limits_{n \rightarrow \infty} f_n(t)=\mathop{\lim}\limits_{n \rightarrow \infty}\mathop{\lim}\limits_{t \rightarrow x}f_n(t)$}; $(b)$ follows since the limit of the subsequence $\E \left\{\zkn{k_l}{n}' \left(\cdf{\Sn,\Shn}\opt \right)  \right\}$ exists in the index $n$, and is therefore equivalent to the limit superior, $\mathop{\lim \sup}\limits_{n \rightarrow \infty} \E \left\{\zkn{k_l}{n}' \left(\cdf{\Sn,\Shn}\opt \right)  \right\}$ \cite[Page 57]{rudin1976principles};  and $(c)$ holds since mutual information is the expected value of the mutual information density rate \cite[Eqn. (2.30)]{cover2006elements}. Finally, we recall that in the proof of Lemma \ref{lem:AsyncZk2} it was established that the backward channel for the \ac{rdf} at the distortion constraint $D$, defined in \eqref{eqn:testchannelrdf_synch}, is information stable, hence for such backward channels, we have from \cite[Thm. 1]{4215152} that the minimum rate is defined as $R_n(D)=\mathop{\lim}\limits_{k \rightarrow \infty}\mathop{\inf}\limits_{{ \cdf{S,\hat{S}} \in \mySet{F}(D)}} \frac{1}{k}I\left(\Shepsk ; \Snk\right)$ and the limit exists; Hence, $\mathop{\lim}\limits_{k \rightarrow \infty}\mathop{\inf}\limits_{{ \cdf{S,\hat{S}} \in \mySet{F}(D)}} \frac{1}{k}I\left(\Shepsk ; \Snk\right)=\mathop{\lim}\limits_{l \rightarrow \infty}\mathop{\inf}\limits_{{ \cdf{S,\hat{S}} \in \mySet{F}(D)}}\frac{1}{k_l} I\left(\hat{\myVec{S}}^{(k_l)} ; \myVec{S}_n^{(k_l)} \right)$ in the index $k$. Substituting this into equation \eqref{eqn:rdfbound3} yields the result:
	
	\begin{equation}
	    R_\eps(D)\ge \mathop{\lim \sup}\limits_{n\rightarrow\infty}R_n(D).
	\end{equation}
	This proves the lemma. 
	\end{proof}
Combining the lemmas \ref{lem:rdfUpBound} and \ref{lem:rdflowBound} proves that $R_\eps (D) = \mathop{\lim \sup}\limits_{n\rightarrow\infty} R_n (D)$ and the rate is achievable with Gaussian inputs, completing the proof of the theorem. 
\end{appendices}

\bibliographystyle{IEEEtran}
\bibliography{myreference.bib}

\begin{thebibliography}{10}
\providecommand{\url}[1]{#1}
\csname url@samestyle\endcsname
\providecommand{\newblock}{\relax}
\providecommand{\bibinfo}[2]{#2}
\providecommand{\BIBentrySTDinterwordspacing}{\spaceskip=0pt\relax}
\providecommand{\BIBentryALTinterwordstretchfactor}{4}
\providecommand{\BIBentryALTinterwordspacing}{\spaceskip=\fontdimen2\font plus
\BIBentryALTinterwordstretchfactor\fontdimen3\font minus
  \fontdimen4\font\relax}
\providecommand{\BIBforeignlanguage}[2]{{%
\expandafter\ifx\csname l@#1\endcsname\relax
\typeout{** WARNING: IEEEtran.bst: No hyphenation pattern has been}%
\typeout{** loaded for the language `#1'. Using the pattern for}%
\typeout{** the default language instead.}%
\else
\language=\csname l@#1\endcsname
\fi
#2}}
\providecommand{\BIBdecl}{\relax}
\BIBdecl

\bibitem{gardner1987spectral}
W.~Gardner, W.~Brown, and C.-K. Chen, ``Spectral correlation of modulated
  signals: Part {II}-digital modulation,'' \emph{IEEE Transactions on
  Communications}, vol.~35, no.~6, pp. 595--601, 1987.

\bibitem{giannakis1998cyclostationary}
G.~B. Giannakis, ``Cyclostationary signal analysis,'' \emph{Digital Signal
  Processing Handbook}, pp. 17--1, 1998.

\bibitem{gardner2006cyclostationarity}
W.~A. Gardner, A.~Napolitano, and L.~Paura, ``Cyclostationarity: Half a century
  of research,'' \emph{Signal processing}, vol.~86, no.~4, pp. 639--697, 2006.

\bibitem{berger1998lossy}
T.~Berger and J.~D. Gibson, ``Lossy source coding,'' \emph{IEEE Transactions on
  Information Theory}, vol.~44, no.~6, pp. 2693--2723, 1998.

\bibitem{cover2006elements}
T.~M. Cover and J.~A. Thomas, \emph{Elements of Information Theory}.\hskip 1em
  plus 0.5em minus 0.4em\relax John Wiley \& Sons, 2006.

\bibitem{wolf1980source}
J.~K. Wolf, A.~D. Wyner, and J.~Ziv, ``Source coding for multiple
  descriptions,'' \emph{The Bell System Technical Journal}, vol.~59, no.~8, pp.
  1417--1426, 1980.

\bibitem{wyner1976rate}
A.~Wyner and J.~Ziv, ``The rate-distortion function for source coding with side
  information at the decoder,'' \emph{IEEE Transactions on Information Theory},
  vol.~22, no.~1, pp. 1--10, 1976.

\bibitem{oohama1997gaussian}
Y.~Oohama, ``{G}aussian multiterminal source coding,'' \emph{IEEE Transactions
  on Information Theory}, vol.~43, no.~6, pp. 1912--1923, 1997.

\bibitem{pandya2004lossy}
A.~Pandya, A.~Kansal, G.~Pottie, and M.~Srivastava, ``Lossy source coding of
  multiple {G}aussian sources: m-helper problem,'' in \emph{Proceedings of the
  Information Theory Workshop}.\hskip 1em plus 0.5em minus 0.4em\relax IEEE,
  Oct. 2004, pp. 34--38.

\bibitem{gallager1968information}
R.~G. Gallager, \emph{Information Theory and Reliable Communication}.\hskip 1em
  plus 0.5em minus 0.4em\relax Springer, 1968, vol. 588.

\bibitem{harrison2008generalized}
M.~T. Harrison, ``The generalized asymptotic equipartition property: Necessary
  and sufficient conditions,'' \emph{IEEE Transactions on Information Theory},
  vol.~54, no.~7, pp. 3211--3216, 2008.

\bibitem{kipnis2018distortion}
A.~Kipnis, A.~J. Goldsmith, and Y.~C. Eldar, ``The distortion rate function of
  cyclostationary {G}aussian processes,'' \emph{IEEE Transactions on
  Information Theory}, vol.~64, no.~5, pp. 3810--3824, 2018.

\bibitem{napolitano2016cyclostationarity}
A.~Napolitano, ``Cyclostationarity: New trends and applications,'' \emph{Signal
  Processing}, vol. 120, pp. 385--408, 2016.

\bibitem{han2003information}
T.~Han, \emph{Information-Spectrum Methods in Information Theory}.\hskip 1em
  plus 0.5em minus 0.4em\relax Springer, 2003, vol.~50.

\bibitem{verdu1994general}
S.~Verd{\'u} and T.~Han, ``A general formula for channel capacity,'' \emph{IEEE
  Transactions on Information Theory}, vol.~40, no.~4, pp. 1147--1157, 1994.

\bibitem{zeng2006information}
W.~Zeng, P.~Mitran, and A.~Kavcic, ``On the information stability of channels
  with timing errors,'' in \emph{Proceedings of the IEEE International
  Symposium on Information Theory (ISIT)}.\hskip 1em plus 0.5em minus
  0.4em\relax IEEE, July 2006, pp. 1885--1889.

\bibitem{shlezinger2019capacity}
N.~Shlezinger, E.~Abakasanga, R.~Dabora, and Y.~C. Eldar, ``The capacity of
  memoryless channels with sampled cyclostationary {G}aussian noise,''
  \emph{IEEE Transactions on Communications}, vol.~68, no.~1, pp. 106--121,
  2020.

\bibitem{shannon1998communication}
C.~E. Shannon, ``Communication in the presence of noise,'' \emph{Proceedings of
  the IEEE}, vol.~86, no.~2, pp. 447--457, 1998.

\bibitem{GUAN20131165}
Y.~Guan and K.~Wang, ``Translation properties of time scales and almost
  periodic functions,'' \emph{Mathematical and Computer Modelling}, vol.~57,
  no.~5, pp. 1165 -- 1174, 2013.

\bibitem{shlezinger2015capacity}
N.~Shlezinger and R.~Dabora, ``On the capacity of narrowband {PLC} channels,''
  \emph{IEEE Transactions on Communications}, vol.~63, no.~4, pp. 1191--1201,
  2015.

\bibitem{shlezinger2016capacity}
------, ``The capacity of discrete-time {G}aussian {MIMO} channels with
  periodic characteristics,'' in \emph{Proceedings of the IEEE International
  Symposium on Information Theory (ISIT)}, July 2016, pp. 1058--1062.

\bibitem{shlezinger2017secrecy}
N.~Shlezinger, D.~Zahavi, Y.~Murin, and R.~Dabora, ``The secrecy capacity of
  {G}aussian {MIMO} channels with finite memory,'' \emph{IEEE Transactions on
  Information Theory}, vol.~63, no.~3, pp. 1874--1897, 2017.

\bibitem{heath1999exploiting}
R.~W. Heath and G.~B. Giannakis, ``Exploiting input cyclostationarity for blind
  channel identification in ofdm systems,'' \emph{IEEE transactions on signal
  processing}, vol.~47, no.~3, pp. 848--856, 1999.

\bibitem{shaked2017joint}
R.~Shaked, N.~Shlezinger, and R.~Dabora, ``Joint estimation of carrier
  frequency offset and channel impulse response for linear periodic channels,''
  \emph{IEEE Transactions on Communications}, vol.~66, no.~1, pp. 302--319,
  2017.

\bibitem{shlezinger2014frequency}
N.~Shlezinger and R.~Dabora, ``Frequency-shift filtering for {OFDM} signal
  recovery in narrowband power line communications,'' \emph{IEEE Transactions
  on Communications}, vol.~62, no.~4, pp. 1283--1295, 2014.

\bibitem{el2011network}
A.~El~Gamal and Y.-H. Kim, \emph{Network Information Theory}.\hskip 1em plus
  0.5em minus 0.4em\relax Cambridge University Press, 2011.

\bibitem{wu2013optimal}
X.~Wu and L.-L. Xie, ``On the optimal compressions in the compress-and-forward
  relay schemes,'' \emph{IEEE Transactions on Information Theory}, vol.~59,
  no.~5, pp. 2613--2628, 2013.

\bibitem{zitkovic2015lecture}
G.~Zitkovic, ``Lecture notes on theory of probability,'' \emph{Lecture Notes
  for M385D (UT)}, 2015.

\bibitem{Papoulis91}
A.~Papoulis, \emph{Probability, Random Variables, and Stochastic
  Processes}.\hskip 1em plus 0.5em minus 0.4em\relax McGraw-Hill, 2002.

\bibitem{zamir2008achieving}
R.~Zamir, Y.~Kochman, and U.~Erez, ``Achieving the {Gaussian} rate--distortion
  function by prediction,'' \emph{IEEE Transactions on Information Theory},
  vol.~54, no.~7, pp. 3354--3364, 2008.

\bibitem{rudin1976principles}
W.~Rudin, \emph{Principles of Mathematical Analysis}, ser. International series
  in pure and applied mathematics.\hskip 1em plus 0.5em minus 0.4em\relax
  McGraw-Hill, 1976.

\bibitem{dixmier2013general}
J.~Dixmier, \emph{General Topology}.\hskip 1em plus 0.5em minus 0.4em\relax
  Springer Science \& Business Media, 2013.

\bibitem{stein2009real}
E.~M. Stein and R.~Shakarchi, \emph{Real Analysis: Measure Theory, Integration,
  and Hilbert spaces}.\hskip 1em plus 0.5em minus 0.4em\relax Princeton
  University Press, 2009.

\bibitem{kolmogorov1956shannon}
A.~Kolmogorov, ``On the {S}hannon theory of information transmission in the
  case of continuous signals,'' \emph{IRE Transactions on Information Theory},
  vol.~2, no.~4, pp. 102--108, 1956.

\bibitem{scheffe1947useful}
H.~Scheff{\'e}, ``A useful convergence theorem for probability distributions,''
  \emph{The Annals of Mathematical Statistics}, vol.~18, no.~3, pp. 434--438,
  1947.

\bibitem{bromiley2003products}
P.~Bromiley, ``Products and convolutions of {Gaussian} probability density
  functions,'' \emph{Tina-Vision Memo}, vol.~3, no.~4, p.~1, 2003.

\bibitem{kosorok2008introduction}
M.~R. Kosorok, \emph{Introduction to Empirical Processes and Semiparametric
  Inference.}\hskip 1em plus 0.5em minus 0.4em\relax Springer, 2008.

\bibitem{williams1991probability}
\BIBentryALTinterwordspacing
D.~Williams, \emph{Probability with Martingales}, ser. Cambridge mathematical
  textbooks.\hskip 1em plus 0.5em minus 0.4em\relax Cambridge University Press,
  1991. [Online]. Available:
  \url{https://books.google.co.il/books?id=e9saZ0YSi-AC}
\BIBentrySTDinterwordspacing

\bibitem{dobrushin1959general}
R.~L. Dobrushin, ``A general formulation of the fundamental theorem of
  {S}hannon in the theory of information,'' \emph{Uspekhi Matematicheskikh
  Nauk}, vol.~14, no.~6, pp. 3--104, 1959.

\bibitem{billingsley1999convergence}
P.~Billingsley, \emph{Convergence of probability measures}.\hskip 1em plus
  0.5em minus 0.4em\relax John Wiley \& Sons, 2013.

\bibitem{4215152}
R.~{Venkataramanan} and S.~S. {Pradhan}, ``Source coding with feed-forward:
  Rate-distortion theorems and error exponents for a general source,''
  \emph{IEEE Transactions on Information Theory}, vol.~53, no.~6, pp.
  2154--2179, June 2007.

\end{thebibliography}



\end{document}